\documentclass[11pt]{article}

\usepackage[a4paper]{geometry} 

\usepackage[utf8]{inputenc}
\usepackage{amssymb}
\usepackage{makeidx}
\usepackage[english]{babel}
\usepackage{graphicx}
\usepackage{amsfonts,amsmath,amssymb,amsthm}
\usepackage{oldgerm}
\usepackage{relsize}

\usepackage{mathrsfs}
\usepackage[active]{srcltx}
\usepackage{verbatim}
\usepackage[shortlabels]{enumitem} 
\usepackage[toc,page]{appendix}
\usepackage{aliascnt,bbm}
\usepackage{array}
\usepackage[colorlinks = true, citecolor = blue]{hyperref}
\usepackage[textwidth=3cm,textsize=footnotesize]{todonotes}
\usepackage{xargs}
\usepackage{cellspace}

\usepackage[Symbolsmallscale]{upgreek}
\usepackage{xstring}

\usepackage{array,multirow,makecell}

\usepackage{accents}
\usepackage{dsfont}
\usepackage{aliascnt}
\usepackage{cleveref}
\makeatletter
\newtheorem{theorem}{Theorem}
\crefname{theorem}{theorem}{Theorems}
\Crefname{Theorem}{Theorem}{Theorems}

\newaliascnt{lemma}{theorem}
\newtheorem{lemma}[lemma]{Lemma}
\aliascntresetthe{lemma}
\crefname{lemma}{lemma}{lemmas}
\Crefname{Lemma}{Lemma}{Lemmas}

\newaliascnt{corollary}{theorem}

\aliascntresetthe{corollary}
\crefname{corollary}{corollary}{corollaries}
\Crefname{Corollary}{Corollary}{Corollaries}

\newaliascnt{proposition}{theorem}
\newtheorem{proposition}[proposition]{Proposition}
\aliascntresetthe{proposition}
\crefname{proposition}{proposition}{propositions}
\Crefname{Proposition}{Proposition}{Propositions}

\newaliascnt{definition}{theorem}

\aliascntresetthe{definition}
\crefname{definition}{definition}{definitions}
\Crefname{Definition}{Definition}{Definitions}

\newaliascnt{definition-proposition}{theorem}

\aliascntresetthe{definition-proposition}
\crefname{definition-proposition}{definition-proposition}{definitions-propositions}
\Crefname{Definition-Proposition}{Definition-Proposition}{Definitions-Propositions}

\newaliascnt{remark}{theorem}
\newtheorem{remark}[remark]{Remark}
\aliascntresetthe{remark}
\crefname{remark}{remark}{remarks}
\Crefname{Remark}{Remark}{Remarks}

\crefname{example}{example}{examples}
\Crefname{Example}{Example}{Examples}

\crefname{figure}{figure}{figures}
\Crefname{Figure}{Figure}{Figures}

\newtheorem{assumption}{\textbf{H}\hspace{-4pt}}
\Crefname{assumption}{\textbf{H}\hspace{-4pt}}{\textbf{H}\hspace{-4pt}}
\crefname{assumption}{\textbf{H}}{\textbf{H}}

\Crefname{assumptionA}{\textbf{A}\hspace{-6pt}}{\textbf{A}\hspace{-6pt}}
\crefname{assumptionA}{\textbf{A}}{\textbf{A}}

\newtheorem{assumptionS}{\textbf{S}\hspace{-3pt}}
\Crefname{assumptionS}{\textbf{S}\hspace{-3pt}}{\textbf{S}\hspace{-3pt}}
\crefname{assumptionS}{\textbf{S}}{\textbf{S}}

\Crefname{probleme}{\textbf{Problem}\hspace{-3pt}}{\textbf{Problem}\hspace{-3pt}}
\crefname{probleme}{\textbf{Problem}}{\textbf{Problem}}

\Crefname{assumptionG}{\textbf{G}\hspace{-4pt}}{\textbf{G}\hspace{-4pt}}
\crefname{assumptionG}{\textbf{G}}{\textbf{G}}

\makeatother

\newcommand{\varphibf}{\boldsymbol{\varphi}}

\def\bareta{\bar{\eta}}
\def\tildem{\tilde{m}}
\def\tildeb{\tilde{b}}
\def\Rrm{K}
\def\barb{\bar{b}}

\def\raymala{K}

\def\pibar{\bar{\pi}}
\def\bound{a}
\def\paramzv{\param^{*}_{\operatorname{zv}}}
\def\Hzv{H_{\operatorname{zv}}}
\def\bzv{b_{\operatorname{zv}}}

\def\mcbb{\mathcal{B}}  

\def\msa{\mathsf{A}}

\def\mrc{\mathrm{C}}

\def\mrC{\mathrm{C}}
\def\rmC{\mrC}

\def\varble{\,\cdot\,}

\newcommandx{\invpihat}[1][1=n]{\hat{\invpi}_{#1}}
\newcommandx{\invpihattrain}[1][1=m]{\tilde{\invpi}_{#1}}

\def\RKer{R}

\def\sgP{P}

\def\sY{Y}
\def\sX{X}
\def\nZ{Z}

\def\invpi{\pi}
\def\invpig{\pi_\step}

\def\step{\gamma}
\newcommandx{\varinf}[1][1=]{\ifthenelse{\equal{#1}{}}{\sigma^2_\infty}{\sigma^2_{\infty,#1}}}

\def\DD{\operatorname{D}}
\def\sPoid{\hat{f}_{\operatorname{d}}}
\def\sPoic{\hat{f}}

\def\generatorA{\mathscr{E}}

\def\genag{\mathscr{E}_{\step}}

\newcommandx{\sPoif}[1][1=f]{\hat{#1}}

\def\tf{f}
\def\tzf{\tilde{f}}
\def\tg{g_\step}

\def\th{h}

\def\Rker{R}
\def\Qker{Q}
\def\Rkerg{R_\step}
\newcommand{\setpoly}[1]{\mathrm{C}^{#1}_{\operatorname{poly}}}
\newcommand{\setpolyinf}{\mathrm{C}^{\infty}_{\operatorname{poly}}}
\newcommand{\setpolyinfr}{\mathrm{C}^{\infty}_{\operatorname{poly}}(\rset^d,\rset)}
\def\Rmala{\Rker_\step^{\scriptscriptstyle{\operatorname{MALA}}}}
\def\Qgam{\Qker_\step}
\def\Rrwm{\Rker_\step^{\scriptscriptstyle{\operatorname{RWM}}}}

\def\Rula{\Rker_\step^{\scriptscriptstyle{\operatorname{ULA}}}}
\def\alphamala{\tau_\step^{\scriptscriptstyle{\operatorname{MALA}}}}
\def\alpharwm{\tau^{\scriptscriptstyle{\operatorname{RWM}}}_{\gamma}}

\newcommandx{\genrula}[1][1=]{\ifthenelse{\equal{#1}{}}
{\generatorA_{\step}^{\scriptscriptstyle{\operatorname{ULA}}}}
{\generatorA_{#1}^{\scriptscriptstyle{\operatorname{ULA}}}}}
\newcommandx{\genrmala}[1][1=]{\ifthenelse{\equal{#1}{}}
{\generatorA_{\step}^{\scriptscriptstyle{\operatorname{MALA}}}}
{\generatorA_{#1}^{\scriptscriptstyle{\operatorname{MALA}}}}}
\newcommandx{\tgenrmala}[1][1=]{\ifthenelse{\equal{#1}{}}
{\tilde{\generatorA}_{\step}^{\scriptscriptstyle{\operatorname{MALA}}}}
{\tilde{\generatorA}_{#1}^{\scriptscriptstyle{\operatorname{MALA}}}}}

\newcommandx{\genrrwm}[1][1=]{\ifthenelse{\equal{#1}{}}
{\generatorA_{\step}^{\scriptscriptstyle{\operatorname{RWM}}}}
{\generatorA_{#1}^{\scriptscriptstyle{\operatorname{RWM}}}}}
\newcommandx{\genrbar}[1][1=]{\ifthenelse{\equal{#1}{}}
{\generatorA_{\step}^{\scriptscriptstyle{\operatorname{B}}}}
{\generatorA_{#1}^{\scriptscriptstyle{\operatorname{B}}}}}

\def\zetag{\zeta_\step}

\def\pU{U}

\def\nablaU{\nabla \pU}

\def\acceptrwm{\accept^{\scriptscriptstyle{\operatorname{RWM}}}}

\def\lV{V}
\def\cFL{b}
\def\lambdaFL{\lambda}

\newcommand{\vertiii}[1]{{\left\vert\kern-0.25ex\left\vert\kern-0.25ex\left\vert #1
    \right\vert\kern-0.25ex\right\vert\kern-0.25ex\right\vert}}

\def\rhoVunif{\rho}

\def\base{\psi}
\def\basea{\psi^{\text{1st}}}
\def\baseb{\psi^{\text{2nd}}}
\def\param{\theta}
\def\paramstar{{\param^{*}}}

\def\invpicv{\invpi^{\scriptscriptstyle{\operatorname{CV}}}}
\def\pb{p}

\def\nb{N}
\def\xb{\mathsf{Z}} 
\def\yb{\mathsf{Y}} 
\def\bb{x}

\newcommandx{\sigS}[1][1=f]{ \hat{\sigma}^2_{N,n}(#1)}
\newcommand{\logl}[1]{%
    \IfEqCase{#1}{%
        {l}{\ell_{\operatorname{log}}}%
        {p}{\ell_{\operatorname{pro}}}%
    }[\PackageError{logl}{Undefined option to logl: #1}{}]%
}%
\newcommand{\Ub}[1]{%
    \IfEqCase{#1}{%
        {l}{\pU_{\operatorname{log}}}%
        {p}{\pU_{\operatorname{pro}}}%
    }[\PackageError{Ub}{Undefined option to Ub: #1}{}]%
}%
\newcommand{\pib}[1]{%
    \IfEqCase{#1}{%
        {l}{\invpi_{\operatorname{log}}}%
        {p}{\invpi_{\operatorname{pro}}}%
    }[\PackageError{Ub}{Undefined option to Ub: #1}{}]%
}%
\def\varbb{\varsigma^2}
\def\erfc{\operatorname{erfc}}

\def\accept{\mathbf{A}}
\def\chirwm{\chi}
\def\xtilde{\tilde{x}}
\def\steptilde{\tilde{\step}}
\def\complementaire{\text{c}}
\newcommandx{\rayrwm}[1][1=]{\ifthenelse{\equal{#1}{}}{K_{\step}}{K_{#1}}}
\def\Krwm{K}
\def\crwm{C}
\def\rrwm{\mathsf{r}_{\step}}
\newcommand{\bouled}[3]{\mathrm{B}_{#3}(#1,#2)}
\newcommand{\ball}[2]{\mathrm{B}(#1,#2)}
\newcommand{\boulefermeed}[3]{\overline{\mathrm{B}}_{#3}(#1,#2)}
\newcommand{\cone}[2]{\operatorname{cone}(#1, #2)}
\def\thetag{\theta_\step}
\def\phirwm{\varphi}
\def\brwm{b}
\def\corwm{c}
\def\cdfc{\bar{\Phi}}
\def\arwm{a}
\def\Grwm{G}
\def\vmin{v_{\operatorname{min}}}
\def\aU{a}
\def\compact{\mathsf{K}}

\def\Lspace{\operatorname{L}}

\newcommandx{\flecheLimiteLoi}[1][1=\mu]{\overset{\PP_{#1}-\text{weakly}}{\underset{n\to+\infty}{\Longrightarrow}}}
\newcommandx{\flecheLimiteLoiNu}{\overset{\text{weakly}}{\underset{n\to+\infty}{\Longrightarrow}}}

\newcommandx{\flecheLimiteLoit}[1][1=x]{\overset{\PP_{#1}-\text{weakly}}{\underset{t\to+\infty}{\Longrightarrow}}}
\def\gauss{\mathcal{N}}

\def\Borel{\mathcal{B}}
\def\borelSet{\mathcal{B}}

\def\Tr{\operatorname{T}}

\newcommandx{\functionspace}[2][1=+]{\mathcal{F}_{#1}(#2)}

\newcommand{\argmin}{\operatorname*{arg\,min}}

\newcommandx{\VarDeux}[3][3=]{\operatorname{Var}^{#3}_{#1}\left[#2 \right]}

\newcommand{\1}{\mathbbm{1}}
\newcommand{\bigone}{\operatorname{1}}

\newcommand{\LeftEqNo}{\let\veqno\@@leqno}

\newcommand{\floor}[1]{\left\lfloor #1 \right\rfloor}
\newcommand{\ceil}[1]{\left\lceil #1 \right\rceil}

\newcommand{\N}{\ensuremath{\mathbb{N}}}

\newcommand{\PE}{\mathbb{E}}
\newcommand{\PP}{\mathbb{P}}

\newcommand{\absolute}[1]{\left\vert #1 \right\vert}
\newcommand{\abs}[1]{\left\vert #1 \right\vert}
\newcommand{\absLigne}[1]{ \vert #1  \vert}
\newcommand{\tvnorm}[1]{\| #1 \|_{\mathrm{TV}}}

\newcommandx{\Vnorm}[2][1=V]{\| #2 \|_{#1}}
\newcommandx{\VnormEq}[2][1=V]{\left\| #2 \right\|_{#1}}
\newcommandx{\VnormEqs}[2][1=V]{\| #2 \|_{#1}}

\newcommandx{\norm}[2][1=]{\ifthenelse{\equal{#1}{}}{\left\Vert #2 \right\Vert}{\left\Vert #2 \right\Vert^{#1}}}

\newcommandx{\normLigne}[2][1=]{\ifthenelse{\equal{#1}{}}{\Vert #2 \Vert}{\Vert #2\Vert^{#1}}}

\newcommand{\parenthese}[1]{\left(#1 \right)}
\newcommand{\parentheseLigne}[1]{(#1 )}
\newcommand{\parentheseDeux}[1]{\left[ #1 \right]}

\newcommand{\defEns}[1]{\left\lbrace #1 \right\rbrace }
\newcommand{\defEnsLigne}[1]{\lbrace #1 \rbrace }

\newcommand{\ps}[2]{\left\langle#1,#2 \right\rangle}
\newcommand{\eqdef}{:=}

\newcommandx\probaMarkovTilde[2][2=]
{\ifthenelse{\equal{#2}{}}{{\widetilde{\mathbb{P}}_{#1}}}{\widetilde{\mathbb{P}}_{#1}\left[ #2\right]}}

\newcommand{\expe}[1]{\PE \left[ #1 \right]}

\newcommand{\expeMarkov}[2]{\PE_{#1} \left[ #2 \right]}

\newcommand{\mmtilf}[1]{m_{#1}}

\newcommand{\plusinfty}{+\infty}

\def\ie{\textit{i.e.}}
\def\eqsp{\;}
\newcommand{\coint}[1]{\left[#1\right)}
\newcommand{\ocint}[1]{\left(#1\right]}
\newcommand{\ooint}[1]{\left(#1\right)}
\newcommand{\ccint}[1]{\left[#1\right]}

\newcommandx{\weight}[2][2=n]{\omega_{#1,#2}^N}

\newcommand{\boule}[2]{\mathrm{B}(#1,#2)}
\newcommand{\boulefermee}[2]{\overline{\mathrm{B}}(#1,#2)}

\def\rmd{\mathrm{d}}
\newcommandx\sequence[3][2=,3=]
{\ifthenelse{\equal{#3}{}}{\ensuremath{\{ #1_{#2}\}}}{\ensuremath{\{ #1_{#2}, \eqsp #2 \in #3 \}}}}
\newcommandx{\sequencen}[2][2=n\in\N]{\ensuremath{\{ #1, \eqsp #2 \}}}
\newcommandx\sequenceDouble[4][3=,4=]
{\ifthenelse{\equal{#3}{}}{\ensuremath{\{ (#1_{#3},#2_{#3}) \}}}{\ensuremath{\{  (#1_{#3},#2_{#3}), \eqsp #3 \in #4 \}}}}
\newcommandx{\sequencenDouble}[3][3=n\in\N]{\ensuremath{\{ (#1_{n},#2_{n}), \eqsp #3 \}}}
\newcommand{\wrt}{w.r.t.}

\def\iid{i.i.d.}
\def\rme{\mathrm{e}}

\def\eg{e.g.}

\def\rset{\mathbb{R}}

\def\nset{\mathbb{N}}
\def\nsets{\mathbb{N}^*}
\def\tildem{\varpi}

\newcommandx{\CPE}[3][1=]{{\mathbb E}^{#3}_{#1}\left[#2 \right]} 
\newcommand{\CPP}[3][]
{\ifthenelse{\equal{#1}{}}{{\mathbb P}\left(\left. #2 \, \right| #3 \right)}{{\mathbb P}_{#1}\left(\left. #2 \, \right | #3 \right)}}

\def\generator{\mathscr{L}}

\newcommandx{\osc}[2][1=]{\mathrm{osc}_{#1}(#2)}

\newcommand{\chunk}[4][]%
{\ifthenelse{\equal{#1}{}}{\ensuremath{{#2}_{#3:#4}}}{\ensuremath{#2^#1}_{#3:#4}}
}

\def\vphi{\varphi}

\def\Id{\operatorname{Id}}

\def\Csetfunction{\mathrm{C}}

\def\param{\theta}

\def\Phibf{\mbox{\protect\boldmath$\Phi$}}

\def\varble{\,\cdot\,}

\def\mrc{\mathrm{C}}
\def\bgamma{\bar{\gamma}}

\newcommand{\ensemble}[2]{\left\{#1\,:\eqsp #2\right\}}
\newcommand{\set}[2]{\ensemble{#1}{#2}}

\def\ControlFunc{g}
\def\ControlFunch{h}

\def\tildeX{\tilde{X}}

\newcommandx{\hControlFuncOpt}[1][1=m]{g_{#1}^{\star}}
\def\lin{\mathrm{lin}}
\newcommandx{\ControlFuncSet}[1][1=]{\mathcal{G}_{#1}}
\newcommandx{\ControlFuncSetH}[1][1=]{\mathcal{H}_{#1}}

\def\Param{\Theta}
\newcommandx{\pen}[1][1=m]{\operatorname{pen}_{#1}}
\newcommandx{\EmpRisk}[1][1=m]{\operatorname{R}_{#1}}
\def\bm{b}

\newcommandx{\PVar}[1][1=]{\ensuremath{\operatorname{Var}_{#1}}}

\newcommandx{\PCov}[1][1=]{\ensuremath{\operatorname{Cov}_{#1}}}

\def\ltwo{\operatorname{L}^2}

\newcommand\ke{k_e}

\newcommand{\fracaaa}[2]{#1/(#2)}

\newcommandx{\dlim}[1]{\ensuremath{\stackrel{#1}{\Longrightarrow}}}

\title{Diffusion approximations and control variates for MCMC}

\author{Nicolas Brosse\textsuperscript{1} \and
Alain Durmus\textsuperscript{2} \and
Sean Meyn\textsuperscript{3} \and
\'Eric Moulines\textsuperscript{1} \and
Anand Radhakrishnan\textsuperscript{3}
}

\date{}

\begin{document}

\footnotetext[1]{Centre de Math\'ematiques Appliqu\'ees, UMR 7641, Ecole Polytechnique, France. \\
 nicolas.brosse@polytechnique.edu, eric.moulines@polytechnique.edu }
\footnotetext[2]{Ecole Normale Sup\'erieure CMLA,
61 Av. du Pr\'esident Wilson,
94235 Cachan Cedex, France. \\
alain.durmus@cmla.ens-cachan.fr}
\footnotetext[3]{University of Florida,
Department of Electrical and Computer Engineering,
Gainesville, Florida. \\
meyn@ece.ufl.edu, a4anandr@ufl.edu}

\maketitle

\begin{abstract}
  A new methodology is presented for the construction of control
  variates to reduce the variance of additive functionals of Markov
  Chain Monte Carlo (MCMC) samplers.  Our control variates are defined
  through the minimization of the asymptotic variance of the Langevin
  diffusion over a family of functions, which can be seen as a quadratic risk
  minimization procedure.
  The use of these control variates is theoretically justified. We  show that the asymptotic variances of some well-known MCMC algorithms, including the Random Walk Metropolis and the (Metropolis) Unadjusted/Adjusted Langevin Algorithm, are close to the asymptotic variance of the Langevin diffusion.
  Several examples of Bayesian inference problems demonstrate that the corresponding reduction in the variance is significant.
\end{abstract}

\section{Introduction}
\label{sec:motivations-contrib}
Let $U : \rset^d \to \rset$ be a measurable function on $(\rset^d, \borelSet(\rset^d))$ such that $\int_{\rset^d} \rme^{-U(x)} \rmd x <\infty$. This function is associated to a probability measure $\pi$ on $(\rset^d, \borelSet(\rset^d))$ defined for all $\msa \in\mcbb(\rset^d)$ by
\[ \pi(\msa) \eqdef \int_{\msa}  \rme^{-U(x)} \rmd x / \int_{\rset^d}\rme^{-U(x)} \rmd x \eqsp. \]
We are interested in approximating $\pi(f) \eqdef \int_{\rset^d} f(x) \pi(\rmd x)$, where $f$ is a $\pi$-integrable function.
The classical Monte Carlo solution to this problem is to simulate \iid\ random variables $(X_k)_{k\in\nset}$ with distribution~$\pi$, and then to estimate $\pi(f)$ by the sample mean
\begin{equation}
  \label{eq:def-invpihatn}
  \hat{\pi}_n(f)=n^{-1}\sum_{i=0}^{n-1}f(X_{i})\eqsp.
\end{equation}
In most applications, sampling from $\pi$ is not an option.   Markov Chain Monte Carlo (MCMC) methods
provide samples from a Markov chain $(X_k)_{k \in \nset}$  with unique invariant probability $\pi$.
Under mild conditions \cite[Chapter 17]{bible}, the estimator $\hat{\pi}_n(f)$ defined by \eqref{eq:def-invpihatn} satisfies for any initial distribution a Central Limit Theorem (CLT)
\begin{equation}\label{eq:TCL-discret}
n^{-1/2} \sum_{k=0}^{n-1} \tzf(\sX_k) \flecheLimiteLoiNu \gauss(0, \varinf[\operatorname{d}](\tf)) \eqsp,
\end{equation}
where  $ \tzf= f - \pi(f)$  and  $\varinf[\operatorname{d}](\tf) \geq 0$ is referred to as the asymptotic variance associated to $f$ and  $\gauss(m,\sigma^2)$ denotes a Gaussian distribution with mean $m$ and variance $\sigma^2$.

The aim of the present paper is to propose a new methodology to reduce
the asymptotic variance  of a family of MCMC algorithms.
This method consists in
constructing suitable control variates, \ie~we consider a family of $\invpi$-integrable
functions $\ControlFuncSetH  \subset \set{h : \rset^d \to \rset}{\pi(h) = 0}$ and  then
choose $\ControlFunch \in  \ControlFuncSetH$ such that
$\varinf[\operatorname{d}](\tf+ \ControlFunch) \leq
\varinf[\operatorname{d}](\tf)$.
Reducing the variance of Monte Carlo estimators is a very active research domain: see \eg~\cite[Chapter 4]{robert:casella:2004}, \cite[Section 2.3]{liu2008monte}, and \cite[Chapter 5]{rubinstein:kroese:2017}  for an overview of the main methods - see also \Cref{subsec:bibliography}.

Analysis and motivation are based on the Langevin diffusion defined  by
\begin{equation}\label{eq:SDE}
\rmd \sY_t = - \nabla \pU (\sY_t) \rmd t + \sqrt{2} \rmd B_t \eqsp,
\end{equation}
where $(B_t)_{t\geq 0}$ is a $d$-dimensional Brownian motion. In the sequel, we assume that the Stochastic Differential Equation (SDE) \eqref{eq:SDE} has a unique strong solution $(\sY_t)_{t\geq 0}$ for
every initial condition $x\in\rset^d$.
Under appropriate conditions (see \cite{Bhattacharya1982Classical,cattiaux2012central}), $\invpi$ is invariant for the Markov process $(Y_t)_{t \geq 0}$ and the following CLT holds:
\begin{equation} \label{eq:TCL-continu}
t^{-1/2} \int_{0}^t \tzf(\sY_s) \rmd s
\overset{\text{weakly}}{\underset{t\to+\infty}{\Longrightarrow}}
\gauss(0, \varinf(\tf)) \eqsp.
\end{equation}
The main contribution of this paper is the introduction of a new method to compute control variates based on the expression of the asymptotic variance $\varinf(f)$ given in \eqref{eq:TCL-continu}.
%
%
%
For any twice continuously differentiable function $\varphi$, the differential generator acting on $\varphi$ is denoted by
\begin{equation}\label{eq:def-generator}
\generator \varphi = - \ps{\nabla U}{\nabla \varphi} + \Delta \varphi \eqsp.
\end{equation}
Under appropriate conditions on $\varphi$ and $\pi$, it may be shown that $\invpi(\generator \varphi) = 0$. This property suggests to consider the class of control functionals $\ControlFuncSetH= \set{h = \generator g}{g \in \ControlFuncSet}$ for the Langevin diffusion, where   $\ControlFuncSet$ is a family of ``smooth'' functions, and minimize over $\ControlFuncSetH$, the criterion
\begin{equation}
\label{eq:optimisation-criterion}
\ControlFunch \mapsto\varinf(\tf + \ControlFunch) \eqsp.
\end{equation}
The use of control functionals $h \in \ControlFuncSetH$ has already been proposed in \cite{assaraf1999zero} with applications to quantum Monte Carlo calculations; improved schemes have been later considered in \cite{Mira2013,papamarkou2014} with applications to computational Bayesian inference.
Although $\ControlFuncSetH$ is a class of control functionals for the
Langevin diffusion, the choice of controls variates minimizing the criterion
\eqref{eq:optimisation-criterion} for some MCMC algorithms is motivated by the fact the asymptotic variance
$\varinf[\operatorname{d}](\tf)$, defined in \eqref{eq:TCL-discret} and associated to the Markov chains associated with these methods, is
(up to a scaling factor) a good approximation of the asymptotic
variance of the Langevin diffusion $\varinf(\tf)$ defined in
\eqref{eq:TCL-continu}.

The remainder of the paper is organized as follows.
In \Cref{sec:asymp-variance}, we present our methodology to minimize \eqref{eq:optimisation-criterion} and the construction of control variates for some MCMC algorithms. In \Cref{sec:asympt-expans-asympt}, we state our main result which guarantees that the asymptotic variance $\varinf[\operatorname{d}](\tf)$ defined in \eqref{eq:TCL-discret} and associated with  a given MCMC method is  close (up to a scaling factor) to the asymptotic variance of the Langevin diffusion $\varinf(\tf)$ defined in \eqref{eq:TCL-continu}.
We show that under appropriate conditions on $U$, the Metropolis Adjusted/Unadjusted Langevin Algorithm (MALA and ULA) and the Random Walk Metropolis (RWM) algorithm fit the framework of our methodology.
In \Cref{sec:application_cv}, Monte Carlo experiments illustrating the performance of our method are presented.
The proofs are postponed to \Cref{sec:proofs,sec:geom-ergodicity-mala} and to the Appendix.


\subsection*{Notation}

Let $\mathcal{B}(\rset^d)$ denote the Borel $\sigma$-field of $\rset^d$. Moreover, let $\Lspace^1(\mu)$ be the set of $\mu$-integrable functions for $\mu$ a probability measure on $(\rset^d, \Borel(\rset^d))$. Further, $\mu(\tf)=\int_{\rset^d} \tf(x) \rmd \mu(x)$ for an $\tf\in\Lspace^1(\mu)$.
Given a Markov kernel $R$ on $\rset^d$, for all $x\in\rset^d$ and $\tf$ integrable under $R(x,\cdot)$, denote by $R \tf(x) = \int_{\rset^d} \tf(y) R(x, \rmd y)$.
Let $V: \rset^d \to \coint{1,\infty}$ be a measurable function.
The $V$-total variation distance between two probability measures $\mu$ and $\nu$ on $(\rset^d, \Borel(\rset^d))$ is defined as $\Vnorm[V]{\mu-\nu} = \sup_{\absolute{f} \leq V} \abs{\mu(f) - \nu(f)}$.
If $V = 1$, then $\Vnorm[V]{\cdot}$ is the total variation  denoted by $\tvnorm{\cdot}$.
For a measurable function $f:\rset^d\to\rset$, define $\Vnorm{f} = \sup_{x\in\rset^d} \absolute{f(x)} / V(x)$.

For $u,v\in\rset^d$, define the scalar product $\ps{u}{v} = \sum_{i=1}^d u_i v_i$ and the Euclidian norm $\norm{u} = \ps{u}{u}^{1/2}$.
Denote by $\mathbb{S}(\rset^d) = \defEns{u\in\rset^d : \norm{u} = 1}$.
For $a,b\in\rset$, denote by $a\vee b = \max(a,b)$, $a \wedge b = \min(a,b)$ and $a_{+} = a \vee 0$.
For $a\in\rset_+$, $\floor{a}$ and $\ceil{a}$ denote respectively the floor and ceil functions evaluated in $a$.
We take the convention that for $n,p \in \nset$, $n <p$ then $\sum_{p}^n =0$, $\prod_p ^n = 1$ and $\defEns{p,\ldots,n} = \emptyset$.
Define for $t\in\rset$, $\Phi(t) = (2\uppi)^{-1/2}\int_{-\infty}^t \rme^{-r^2 / 2} \rmd r$ and $\cdfc(t) = 1 - \Phi(t)$.
In addition,  $\varphibf$ stands for the $d$-dimensional standard Gaussian density, \ie~$\varphibf(z) = (2\uppi)^{-d/2} \rme^{-\norm{z}^2 /2}$ for $z\in\rset^d$.

For $k \in\nset$, $m,m' \in\nset^*$ and $\Omega,\Omega'$ two open sets of $\rset^m, \rset^{m'}$ respectively, denote by $\Csetfunction^k(\Omega, \Omega')$, the set of
$k$-times continuously differentiable functions. For $f \in
\Csetfunction^2(\rset^d, \rset)$, denote by $\nabla f$ the gradient of $f$ and by $\Delta f$ the Laplacian of $f$. 
For $k\in\nset$ and $\tf \in\Csetfunction^k(\rset^d, \rset)$, denote by $\DD^i \tf$ the $i$-th order differential of $\tf$ for $i\in\defEns{0,\ldots,k}$.
For $x\in\rset^d$ and $i\in\defEns{1,\ldots,k}$, define $\norm{\DD^0 \tf(x)} = \absolute{\tf(x)}$, $\norm{ \DD^i \tf (x)} = \sup_{u_1,\ldots,u_i \in\mathbb{S}(\rset^d)} \DD^i \tf(x)[u_1, \ldots, u_i]$. 
For $k,p\in\nset$ and $f\in\Csetfunction^{k}(\rset^d,\rset)$, define the semi-norm
\begin{equation*}
  \VnormEq[k,p]{f} = \sup_{x\in\rset^d, \eqsp i\in\{0,\ldots,k\}} \norm{\DD^i f(x)} / (1+\norm[p]{x}) \eqsp.
\end{equation*}
Define  $\setpoly{k}(\rset^d,\rset) = \defEns{f\in\Csetfunction^{k}(\rset^d,\rset) : \inf_{p\in\nset} \Vnorm[k,p]{f} < \plusinfty}$ and for any $f \in \setpoly{k}(\rset^d,\rset)$, we consider the semi-norm
\begin{equation*}
  \norm{f}_{k} = \norm{f}_{k,p} \text{ where } p = \min\{ q \in \nset \, : \, \Vnorm[k,q]{f} < \plusinfty\} \eqsp.
\end{equation*}
Finally, define
$\setpoly{\infty}(\rset^d,\rset) = \cap_{k \in \nset} \setpoly{k}(\rset^d,\rset)$.




 \section{Langevin-based control variates for MCMC methods}
 \label{sec:asymp-variance}

\subsection{Method}

We introduce in the following our methodology based on control variates for the Langevin diffusion.  In order not to obscure the main ideas of this method, we present it informally. Results which justify rigorously the related derivations are postponed to \Cref{sec:asympt-expans-asympt}.

We consider a family of control functionals $\ControlFuncSet \subset \setpoly{2}(\rset^d,\rset)$.
There is a great flexibility in the choice of the family $\ControlFuncSet$.
We illustrate our methodology through a simple example
\begin{equation}
\label{eq:definition-linear-control}
\ControlFuncSet[\lin]= \set{g=\ps{\param}{\base}}{\param \in \Param}
 \text{ where } \base= \{\base_i\}_{i=1}^{\pb}, \base_i \in \setpoly{2}(\rset^d,\rset), \eqsp i \in \{1,\dots,\pb\} \eqsp,
\end{equation}
with $\Theta \subset \rset^{\pb}$,
but the method developed in this paper  is by no means restricted to a linear parameterized family.

A key property of the Langevin diffusion which is the basis of our methodology is the following ``carré du champ'' property (see for example \cite[Section 1.6.2, formula 1.6.3]{bakry:gentil:ledoux:2014}): for all $g_1,g_2\in\setpoly{2}(\rset^d,\rset)$,
\begin{equation}
\label{eq:carre-du-champ}
\pi\parenthese{ g_1 \generator g_2} = \pi\parenthese{ g_2 \generator g_1} = -\pi\parenthese{\ps{\nabla g_1}{\nabla g_2}} \eqsp,
\end{equation}
which reflects in particular that $\generator$ is a self-adjoint operator on a dense subspace of $\ltwo(\pi)$, the Hilbert space of square integrable function \wrt\ $\pi$.
A straightforward consequence of \eqref{eq:carre-du-champ} (setting  $g_1 = \bigone$) is that
$\pi(\generator g) = 0$ for any function $g \in \setpoly{2}(\rset^d,\rset)$. This observation implies that
$f$ and $f+ \generator g$ have the same expectation with respect to $\pi$ for any $f \in \setpoly{2}(\rset^d,\rset)$ and $g \in \setpoly{2}(\rset^d,\rset)$.
Therefore, as emphasized in the introduction, if the CLT \eqref{eq:TCL-continu} holds, a relevant choice of control variate for the Langevin diffusion to estimate $f \in \setpoly{2}(\rset^d,\rset)$, is $h^{\star} = \generator g^{\star}$, where $g^{\star}$ is a minimizer of
\begin{equation}
  \label{eq:fun_asympt_var_g_g}
  g \mapsto\varinf(\tf + \generator \ControlFunc) \eqsp.
\end{equation}
In the following, we explain how this optimization problem can be practically solved.

It is shown in \cite{Bhattacharya1982Classical} (see also \cite{glynn1996} and \cite{cattiaux2012central}) that under appropriate conditions on $U$ and $f$,  the solution $(Y_t)_{t \geq 0}$ of the Langevin diffusion \eqref{eq:SDE} satisfies the CLT  \eqref{eq:TCL-continu} where the asymptotic variance is given by
\begin{equation}
\label{eq:asymptotic-variance}
\varinf(\tf) = 2 \invpi\parentheseLigne{\sPoic \{f-\pi(f)\} } \eqsp,
\end{equation}
and $\sPoic \in \setpoly{2}(\rset^d,\rset)$ satisfies Poisson's equation:
\begin{equation}\label{eq:poisson-eq-langevin}
\generator \sPoic = - \tzf \eqsp, \quad \text{where $\tzf= f - \pi(f)$} \eqsp.
\end{equation}
Another expression for $\varinf(\tf)$ is,  using \eqref{eq:carre-du-champ} and \eqref{eq:poisson-eq-langevin}:
\begin{equation}
\label{eq:key-relation-variance}
\varinf(\tf) = 2\invpi(\sPoic \tzf )
= - 2\invpi (  \sPoic \generator \sPoic )
 = 2 \invpi  ( \|\nabla \sPoic \|^2 ) \eqsp.
\end{equation}
Based on \eqref{eq:carre-du-champ}, \eqref{eq:asymptotic-variance} and \eqref{eq:key-relation-variance}, we see now how the minimization of \eqref{eq:fun_asympt_var_g_g} can be computed in practice. First, by definition \eqref{eq:poisson-eq-langevin}, for all $\ControlFunc \in \ControlFuncSet$, $\sPoic- \ControlFunc \in\setpoly{2}(\rset^d,\rset)$ is a solution to the Poisson equation
\[
\generator (\sPoif-\ControlFunc) = \invpi(f + \generator \ControlFunc) - (f + \generator \ControlFunc) \eqsp.
\]
Therefore, we get for all $\ControlFunc \in \ControlFuncSet$, using $\pi(\generator \ControlFunc) = 0$ and   \eqref{eq:asymptotic-variance} 
\begin{equation*}
\varinf(\tf + \generator \ControlFunc)
= 2\invpi\parenthese{(\sPoic - \ControlFunc) \defEns{\tzf + \generator \ControlFunc}} \eqsp.
= 2\invpi(\| \nabla \hat{f} - \nabla \ControlFunc \|^2)\eqsp.
\end{equation*}
In addition, by \eqref{eq:carre-du-champ} and \eqref{eq:poisson-eq-langevin}, we get that  $\invpi(\sPoif \generator \ControlFunc)= -\invpi(\tzf \ControlFunc)$, and we obtain using  \eqref{eq:key-relation-variance} that
\begin{align}
\nonumber
\varinf(\tf+ \generator \ControlFunc)
&= 2\invpi  ( \sPoif \tzf ) - 2\invpi(\ControlFunc \tzf ) + 2\invpi(\sPoif \generator \ControlFunc) - 2\invpi(\ControlFunc \generator \ControlFunc) \\
\label{eq:key-expression}
&= 2\invpi (\sPoif \tzf) - 4\invpi(\ControlFunc \tzf ) + 2\invpi( \| \nabla \ControlFunc\|^2) \eqsp.
\end{align}
Minimizing the map \eqref{eq:fun_asympt_var_g_g} is  equivalent  to minimization of  $\ControlFunc \mapsto - 4\invpi(\ControlFunc \tzf) + 2\invpi( \| \nabla \ControlFunc\|^2)$. It means that we might actually minimize the function
$\ControlFunc \mapsto \varinf(\tf + \generator \ControlFunc)$ \emph{without} computing the solution  $\sPoic$ of the Poisson equation, which is in general a computational bottleneck.

When $\ControlFunc_\param = \ps{\param}{\base} \in \ControlFuncSet[\lin]$, then \eqref{eq:key-expression} may be rewritten as:
\begin{equation*}
\varinf(\tf + \generator \ControlFunc_\param) = 2 \param^{\Tr} H \param - 4\ps{\param}{\bm} + \varinf(\tf) \eqsp,
\end{equation*}
where $H \in \rset^{p \times p}$ and $\bm$ are given for any $i,j \in \{1,\ldots,p\}$ by
\begin{equation*}
  H_{ij} = \pi(\ps{\nabla \base_i}{\nabla \base_j})  \quad \text{and} \quad \bm_i= \invpi(\base_i \tzf) \eqsp.
\end{equation*}
Note that $H$ is by definition a symmetric
semi-positive definite matrix. If $(1, \base_1,\ldots,\base_\pb)$ are linearly independent in
$\setpoly{2}(\rset^d,\rset)$, then $H$ is full rank and the minimizer of $\varinf(\tf + \generator \ControlFunc_\param)$ is given by
\begin{equation}
\label{eq:min-asymp-var-diffusion}
  \paramstar = H^{-1} \bm \eqsp.
\end{equation}

In conclusion, in addition to its  theoretical interest,   the Langevin diffusion \eqref{eq:SDE} is an attractive model because optimization of the asymptotic variance is greatly simplified.    However, we are not advocating simulation of this diffusion in MCMC applications.   The main contribution of this paper is to show that the optimal control variate for the diffusion remains nearly optimal for many standard MCMC algorithms.

One example is  the Unadjusted Langevin Algorithm (ULA), the Euler
discretization scheme associated to the Langevin SDE \eqref{eq:SDE}:
\begin{equation*}
X_{k+1} = X_k - \step \nablaU(X_k) + \sqrt{2\step} Z_{k+1} \eqsp,
\end{equation*}
where $\step>0$ is the step size and $(Z_k)_{k\in\nset}$ is an \iid~sequence
of standard Gaussian $d$-dimensional random vectors. The idea of using
the Markov chain $(X_k)_{k\in\nset}$ to sample approximately from
$\invpi$ has been first introduced in the physics literature by
\cite{parisi:1981} and popularized in the computational statistics
community by \cite{grenander:1983} and \cite{grenander:miller:1994}.
As shown below, other examples are the Metropolis Adjusted Langevin Algorithm (MALA) algorithm (for which an additional Metropolis-Hastings correction step is added) but also for MCMC algorithms  which do not seem to be ``directly'' related to the Langevin diffusion,  like the Random Walk Metropolis algorithm (RWM).

To deal with these different algorithms within the same theoretical framework, we consider a family of Markov kernels $\set{\RKer_\step}{ \gamma \in \ocint{0,\bar{\gamma} }}$,  parameterized  by  a scalar parameter $\step\in \ocint{0,\bar{\gamma} }$ where $\bgamma >0$.
For the ULA and MALA algorithm, $\gamma$ is the stepsize in the Euler discretization of the diffusion; for the RWM this is the variance of the random walk proposal.
For any initial probability $\xi$ on $(\rset^d, \borelSet(\rset^d))$ and $\step\in \ocint{0,\bar{\gamma}}$, denote by $\PP_{\xi, \step}$ and $\PE_{\xi, \step}$ the probability and the expectation respectively on the canonical space of the Markov chain  with initial probability $\xi$ and of transition kernel $\Rkerg$.
By convention, we set $\PE_{x,\step} = \PE_{\updelta_x,\step}$ for all $x\in\rset^d$. We denote by $(X_k)_{k \geq 0}$ the canonical process.
It is assumed below that $\set{\RKer_\step}{ \gamma \in \ocint{0,\bar{\gamma} }}$, $f$ and $\ControlFuncSet$ satisfy the following assumptions. Roughly speaking, these conditions impose that for any $\gamma \in \ocint{0,\bgamma}$ and $g \in \ControlFuncSet$, the discrete CLT \eqref{eq:TCL-discret} holds for the function $f + \generator g$, and that the associated asymptotic variance $\sigma^2_{\infty,\gamma}(f+\generator g)$ is sufficiently close to $\sigma_{\infty}(f + \generator g)$ given by the continuous CLT \eqref{eq:SDE}, as  $\step\downarrow 0^+$, so that control functionals for the Markov chain $(X_k)_{k \in\nset}$ can be derived using the methodology we developed above for the Langevin diffusion.
\begin{enumerate}[label=(\Roman*)] 
\item \label{item:assum-general-1}
  For each  $\step \in \ocint{0,\bar{\step} } $, $\Rker_{\step}$ 
  has an  invariant probability distribution  $\invpig$ satisfying $\pi_\step(|f+ \generator g|) < \infty$ for any  $g \in \ControlFuncSet$.
\item \label{item:assum-general-2}
  For any   $g \in \ControlFuncSet$ and  $\step \in \ocint{0,\bar{\gamma} } $,
  \begin{equation}\label{eq:clt-Rkerg}
  \sqrt{n}(\invpihat(f+\generator g) - \invpi_\step(f+\generator g))\overset{\text{weakly}}{\underset{n\to+\infty}{\Longrightarrow}}
  \gauss(0, \varinf[\step](f+\generator g))
  \end{equation}
  where $\invpihat(f+\generator g)$ is the sample mean (see \eqref{eq:def-invpihatn}),
  and $\varinf[\step](f+\generator g) \geq 0$ is the asymptotic variance (see \eqref{eq:TCL-discret}) relatively to $\Rkerg$.
\item \label{item:assum-general-3}
For any $g \in \ControlFuncSet$, as $\gamma \downarrow 0^+$,
\begin{align}
\label{eq:approximation-loi-variance}
\step \varinf[\step](f+\generator g) &=  \varinf(f+\generator g) + o(1) \eqsp, \\
\label{eq:approximation-invpig-invpi}
\invpig(f+\generator g) &= \invpi(f+\generator g) + O(\step) \eqsp,
\end{align}
where $\varinf(f+\generator g)$ is defined in \eqref{eq:asymptotic-variance}.
\end{enumerate}
The verification that these assumptions are satisfied for the ULA, RWM and MALA algorithms (under appropriate technical conditions), in the case $f \in \setpoly{\infty}(\rset^d,\rset)$ and $\ControlFuncSet \subset \setpoly{\infty}(\rset^d,\rset)$,  is postponed to \Cref{sec:asympt-expans-asympt}.
The standard conditions \ref{item:assum-general-1}--\ref{item:assum-general-2} are in particular satisfied if, for any $\gamma \in \ocint{0,\bar{\gamma}}$,  $\Rkerg$ is $V$-uniformly geometrically ergodic for some measurable function  $V: \rset^d \to \coint{1,\plusinfty}$, \ie\ it admits an invariant probability measure $\pi_\gamma$ such that $\pi_\gamma(V) < \plusinfty$ and
there exist $C_\gamma \geq 0$ and $\rho_\gamma \in \coint{0,1}$ such that for any probability measure $\xi$ on $(\rset^d,\mcbb(\rset^d))$  and $n \in \nset$,
\begin{equation*}
  \label{eq:def_v_unif_ergo}
  \Vnorm{\xi R^n_{\gamma} - \pi_{\gamma}} \leq C_\gamma \xi(V) \rho_\gamma^n \eqsp,
\end{equation*}
(see \eg~\cite{bible} or \cite{douc:moulines:priouret:soulier:2018}). Condition \ref{item:assum-general-3}  requires a specific form of the dependence of $C_\gamma$ and $\rho_\gamma$ on $\gamma$.

Based on \ref{item:assum-general-1}--\ref{item:assum-general-3} and \eqref{eq:min-asymp-var-diffusion}, the estimator of $\invpi(\tf)$ we suggest is given for $N,n,m\in\nset^*$ by
\begin{equation}\label{eq:def-invpi-cv}
  \invpicv_{N, n, m}(\tf) = \frac{1}{n}\sum_{k=N}^{n+N-1} \parenthese{\tf(X_k) + \generator \hControlFuncOpt (X_k)} \eqsp,
\end{equation}
where $N$ is the length of the burn-in period
and $\hControlFuncOpt \in \argmin_{\ControlFunc \in \ControlFuncSet} \EmpRisk(\ControlFunc)$ is a minimizer of the structural  risk associated with \eqref{eq:key-expression}
\begin{equation}
\label{eq:empirical-risk-minimization}
\EmpRisk(\ControlFunc) = \frac{1}{m} \sum_{k=N}^{N+m-1} \left\{-2 \ControlFunc(\tilde{X}_k) \tzf_m(\tilde{X}_k) +  \| \nabla \ControlFunc(\tilde{X}_k) \|^2 \right\} \eqsp, 
\end{equation}
where $\tzf_m(x)=  f(x)-m^{-1} \sum_{k=N}^{N+m-1} f(\tilde{X}_k)$.
Here $(\tilde{X}_k)_{k\in\nset}$ can be an independent copy of (or be identical to) the Markov chain $(X_k)_{k \in\nset}$ and $m$ is the length of the sequence used to estimate the control variate.
In this article, we do not study to what extent minimizing the empirical asymptotic variance \eqref{eq:empirical-risk-minimization} leads to the minimization of the asymptotic variance of $\invpicv_{N, n, m}(\tf)$ \eqref{eq:def-invpi-cv} as $n\to\plusinfty$; such a problem has been tackled by \cite{Belomestny2018} in the \iid~case.
To control the complexity of the class of functions $\ControlFuncSet$, a penalty term  may be added in \eqref{eq:empirical-risk-minimization}. The use of a penalty term to control the excess risk in the estimation of the control variate has been proposed and discussed in \cite{south:mira:drovandi:2018}.
Concerning the choice of $\ControlFuncSet$, the simplest case is $\ControlFuncSet[\lin]$ defined by \eqref{eq:definition-linear-control}, corresponding to the parametric case, and it is by far the most popular approach.
It is possible to go one step further and adopt fully non-parametric approaches like kernel regression methods \cite{oates:girolami:chopin:2016} or neural networks \cite{zhu:wan:zhong:2018}.

If the control function is a linear combination of  functions, $\ControlFunc_\param = \ps{\param}{\base}$ where $\base= \set{\base_i}{1 \leq i \leq \pb}$, then the empirical risk \eqref{eq:empirical-risk-minimization} may be expressed as
\begin{equation*}
  \EmpRisk(\ControlFunc_\param)= -2\ps{\param}{\bm_m} + \ps{\param}{H_m \param} \eqsp,
\end{equation*}
where  for $1 \leq i,j \leq p$,
\begin{equation*}
[ \bm_m ]_i = \frac{1}{m} \sum_{k=N}^{N+m-1} \base_i(\tilde{X}_k) \tzf_m (\tilde{X}_k) \quad, \eqsp
[H_m]_{ij} = \frac{1}{m} \sum_{k=N}^{N+m-1} \ps{\nabla \base_i(\tildeX_k)}{\nabla \base_j(\tildeX_k)}  \, .
\end{equation*}
In this simple case, an optimizer is obtained in closed form
\begin{equation}
  \label{eq:def-paramhat-n-star}
  \param_m^* = H_m^{+}  \bm_m \eqsp,
\end{equation}
where $H_m^{+}$ is the Moore-Penrose pseudoinverse of $H_m$.

\subsection{Comparison with other control variate methods for Monte Carlo simulation}
\label{subsec:bibliography}
The construction of control variates for MCMC and the related problem of approximating solutions of Poisson equations are very active fields of research. It is impossible to give credit for all the contributions undertaken in this area: see \cite{Dellaportas2012Control}, \cite{papamarkou2014}, \cite{oates:girolami:chopin:2016} and references therein for further background.
We survey in this section only the methods which are closely connected to our approach.
\cite{HendersonThesis} and \cite[Section 11.5]{meyn2008control} proposed control variates of the form $(\RKer - \Id) \ControlFunc_\param$  where $\ControlFunc_\param \eqdef \ps{\param}{\base}$ and $R$ is the Markov kernel associated to a Markov chain $(X_k)_{k \in \nset}$ and
$\base = (\base_1, \ldots, \base_\pb)$ are known $\pi$-integrable functions. The parameter $\param\in\rset^\pb$ is obtained by minimizing the asymptotic variance
\begin{equation}\label{eq:criterion-min-asymp-var}
  \min_{\param\in\rset^{\pb}} \varinf[\operatorname{d}](\tf + (\RKer - \Id) \ControlFunc_\param) = \min_{\param\in\rset^{\pb}} \invpi\parenthese{\defEns{\sPoid - \ControlFunc_\param}^2 - \defEns{\Rker(\sPoid - \ControlFunc_\param)}^2} \eqsp,
\end{equation}
where  $\sPoid$ is solution of the \textit{discrete} Poisson equation $(\RKer - \Id) \sPoid = -\tzf$.
The method suggested in \cite[Section 11.5]{meyn2008control}  to minimize \eqref{eq:criterion-min-asymp-var} requires estimates of the solution $\sPoid$ of the Poisson equation. Temporal Difference learning is a possible candidate, but this method is complex to implement and suffers from high variance.

\cite{Dellaportas2012Control} noticed that if $\Rker$ is reversible \wrt~$\invpi$, it is possible to optimize the limiting variance \eqref{eq:criterion-min-asymp-var} without computing explicitly the Poisson solution $\sPoid$.  This approach is of course closely related with our proposed method: the reversibility of the Markov kernel is replaced here by the self-adjointness of the generator of the Langevin diffusion which implies the reversibility of the semi-group.

Each of the algorithms in the aforementioned literature  requires computation of $R\psi_i$ for each $i\in\defEns{1,\ldots,\pb}$, which is in general difficult except in very specific examples.  In \cite{HendersonThesis,meyn2008control} this is addressed by restricting to kernels $R(x,\varble)$ with finite support for each $x$. In \cite{Dellaportas2012Control} the authors consider mainly  Gibbs samplers in their numerical examples.

Our methodology is also related to the Zero Variance method proposed by \cite{Mira2013,papamarkou2014,oates:girolami:chopin:2016,south:mira:drovandi:2018}, which uses $\generator \ControlFunc$ as a control variate and chooses $\ControlFunc$ by minimizing $\invpi( \{\tzf +\generator \ControlFunc\}^2)$.
A drawback of this method stems from the fact that the optimization criterion is theoretically justified if $(X_k)_{k\in\nset}$ is \iid\ and
might significantly differ from the asymptotic variance $\varinf[\step](\tf+\generator \ControlFunc)$ defined in \eqref{eq:clt-Rkerg}.
We compare the two approaches in \Cref{sec:application_cv}.

\section{Asymptotic expansion for the asymptotic variance of MCMC algorithms}
\label{sec:asympt-expans-asympt}
In this Section, we provide conditions upon which the approximations
\eqref{eq:approximation-loi-variance}-\eqref{eq:approximation-invpig-invpi}
are satisfied for $f \in \setpoly{\infty}(\rset^d,\rset)$ and
$\ControlFuncSet \subset \setpoly{\infty}(\rset^d,\rset)$.
We first assume that  the gradient of the potential is Lipschitz:
\begin{assumption}
\label{assumption:U-Sinfty}
$\pU\in\setpoly{\infty}(\rset^d,\rset)$ and  $\nabla U$ is Lipschitz, \ie~there exists $L \geq 0$ such that for all $x,y\in\rset^d$,
\begin{equation*}
  \norm{\nabla U(x) - \nabla U(y)} \leq L \norm{x-y} \eqsp.
\end{equation*}
\end{assumption}

Denote by $(\sgP_t)_{t\geq 0}$ the semigroup associated to the SDE \eqref{eq:SDE} defined by $P_t f(x) = \expe{f(Y_t)}$ where $f$ is bounded measurable and $(Y_t)_{t\geq 0}$ is a solution of \eqref{eq:SDE} started at $x$. By construction, the target distribution $\pi$ is invariant for  $(\sgP_t)_{t\geq 0}$.

The
conditions we consider require that
$\{R_{\step} ,\: \, \step \in \ocint{0,\bgamma}\}$ is a family of
Markov kernels such that for any $\step \in \ocint{0,\bgamma}$, $R_{\step}$ approximates  $P_{\step}$ in a sense specified below. Let $V:\rset^d\to\coint{1,\plusinfty}$ be a
measurable function.
\begin{assumption}
\label{ass:geo_ergod}
\begin{enumerate}[label=(\roman*)]
\item \label{ass:geo_ergod_i}
For any $\step\in\ocint{0,\bgamma}$, $\Rkerg$ has a unique invariant distribution $\invpig$.
\item \label{ass:geo_ergod_ii} There exists $c >0$ such that  $  \liminf_{\norm{x} \to \infty}\{ \lV(x) \exp(-c\norm{x})\} > 0$,
$\pi(V) < \plusinfty$ and  $ \sup_{\step \in \ocint{0,\bgamma}} \pi_{\step}(V) < \plusinfty$.
\item \label{ass:geo_ergod_iii} There exist $C>0$ and $\rho\in\coint{0,1}$ such that for all  $x\in\rset^d$,
  \begin{align}
    \label{eq:def-V-unif}
    \text{for any  $n\in\nset$,  $\step\in\ocint{0,\bgamma}$} \eqsp, \qquad  & \Vnorm{\updelta_x \RKer_\step^n - \invpig} \leq C \rhoVunif^{n\step} \lV(x)  \eqsp, \\
    \label{eq:def-V-unif_ii}
  \text{ for any $t \geq 0$} \eqsp,  \qquad   &  \Vnorm{\updelta_x P_t - \pi} \leq C \rhoVunif^{t} \lV(x) \eqsp.
\end{align}
\end{enumerate}
\end{assumption}
These conditions imply that the kernels $\Rker_\step$ are
$V$-uniformly geometrically ergodic ``uniformly'' with
respect to the parameter $\step \in \ocint{0,\bgamma}$ with a mixing
time going to infinity as the inverse of the stepsize $\step$ when
$\step \downarrow 0^+$. Note that the mixing time of $P_{\step}$ is also inversely proportional to $\step$ when  $\step \downarrow 0^+$.

Under \Cref{assumption:U-Sinfty} and  \Cref{ass:geo_ergod}, by \cite[Lemma 2.6]{kopec:2015},  there exists a solution  $\sPoif \in\setpolyinf(\rset^d,\rset)$ to Poisson's equation \eqref{eq:poisson-eq-langevin} for any $f \in \setpolyinf(\rset^d,\rset)$ which is given for any $x \in \rset^d$ by
\begin{equation}
  \label{eq:def_poisson_int}
  \sPoif(x) = \int_{0}^{\plusinfty} P_t \tzf (x) \rmd t \eqsp.
\end{equation}
Moreover, \cite[Theorem 3.1]{cattiaux2012central} shows that, for any $f \in \setpolyinf(\rset^d,\rset)$, $t^{-1/2} \int_{0}^t \tzf(\sY_s) \rmd s$ where $(\sY_t)_{t \geq 0}$ is the solution of the Langevin SDE,
converges weakly to  $\gauss(0,\sigma^2_{\infty}(f))$ where $\sigma_\infty^2(f)$ is given by \eqref{eq:asymptotic-variance}.

Note that the assumption \Cref{ass:geo_ergod} implies that for any  $x\in\rset^d$,
\begin{align}
  \label{eq:discrete-drift-uniform-bound}
  \text{ for any  $\step\in\ocint{0,\bgamma}$, $n\in\nset^*$} \eqsp, \qquad &  \Rkerg^n \lV(x) \leq C \rho^{n\step} \lV(x) +  \sup_{\step \in \ocint{0,\bgamma}} \pi_{\step} (V)\eqsp,\\
  \nonumber
\text{ for any $t \geq 0$} \eqsp, \qquad & P_t \lV(x) \leq C \rho^{t} \lV(x) +   \pi (V)
  \eqsp.
\end{align}
We now introduce an assumption guaranteeing that the limit $\step^{-1} (\Rkerg - \Id)$  as $\step \downarrow 0^+$ is equal to the infinitesimal generator of the Langevin diffusion defined, for a bounded measurable function $f$ and $x \in \rset^d$, as $\generator f(x) = \lim_{t\to \plusinfty} \{(P_tf(x) - f(x))/t \}$, if the limit exists. This is a natural assumption if the semigroup of the Langevin diffusion evaluated at time $t=\step$, $P_\step$, and $\Rkerg$ are close as $\step \downarrow 0^+$.
\begin{assumption}
  \label{ass:dev_generator_discrete}
There exist $\alpha\geq 1$ and a family of operators $(\genag)_{\step\in\ocint{0,\bgamma}}$ with $\genag :\setpolyinf(\rset^d,\rset) \to \setpolyinf(\rset^d,\rset)$, such that for all $f\in\setpolyinf(\rset^d,\rset)$ and $\step\in\ocint{0,\bgamma}$,
\begin{equation*}
  \Rkerg f = f + \step \generator f + \step^{\alpha} \genag f \eqsp.
\end{equation*}
In addition,  there exists $\ke \in \nset$, $\ke\geq 2$ such that for all $p \in \nset$ there exist $q \in \nset$ and $C \geq 0$ (depending only on $k_e, p$)  such that for any $f \in \setpolyinf(\rset^d, \rset)$,
\begin{equation*}
  \sup_{\step \in \ocint{0,\bgamma}}  \VnormEq[0,q]{\genag f} \leq C \VnormEq[k_e,p]{f} \eqsp.
\end{equation*}
\end{assumption}
We show below that these conditions are satisfied for the Metropolis Adjusted /
Unadjusted Langevin Algorithm (MALA and ULA) algorithms (in which case
$\step$ is the stepsize in the Euler discretization of the Langevin
diffusion) and also by the Random Walk Metropolis algorithm (RWM) (in
which case $\step$ is the variance of the increment distribution).
We next give an upper bound on the difference between $\invpig$ and $\invpi$ which implies that \eqref{eq:approximation-invpig-invpi} holds. The proofs are postponed to \Cref{sec:proofs}.
\begin{proposition}
\label{item-thm-var-3}
Assume \Cref{assumption:U-Sinfty}, \Cref{ass:geo_ergod} and \Cref{ass:dev_generator_discrete} and let $p \in \nset$. Then there exists $C < \infty$ such that for all $f\in\setpolyinf(\rset^d,\rset)$ and  $\step\in\ocint{0,\bgamma}$,
\begin{equation*}
\absolute{\invpig(f) - \invpi(f)} \leq C \normLigne{f}_{\ke,p}  \step^{\alpha-1} \eqsp.
\end{equation*}
\end{proposition}
\begin{proof}
  The proof is postponed to \Cref{subsec:proof:item-thm-var-3}.
\end{proof}

The next result which is the main theorem of this Section precisely formalizes \eqref{eq:approximation-loi-variance}.

\begin{theorem} \label{prop:dev-weak-error}
Assume \Cref{assumption:U-Sinfty}, \Cref{ass:geo_ergod} and \Cref{ass:dev_generator_discrete}.
Then,  there exists $C\geq 0$ such that for all $\tf\in\setpolyinf(\rset^d,\rset)$,  $\step\in\ocint{0,\bgamma}$, $x\in\rset^d$, and $n\in\nset^*$
\begin{multline*}
  \absolute{\frac{\step}{n}\expeMarkov{x, \step}{\parenthese{\sum_{k=0}^{n-1} \defEns{ \tf(X_k) - \invpig(\tf)}}^2} - \varinf(\tf)}\\ \leq C \normLigne[2]{f}_{\ke +2,p} \defEns{\step^{(\alpha-1) \wedge 1} + \fracaaa{V(x)}{n^{1/2}\step^{1/2}}} \eqsp,
\end{multline*}
where $\varinf(\tf)$ is defined in \eqref{eq:asymptotic-variance}.
\end{theorem}
\begin{proof}
  The proof is postponed to \Cref{subsec:proof-weak-error-dev}.
\end{proof}

We now consider the ULA algorithm.  
The Markov kernel $\Rula$ associated to the ULA algorithm is given for $\step>0$, $x\in\rset^d$ and $\msa \in\borelSet(\rset^d)$ by
\begin{equation}\label{eq:def-kernel-ULA}
  \Rula(x ,\msa) =
\int_{\rset^d} \1_{\msa}\parenthese{x- \step \nabla U(x) + \sqrt{2\step} z} \varphibf(z) \rmd z \eqsp,
\end{equation}
where $\varphibf$ is the $d$-dimensional standard Gaussian density  $\varphibf : z \mapsto (2\uppi)^{-d/2} \rme^{-\norm{z}^2}$.
Consider the following additional assumption.
\begin{assumption}
\label{ass:condition_MALA}
There exist $\raymala_1 \geq 0$ and $m >0$ such that for any $x \not \in \ball{0}{\raymala_1}$,
    and $y \in \rset^d$, $      \ps{\DD^2 U(x) y }{y} \geq m \norm[2]{y}$. Moreover, there exists $M \geq 0$ such that for any $x \in \rset^d$, $    \norm{ \DD^3 U(x)} \leq M $.
\end{assumption}

\begin{proposition}\label{thm:geometric_ergodicity_ula}
  Assume \Cref{assumption:U-Sinfty} and \Cref{ass:condition_MALA}.
  There exist $\bgamma>0$ and $V:\rset^d\to\coint{1,\plusinfty}$ such that \Cref{ass:geo_ergod} is satisfied for  the family of Markov kernels $\{\Rula \, : \, \step \in \ocint{0,\bgamma}\}$. 
\end{proposition}

\begin{proof}
  The proof follows from \cite[Theorem 14, Proposition 24]{debortoli2018back}. However, for completeness and since all the tools needed for the proof of this result are used in the study of  MALA, the proof is given in  \Cref{subsec:geom-ergodicity-ula}.
\end{proof}
\begin{remark}
Note that in fact  \Cref{ass:geo_ergod} holds for ULA  under milder conditions on
$U$ using the results obtained in \cite{eberle:2015,eberle2018quantitative,debortoli2018back}. For
example, if \Cref{assumption:U-Sinfty} holds and there exist
$a_1,a_2>0$ and $c \geq 0$ such that
\begin{equation}\label{eq:cond-vgeom-ula}
  \ps{\nabla U(x)}{x} \geq a_1 \norm{x} + a_2 \norm[2]{\nabla U(x)} -c \eqsp,
\end{equation}
\cite[Theorem 14, Proposition 24]{debortoli2018back} imply that \Cref{ass:geo_ergod} holds with $V(x) = \exp\{(a_1/8)(1+\norm[2]{x})^{1/2}\}$.
\end{remark}

We now establish \Cref{ass:dev_generator_discrete}. Let $\vphi\in\setpoly{\infty}(\rset^d, \rset)$, $\bgamma>0$, $\step\in\ccint{0,\bgamma}$ and $x\in\rset^d$. Denote by $\sX_1 = x - \step \nablaU(x) + \sqrt{2\step} \nZ$ where $\nZ$ is a standard $d$-dimensional Gaussian vector, the first step of ULA. A Taylor expansion of $\varphi(X_1)$ at $x$ and integration show that  $\Rula \vphi(x) = \vphi(x) + \step \generator \vphi(x) + \step^2 \genrula \vphi(x)$ where
\begin{multline}
\genrula \vphi(x) = \frac{1}{2} \DD^2 \vphi(x)[\nablaU(x)^{\otimes 2}]
- \frac{1}{6} \step \DD^3 \vphi(x)[\nablaU(x)^{\otimes 3}] \\
- \expe{\DD^3 \vphi(x)[\nablaU(x), \nZ^{\otimes 2}]} \\
\label{eq:Eula-def}
+ \frac{1}{6} \int_0^1 (1-t)^3 \expe{\DD^4 \vphi(x - t \step \nablaU(x) + t \sqrt{2\step} \nZ)[(-\sqrt{\step} \nablaU(x) + \sqrt{2} \nZ)^{\otimes 4}]} \rmd t \eqsp.
\end{multline}
A simple application of the Lebesgue dominated convergence theorem implies then the following result.
\begin{lemma}\label{prop:ULA-dev-ergo}
  Assume \Cref{assumption:U-Sinfty}. Then for any $\bgamma >0$, $\set{\Rula}{\step \in \ocint{0,\bgamma}}$ satisfies \Cref{ass:dev_generator_discrete} with $\alpha = 2$, $\genag = \genrula$ and $\ke = 4$. 
\end{lemma}


We now examine the MALA algorithm.
The Markov kernel $\Rmala$ of the MALA algorithm, see \cite{roberts:tweedie-Langevin:1996}, is given for $\step>0$, $x\in\rset^d$, and $\msa \in\borelSet(\rset^d)$ by
\begin{align}
\nonumber
&\Rmala(x ,\msa) = \int_{\rset^d}  \1_{\msa}\parenthese{x- \gamma \nabla U(x) + \sqrt{2\gamma} z} \min(1,\rme^{-\alphamala(x,z)})  \varphibf(z) \rmd z  \\
\label{eq:def-kernel-MALA}
&\phantom{--}+ \updelta_x(\msa) \int_{\rset^d} \defEns{1- \min(1,\rme^{-\alphamala(x,z)})} \varphibf(z) \rmd z \eqsp,\\
  \nonumber
& \alphamala(x,z) = \pU(x - \gamma \nabla U(x) + \sqrt{2 \gamma}z) - \pU(x) \\
  \label{eq:def-alpha-MALA}
&  \phantom{--} + \frac{\norm[2]{z-(\gamma/2)^{1/2}\defEns{\nabla U(x) + \nabla U(x-\gamma \nabla U(x) + \sqrt{2\gamma} z)}} -\norm[2]{z}}{2}\eqsp.
\end{align}
The analysis of the MALA algorithm is closely related to the study of the ULA algorithm. Indeed, the difference between the two Markov kernels can be expressed for any bounded measurable function $\phi:\rset^d\to\rset$ by
\begin{multline}
\label{eq:diff-rula-rmala}
  \Rmala \phi(x)  - \Rula \phi(x) = \int_{\rset^d}\{\phi(x) - \phi(x-\gamma \nabla U(x) + \sqrt{2 \gamma} z)\} \\
  \times \{1 - \min(1,\rme^{-\alphamala(x,z)}) \} \varphibf(z) \rmd z \eqsp.
\end{multline}
Since $1-\min(1,\rme^{-t}) \leq \abs{t}$ for any $t \in \rset$, properties of ULA can then be transferred to MALA from perturbation arguments achieved by a careful analysis of $\alphamala$ which is the content of the following result.
\begin{lemma}
  \label{lem:bound_alpha_mala_1}
  Assume \Cref{assumption:U-Sinfty} and \Cref{ass:condition_MALA}. Then, for any $\bgamma >0$, there exists $C_{1,\bgamma} < \infty$ such that for any $x,z \in \rset^d$, $\gamma \in \ocint{0,\bgamma}$, it holds
  \begin{equation*}
  \abs{  \alphamala(x,z)} \leq C_{1,\bgamma} \gamma^{3/2}\{1+\norm[4]{z} + \norm[2]{x}\} \eqsp.
  \end{equation*}
\end{lemma}
\begin{proof}
The proof is postponed to \Cref{subsec:geom-ergodicity-mala}.
\end{proof}
A first easy consequence of \eqref{eq:diff-rula-rmala} using \eqref{eq:Eula-def} is that we get for any $\vphi\in\setpoly{\infty}(\rset^d, \rset)$, $\bgamma>0$, $\step\in\ocint{0,\bgamma}$, $\Rmala\vphi = \vphi + \step \generator \vphi + \step^2 \genrmala \vphi$,
with $\genrmala \vphi= \genrula \vphi + \tgenrmala \vphi$ and for any $x \in \rset^d$,
\begin{multline*}
  \tgenrmala \vphi(x) =  \PE \bigg[\step^{-3/2} \defEns{1-\min(1,\rme^{-\alphamala(x,Z)})} \\
  \times \defEns{\int_0^1 \ps{\nabla \vphi(x-t\step\nablaU(x) + t\sqrt{2\step}\nZ)}{\sqrt{\step} \nablaU(x) - \sqrt{2} \nZ} \rmd t} \bigg] \eqsp,
\end{multline*}
where $Z$ is a $d$-dimensional standard Gaussian vector.
Note that by the Lebesgue dominated convergence theorem, for any $\vphi\in\setpoly{\infty}(\rset^d, \rset)$, $\bgamma>0$, $\step\in\ocint{0,\bgamma}$, $  \tgenrmala \vphi$ is continuous. As a result and using \Cref{prop:ULA-dev-ergo} and \Cref{lem:bound_alpha_mala_1}, it follows  that \Cref{ass:dev_generator_discrete} holds.
\begin{lemma}\label{prop:MALA-dev-ergo}
     Assume \Cref{assumption:U-Sinfty} and \Cref{ass:condition_MALA}.
    Then for any $\bgamma >0$, $\{\Rmala \, : \, \gamma \in \ocint{0,\bgamma}\}$ satisfies \Cref{ass:dev_generator_discrete} with $\alpha = 2$, $\genag = \genrmala$ and $\ke = 4$.
\end{lemma}

We now turn to verifying \Cref{ass:geo_ergod}. Similarly to \Cref{prop:MALA-dev-ergo}, a key tool is the decomposition of $\Rmala$ given by \eqref{eq:diff-rula-rmala}.
\begin{proposition}\label{thm:geometric_ergodicity_mala}
  Assume \Cref{assumption:U-Sinfty} and \Cref{ass:condition_MALA}.
  There exist $\bgamma>0$ and $V:\rset^d\to\coint{1,\plusinfty}$ such that \Cref{ass:geo_ergod} is satisfied for  the family of Markov kernels $\{\Rmala \, : \, \gamma \in \ocint{0,\bgamma}\}$. 
\end{proposition}
\begin{proof}
  The proof is postponed to \Cref{subsec:geom-ergodicity-mala}.
\end{proof}
We now turn to the analysis of the RWM algorithm. For $\step>0$, the Markov kernel $\Rrwm$ of the RWM algorithm with a Gaussian proposal of mean $0$ and variance $2\step$ is given for $x\in\rset^d$ and $\msa \in\borelSet(\rset^d)$ by
\begin{multline*}
\Rrwm(x,\msa) = \int_{\rset^d}  \1_{\msa} (x + \sqrt{2 \step} z) \min(1,\rme^{-\alpharwm(x,z)}) \varphibf(z) \rmd z  \\
+ \updelta_x(\msa) \defEns{1-\int_{\rset^d} \min(1,\rme^{-\alpharwm(x,z)})} \varphibf(z) \rmd z \eqsp,
\end{multline*}
where  $\alpharwm(x,z) = \pU(x+\sqrt{2 \step} z) - \pU(x)$. We first consider \Cref{ass:dev_generator_discrete} and adapt the proof of \cite[Lemma 1]{refId0}.
To this end,  consider the following decomposition for any $\vphi \in \setpolyinf(\rset^d,\rset)$,
\begin{multline}
\label{eq:proof-RWM-1}
\Rrwm \vphi (x) - \vphi(x) = \expe{\vphi(x+\sqrt{2\step}\nZ) - \vphi(x)} \\
+ \expe{\parenthese{\min(1,\rme^{-\alpharwm(x, \nZ)}) -1} \defEns{\vphi(x+\sqrt{2\step}\nZ) - \vphi(x)}} \eqsp,
\end{multline}
where $\nZ$ is a standard $d$-dimensional Gaussian vector. While the first term in this decomposition can be easily handled by a Taylor expansion, we rely on the following result for the second term. Define $\zetag:\rset^d \times \rset^d \to \rset$ for all $x,z\in\rset^d$ and $\step\in\ocint{0,\bgamma}$ by,
\begin{equation*}
\zetag(x,z) = 1 - \min\parenthese{1,\exp\defEns{-\alpharwm(x,z)}} - \sqrt{2\step} \ps{\nablaU(x)}{z}_{+}  \eqsp.
\end{equation*}
\begin{lemma}
  \label{lem:approx_rwm}
  Assume \Cref{assumption:U-Sinfty} and \Cref{ass:condition_MALA}. Then, for all $\step\in\ocint{0,\bgamma}$ and $x,z\in\rset^d$,
  \begin{equation*}
    \abs{\zetag(x,z)}
    \leq \gamma \norm[2]{z} \defEnsLigne{L + 2\norm[2]{\nabla U(x)} + 4\step L^2 \norm[2]{z}} \eqsp.
  \end{equation*}
\end{lemma}
\begin{proof}
  First, by a Taylor expansion and since $t \mapsto \max(0,t)$ is $1$-Lipschitz, we get for all $\step\in\ocint{0,\bgamma}$ and $x,z\in\rset^d$
\begin{equation*}
  \absolute{\alpharwm(x,z)_+ - \sqrt{2\step}\ps{\nablaU(x)}{z}_{+}} \leq L\gamma \norm{z}^2 \eqsp,
\end{equation*}
where for any $a \in \rset$, $a_+ = \max(0,a)$.
Using  that that for all $x,z\in\rset^d$, \[ \min\{1,\exp(-\alpharwm(x,z))\} = \exp(-\alpharwm(x,z)_+)\] and
\begin{equation*}
\alpharwm(x,z)_{+} - (1/2) \defEns{\alpharwm(x,z)_+}^2 \leq 1 - \rme^{-\alpharwm(x,z)_{+}} \leq \alpharwm(x,z)_+ \eqsp,
\end{equation*}
concludes the proof.
\end{proof}
Let $\vphi \in \setpolyinf(\rset^d,\rset)$. Using a Taylor expansion, we get for all $x,z\in\rset^d$,
\begin{align}
\nonumber
&\vphi(x+\sqrt{2\step}z) - \vphi(x) \\
&\quad = \sqrt{2\step}\ps{\nabla \vphi(x)}{z} + (2\step)\int_0^1 (1-t)
\nonumber
\DD^2\vphi(x+t\sqrt{2\step}z)[z^{\otimes 2}] \rmd t \\
\nonumber
&\quad =  \sqrt{2\step}\ps{\nabla \vphi(x)}{z} + \step \DD^2\vphi(x)[z^{\otimes 2}] + (\sqrt{2}/3) \step^{3/2} \DD^3\vphi(x)[z^{\otimes 3}] \\
\nonumber
&\phantom{-------------}+  (2/3)\step^2 \int_0^1 (1-t)^3 \DD^4\vphi(x+t\sqrt{2\step}z)[z^{\otimes 4}] \rmd t \eqsp.
\end{align}
In addition, since for any $x \in \rset^d$,
\[ \expe{\ps{\nabla \pU(x)}{\nZ}_{+} \ps{\nabla \vphi(x)}{\nZ}} = (1/2) \ps{\nablaU(x)}{\nabla \vphi(x)} \eqsp, \]
where $Z$ is a standard $d$-dimensional Gaussian vector, we get  that by \eqref{eq:proof-RWM-1} and \Cref{lem:approx_rwm}, for any $\vphi\in\setpoly{\infty}(\rset^d,\rset)$ and $\gamma >0$, $\Rrwm\vphi = \vphi + \step\generator \vphi +\step^{3/2}\genrrwm \vphi$ where for any $x \in \rset^d$,
\begin{align*}
\nonumber
  &\genrrwm\vphi(x) \\
\nonumber
  &= -\PE \bigg[
\int_0^1 (1-t)\DD^2 \vphi(x+\nZ_t)[\nZ^{\otimes 2}] \rmd t
\defEns{2^{3/2}\ps{\nablaU(x)}{\nZ}_{+} + 2\step^{-1/2}\zetag(x,\nZ)} \\
&\phantom{--}+\sqrt{2} \step^{-1}\zetag(x,\nZ) \ps{\nabla \vphi(x)}{\nZ}
-(2/3)\sqrt{\step} \int_0^1 (1-t)^3 \DD^4 \vphi (x+\nZ_t) [\nZ^{\otimes 4}] \rmd t
\bigg] \eqsp,
\end{align*}
where $\nZ_t = t\sqrt{2\step}\nZ$. Then, since $\zetag$ is continuous and using the Lebesgue dominated convergence theorem, we end up with the following result.

\begin{lemma}\label{prop:RWM-dev-ergo}
     Assume \Cref{assumption:U-Sinfty} and \Cref{ass:condition_MALA}.   Then for any $\bgamma >0$, $\set{\Rrwm}{\step \in \ocint{0,\bgamma}}$ satisfies \Cref{ass:dev_generator_discrete} with $\genag = \genrrwm$, $\alpha = 3/2$ and $\ke = 4$.
\end{lemma}

In \Cref{sec:additional-proofs}, we establish a similar result as \Cref{thm:geometric_ergodicity_ula} and \Cref{thm:geometric_ergodicity_mala} for the RWM algorithm.


\section{Numerical experiments}
\label{sec:application_cv}

In this Section, we compare numerically our methodology with the Zero Variance method suggested by \cite{Mira2013}, see \Cref{subsec:bibliography}, that consists in minimizing the marginal variance $\min_{\ControlFunc\in\ControlFuncSet} \invpi(\{\tzf + \generator \ControlFunc\}^2)$ instead of the asymptotic variance $\min_{\ControlFunc\in\ControlFuncSet} \varinf(f + \generator \ControlFunc)$. In \Cref{subsec:numerical-comparison-toy-examples}, we consider a one dimensional example where explicit calculations are possible. In \Cref{subsec:bayesian-examples-numeric}, we study Bayesian logistic and probit regressions.
The code used to run the experiments is  available at \url{https://github.com/nbrosse/controlvariates}.

\subsection{One dimensional example}
\label{subsec:numerical-comparison-toy-examples}

We consider an equally weighted mixture of two Gaussian densities of means $(\mu_1, \mu_2) = (-1, 1)$ and variance $\sigma^2 = 1/2$, and a test function $f(x) = x + x^3 /2 + 3\sin(x)$. The derivative of the Poisson equation \eqref{eq:poisson-eq-langevin} is in such case analytically known: $\sPoic '(x) = -(1/\invpi(x))\int_{-\infty}^{x} \invpi(t) \tzf(t) \rmd t$, see \Cref{subsec:1-2d-numerics-practice} for a practical implementation.

We build a control variate $\ControlFunc_\param \in \ControlFuncSet[\lin] = \defEns{\ps{\param}{\base} : \param\in\rset^\pb}$ where $\base=(\base_1,\ldots,\base_\pb)$ are $\pb$ Gaussian kernels regularly spaced on $\ccint{-4, 4}$, \ie~for all $i\in\defEns{1,\ldots,\pb}$ and $x\in\rset$
\begin{equation*}
  \label{eq:def-basis-gaussian-kernels}
  \base_i(x) = (2\uppi)^{-1/2} \rme^{-(x-\mu_i)^2 / 2} \eqsp, \quad \text{where } \mu_i\in\ccint{-4,4} \eqsp.
\end{equation*}
The optimal parameter $\paramstar\in\rset^\pb$ minimizing the asymptotic variance $\varinf(f + \generator \ControlFunc_\param)$ can be explicitly computed according to \eqref{eq:min-asymp-var-diffusion}. For the Zero Variance estimator, the optimal parameter is given by
\begin{equation}\label{eq:paramzv}
  \paramzv = - \Hzv^{-1} \bzv \eqsp,
\end{equation}
where for $1\leq i,j \leq \pb$, $[\Hzv]_{ij} = \invpi(\ps{\generator \base_i}{\generator \base_j})$ and $[\bzv]_i = \invpi(\tzf \generator \base_i)$.
$\Hzv$ is invertible if $(\generator \base_1,\ldots,\generator\base_\pb)$ are linearly independent in $\setpoly{2}(\rset^d,\rset)$.

The asymptotic variance $\varinf(f + \generator \ControlFunc_\param)$ for the two different parameters, $\paramstar$ and $\paramzv$ are compared against the number of Gaussian kernels $\pb\in\defEns{4,\ldots,10}$ in \Cref{figure:resume_1d}.
Note that the asymptotic variance $\varinf(f)$ is $92.5$. We observe that the variance reduction is better for an even number $\pb$ of basis functions; when $\pb \geq 8$, the two methods achieve an almost identical large variance reduction. These results are supported by the plots of $\ControlFunc_\param '$ and $\generator \ControlFunc_\param$ for $\param\in\defEns{\paramstar, \paramzv}$ in \Cref{figure:resume_1d}, see also \Cref{subsec:1-2d-numerics-practice}.
\begin{figure}
\begin{center}
\includegraphics[scale=0.45]{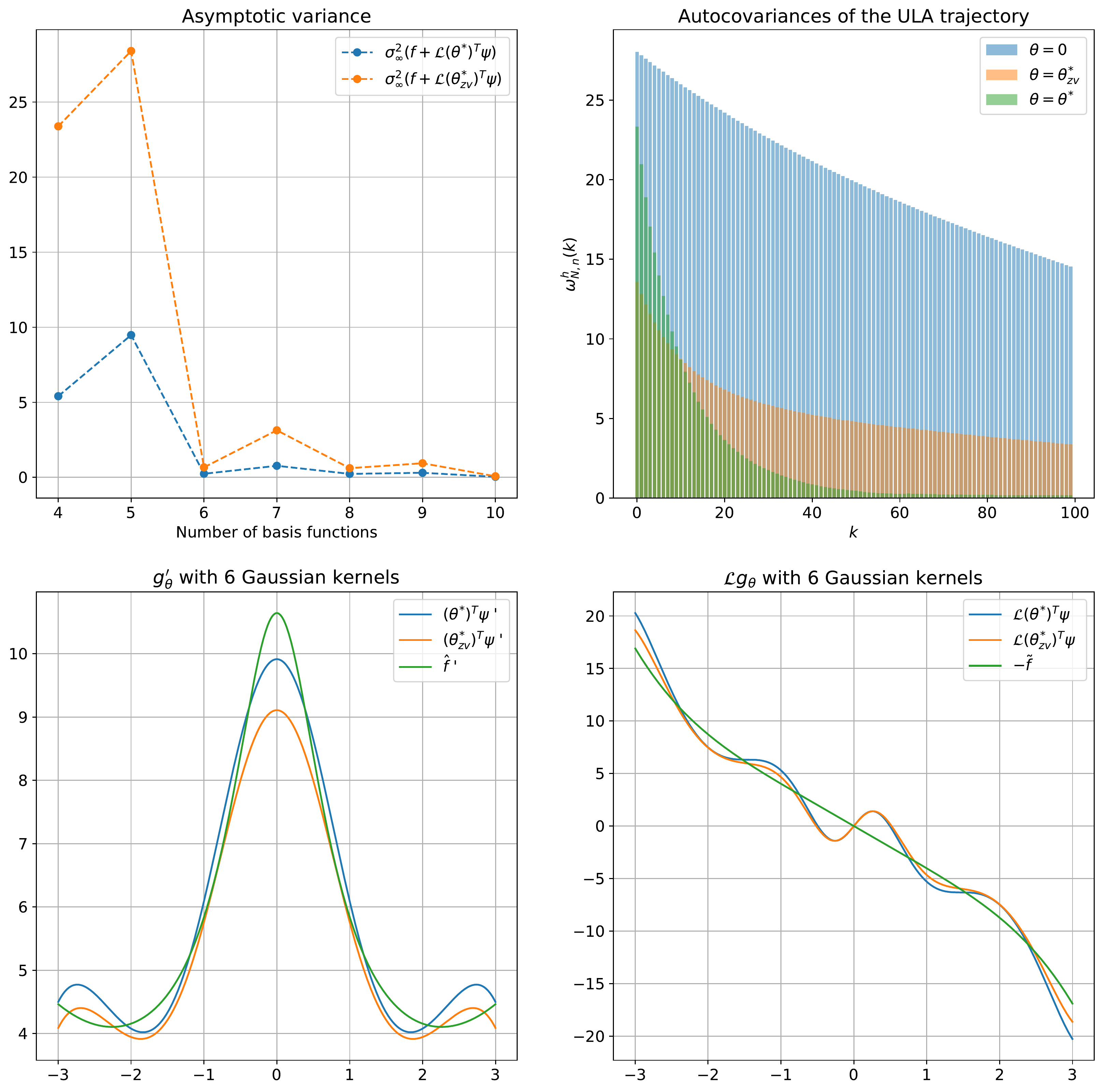}
\end{center}
\caption{\label{figure:resume_1d} \textbf{Top Left.} Plot of the asymptotic variance $\varinf(f + \generator \ControlFunc_\param)$ for $\param\in\defEns{\paramstar, \paramzv}$ and $\pb\in\defEns{4,\ldots,10}$.
\textbf{Top Right.} Autocovariances plot of ULA displaying $\omega^h_{N,n}(k)$ for $h=f +\generator\ps{\param}{\base}$, $\param\in\defEns{0,\paramstar,\paramzv}$ and $0\leq k < 100$.
\textbf{Bottom Left and Right.} Plots of $\ControlFunc_\param '$ and $\generator \ControlFunc_\param$ for $\param\in\defEns{\paramstar, \paramzv}$ and $\pb=6$.}
\end{figure}

We fix the number of basis functions $\pb=4$ and we now turn to the application to MCMC algorithms. We first define the sample mean with a burn-in period $N\in\nset^*$ by
\begin{equation}\label{eq:def-invpihat-N-n}
  \invpihat[N, n](f) = \frac{1}{n}\sum_{k=N}^{N+n-1} f(X_k) \eqsp,
\end{equation}
where $n\in\nset^*$ is the number of samples. In this Section, we consider the following estimators of $\invpi(f)$: $\invpihat[N,n](f + \generator \ps{\paramstar}{\base})$ and $\invpihat[N,n](f + \generator \ps{\paramzv}{\base})$ where $\paramstar$ and $\paramzv$ are given in \eqref{eq:min-asymp-var-diffusion} and \eqref{eq:paramzv} respectively. In this simple one dimensional example, the optimal parameters $\paramstar$ and $\paramzv$ are explicitly computable; the problem of estimating them in higher dimensional models is addressed numerically in \Cref{subsec:bayesian-examples-numeric}.

The sequence $(X_k)_{k\in\nset}$ is generated by the ULA, MALA or RWM algorithms starting at $0$, with a step size $\step=10^{-2}$ for ULA and $\step=5\times 10^{-2}$ for RWM and MALA, a burn-in period $N=10^5$ and a number of samples $n=10^6$.
For a test function $h:\rset\to\rset$ ($h=f + \generator \ps{\param}{\base}$, $\param\in\defEns{0, \paramstar, \paramzv}$), we estimate the asymptotic variance $\varinf[\step](h)$ of $\invpihat[N,n](h)$ by a spectral estimator $\sigS[h]$ with a Tukey-Hanning window, see \cite{flegal:2010}, given by
\begin{align}
\label{eq:def-sigS}
 \sigS[h] &= \sum_{k=-(\floor{n^{1/2}} -1)}^{\floor{n^{1/2}} -1} \defEns{\frac{1}{2} + \frac{1}{2}\cos\parenthese{\frac{\uppi\absolute{k}}{\floor{n^{1/2}}}}} \omega^h_{N,n}(\absolute{k}) \eqsp, \\
 \nonumber
 \omega^h_{N,n}(k) &= \frac{1}{n} \sum_{s=N}^{N+n-1-k} \defEns{h(X_{s}) - \hat{\pi}_{N,n}(h)} \defEns{h(X_{s+k}) - \hat{\pi}_{N,n}(h)} \eqsp.
\end{align}
We compute the average of these estimators $\sigS[f+\generator \ps{\param}{\base}]$, $\param\in\defEns{0,\paramstar,\paramzv}$ over $10$ independent runs of the Monte Carlo algorithm (ULA, RWM or MALA), see \Cref{table:ula-asympt-var-1d}. We observe that minimizing the asymptotic variance improves upon the Zero Variance estimator.
\begin{table}
  \centering
  \begin{tabular}{|c|c|c|c|}
    \hline
    & $\step\sigS[f]$ & $\step\sigS[f+\generator \ps{\paramzv}{\base}]$ & $\step\sigS[f+\generator \ps{\paramstar}{\base}]$ \\
    \hline
    ULA  & $82.06$ & $20.74$ & $5.33$ \\
    \hline
    RWM & $105.2$ & $28.19$ & $8.41$ \\
    \hline
    MALA & $93.27$ & $23.40$ & $5.00$ \\
    \hline
  \end{tabular}
  \caption{Values of $\sigS[f+\generator \ps{\param}{\base}]$, $\param\in\defEns{0,\paramstar,\paramzv}$ rescaled by the step size $\step$.}\label{table:ula-asympt-var-1d}
\end{table}

A more detailed analysis can be carried out using the autocovariances plots that consist in displaying $\omega^h_{N,n}(k)$ for $h=f +\generator\ps{\param}{\base}$, $\param\in\defEns{0,\paramstar,\paramzv}$ and $0\leq k < 100$, see \Cref{figure:resume_1d}. The autocovariances plots for RWM and MALA are similar. By \cite[Theorem 21.2.11]{douc:moulines:priouret:soulier:2018}, the asymptotic variance $\varinf[\step](h)$ is the sum of the autocovariances:
\begin{equation*}
  \varinf[\step](h) = \invpig(\tilde{h}_\step^2) + 2 \sum_{k=1}^{\plusinfty} \invpig(\tilde{h}_\step\Rkerg^k \tilde{h}_\step) \eqsp, \quad \text{where} \eqsp \tilde{h}_\step = h - \invpig(h) \eqsp.
\end{equation*}
The two methods are effective at reducing the autocovariances compared to the case without control variate. The zero variance estimator decreases more the autocovariance at $k=0$ compared to our method, which is indeed the objective of $\paramzv$, the minimizer of $\param\mapsto\invpi((\tzf + \generator \ps{\param}{\base})^2)$. Using $\param = \paramstar$ lowers more effectively the tail of the autocovariances (for $k$ large enough) compared to $\param=\paramzv$. This effect is predominant and explains the results of \Cref{table:ula-asympt-var-1d}.

\subsection{Bayesian logistic and probit regressions}
\label{subsec:bayesian-examples-numeric}

We illustrate the proposed control variates method on Bayesian logistic and probit regressions, see \cite[Chapter 16]{gelman2014bayesian}, \cite[Chapter 4]{marin2007bayesian}.
The examples and the data sets are taken from \cite{papamarkou2014}. Let $\yb=(\yb_1,\ldots,\yb_n)\in\defEns{0,1}^\nb$ be a vector of binary response variables, $\bb\in\rset^d$ be the regression coefficients, and $\xb\in\rset^{\nb \times d}$ be a design matrix.
The log-likelihood for the logistic and probit regressions are given respectively by
\begin{align*}
  \logl{l}(\yb | \bb, \xb) & = \sum_{i=1}^{\nb} \defEns{\yb_i  \xb_i^{\Tr} x -  \ln\parenthese{1+\rme^{\xb_i^{\Tr} \bb}}} \eqsp, \\
  \logl{p}(\yb | \bb, \xb) & = \sum_{i=1}^{\nb} \defEns{\yb_i \ln(\Phi(\xb_i^{\Tr} \bb)) + (1-\yb_i)\ln(\Phi(-\xb_i^{\Tr} \bb))} \eqsp,
\end{align*}
where $\xb_i^{\Tr}$ is the $i^{\text{th}}$ row of $\xb$ for $i\in\defEns{1,\ldots,\nb}$.
For both models,
a Gaussian prior of mean $0$ and variance $\varbb\Id$ is assumed for $\bb$ where $\varbb=100$.
The unnormalized posterior probability distributions $\pib{l}$ and $\pib{p}$ for the logistic and probit regression models  are defined for all $\bb\in\rset^d$ by
\begin{align*}
  \pib{l}(\bb | \yb, \xb) &\propto \exp\parenthese{-\Ub{l}(\bb)} \eqsp \text{with} \quad \Ub{l}(\bb) = -\logl{l}(\yb | \bb, \xb) + (2\varbb)^{-1}\norm[2]{\bb} \eqsp,\\
  \pib{p}(\bb | \yb, \xb) &\propto \exp\parenthese{-\Ub{p}(\bb)} \eqsp \text{with} \quad \Ub{p}(\bb) = -\logl{p}(\yb | \bb, \xb) + (2\varbb)^{-1}\norm[2]{\bb} \eqsp.
\end{align*}
The following lemma enables to check the assumptions on $\Ub{l}$ and $\Ub{p}$ required to apply \Cref{prop:dev-weak-error} for the ULA, MALA and RWM algorithms.
\begin{lemma}
\label{lemma:log-probit-assumptions}
$\Ub{l}$ and $\Ub{p}$ satisfy  \Cref{assumption:U-Sinfty} and \Cref{ass:condition_MALA}.
\end{lemma}
\begin{proof}
  The proof is postponed to \Cref{subsec:proof-log-probit-assumptions}.
\end{proof}

Following \cite[Section 2.1]{papamarkou2014}, we compare two bases for the construction of a control variate,  based on first and second degree polynomials and denoted by $\basea = (\basea_1,\ldots, \basea_d)$ and $\baseb = (\baseb_1,\ldots,\baseb_{d(d+3)/2})$ respectively, see \Cref{sec:suppl-probit-reg} for their definitions.
The estimators associated to $\basea$ and $\baseb$ are  referred to as CV-1 and  CV-2, respectively.

For the ULA, MALA and RWM algorithms, we make a run of $n=10^6$ samples with a burn-in period of $10^5$ samples, started at the mode of the posterior. The step size is set equal to $10^{-2}$ for ULA and to $5 \times 10^{-2}$ for MALA and RWM: with these step sizes, the average acceptance ratio in the stationary regime is equal to 0.23 for RWM and 0.57 for MALA, see \cite{roberts:gelman:gilks:1997,roberts:rosenthal:1998}.
We consider $2d$ scalar test functions $\{f_k\}_{k=1}^{2d}$ defined for all $x\in\rset^d$ and $k\in\{1,\ldots,d\}$ by $f_k(x) = x_k$ and $f_{k+d}(x) = x_k^2$.

Contrary to the one dimensional case handled in \Cref{subsec:numerical-comparison-toy-examples}, the optimal parameters $\paramstar$ and $\paramzv$ corresponding to our method and to the zero variance estimator can not be computed in closed form and must be estimated. We consider then the control variate estimator $\invpicv_ {N,n,n}(f)$ defined in \eqref{eq:def-invpi-cv} where $m=n$ and $(\tilde{X}_k)_{k\in\nset}$ is equal to $(X_k)_{k\in\nset}$; $\paramstar$ is approximated by $\param^*_n$ given in \eqref{eq:def-paramhat-n-star}.
For $k\in\{1,\ldots,2d\}$, we compute the empirical average $\invpihat[N,n](f_k)$ defined in \eqref{eq:def-invpihat-N-n} and confront it to $\invpicv_{N,n,n}(f_k)$.
For comparison purposes, the zero-variance estimators of \cite{papamarkou2014} using the same bases of functions $\basea$, $\baseb$ are also computed and are referred to as ZV-1 for $\basea$ and ZV-2 for $\baseb$.

We run $100$ independent Markov chains for ULA, MALA, RWM algorithms. The boxplots for the logistic example are displayed in \Cref{figure:log-1} for $x_1$ and $x_1^2$. Note the impressive decrease in the variance using the control variates for each algorithm ULA, MALA and RWM. It is worthwhile to note that for ULA, the bias $\absolute{\invpi(\tf) - \invpig(\tf)}$ is reduced dramatically using the CV-2 estimator. It can be explained by the fact that for $n$ large enough, $\ControlFunc_{\param^*_n}= \ps{\param^*_n}{\baseb}$ approximates well the solution $\sPoic$ of the Poisson equation $\generator \sPoic = -\tzf$. We  then get
\begin{equation*}
  \invpig(\tf) + \invpig\parenthese{\generator \ControlFunc_{\param^*_n}} \approx \invpig(\tf) - \invpig\parenthese{\tzf} = \invpi(\tf) \eqsp.
\end{equation*}

To have a more quantitative estimate of the variance reduction, we compute for each algorithm and test function $h\in\setpoly{}(\rset^d, \rset)$, the spectral estimator $\sigS[h]$ defined in \eqref{eq:def-sigS} of the asymptotic variance.
The average of these estimators $\sigS[f+\generator\ps{\param}{\base}]$ for $\param\in\defEns{0,\param^*_n, [\paramzv]_n}$ over the $100$ independent runs of the Markov chains for the logistic regression are reported in Table~\ref{table:1}.
$[\paramzv]_n$ is an empirical estimator of $\paramzv$, see \cite{papamarkou2014} for its construction.
The Variance Reduction Factor (VRF) is defined as the ratio of the asymptotic variances obtained by the ordinary empirical average and the control variate (or zero-variance) estimator. We again observe the considerable decrease of the asymptotic variances using control variates.
In this example, our approach produces slightly larger VRFs compared to the zero-variance estimators.
We obtain similar results for the probit regression;
see \Cref{sec:suppl-probit-reg}.

\begin{figure}
\begin{center}
\includegraphics[scale=0.45]{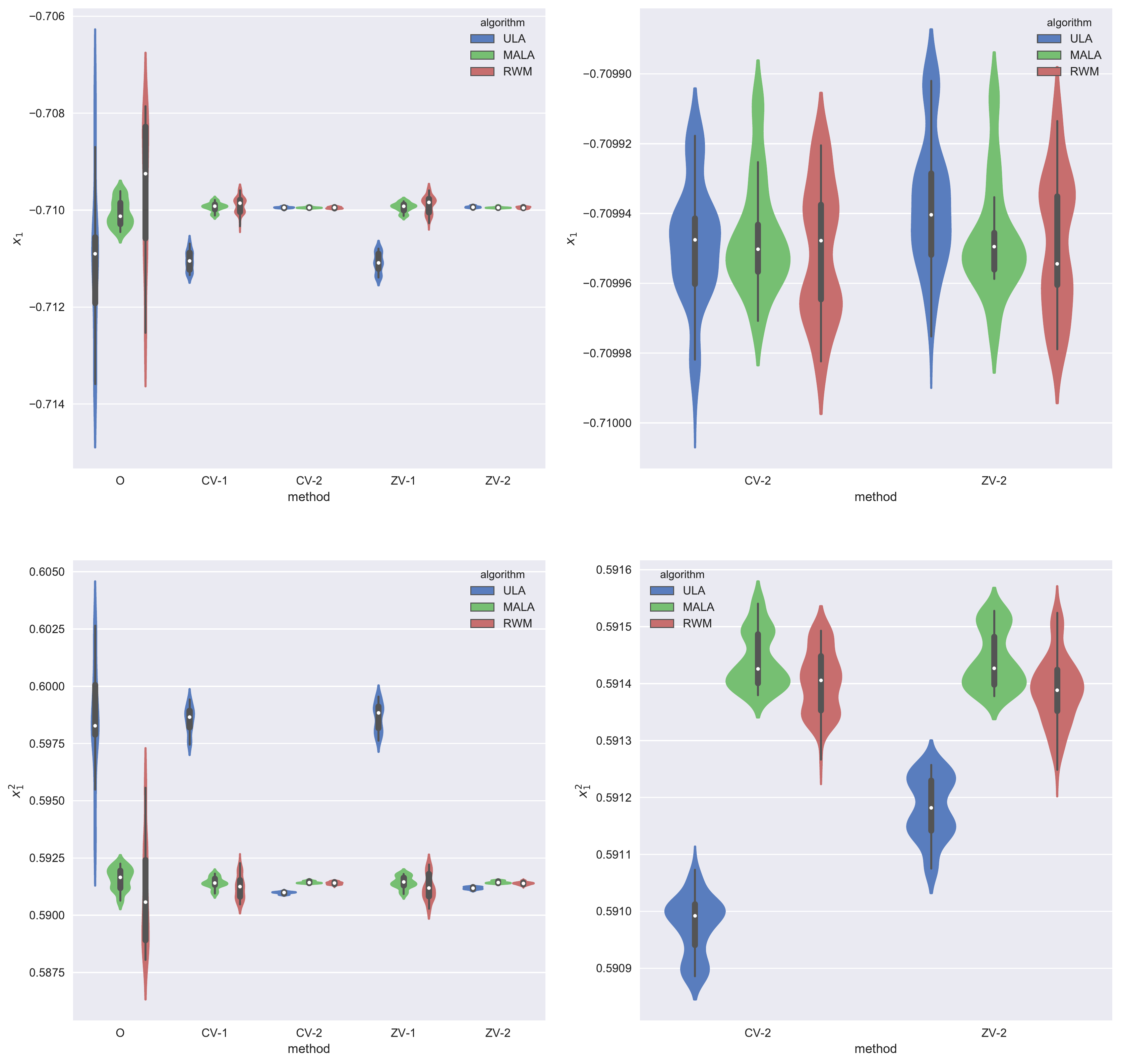}
\end{center}
\caption{\label{figure:log-1} Boxplots of $\bb_1,\bb_1^2$ using the ULA, MALA and RWM algorithms for the logistic regression. The compared estimators are the ordinary empirical average (O), our estimator with a control variate \eqref{eq:def-invpi-cv} using first (CV-1) or second (CV-2) order polynomials for $\base$, and the zero-variance estimators of \cite{papamarkou2014} using a first (ZV-1) or second (ZV-2) order polynomial bases.   The plots in the second column are close-ups for CV-2 and ZV-2.    }
\end{figure}

\begin{table}
\begin{tabular}{c|c|c|c c|c c|c c|c c|}
   \multicolumn{11}{c}{} \\
   & & MCMC & \multicolumn{2}{c|}{CV-1-MCMC} & \multicolumn{2}{c|}{CV-2-MCMC}
   & \multicolumn{2}{c|}{ZV-1-MCMC} & \multicolumn{2}{c|}{ZV-2-MCMC} \\
   & & Var.\  & VRF & Var.\  & VRF & Var.\  & VRF & Var.\  & VRF & Var.\  \\
   \hline
   $\bb_1$ & ULA &         2 &         33 &      0.061 &    3.2e+03 &    6.2e-4 &         33 &      0.061 &      3e+03 &    6.6e-4 \\
   & MALA &      0.41 &         33 &      0.012 &    2.6e+03 &    1.6e-4 &         30 &      0.014 &    2.5e+03 &    1.7e-4 \\
   & RWM &       1.3 &         33 &      0.039 &    2.6e+03 &    4.9e-4 &         32 &       0.04 &    2.7e+03 &    4.8e-4 \\
   \hline
   $\bb_2$ & ULA &        10 &         57 &       0.18 &    8.1e+03 &     1.3e-3 &         53 &       0.19 &    7.4e+03 &     1.4e-3 \\
   & MALA &       2.5 &         59 &      0.042 &    7.7e+03 &    3.2e-4 &         54 &      0.046 &    7.3e+03 &    3.4e-4 \\
   & RWM &       5.6 &         52 &       0.11 &    5.6e+03 &      1.0e-3 &         50 &       0.11 &    5.6e+03 &      1.0e-3 \\
   \hline
   $\bb_2$ & ULA &        10 &         56 &       0.18 &    7.3e+03 &     1.4e-3 &         52 &       0.19 &    6.7e+03 &     1.0e-35 \\
   &MALA &       2.4 &         58 &      0.041 &    6.8e+03 &    3.5e-4 &         52 &      0.045 &    6.5e+03 &    3.7e-4 \\
   &RWM &       5.6 &         45 &       0.13 &    5.1e+03 &     1.0e-31 &         42 &       0.13 &    5.1e+03 &     1.0e-31 \\
   \hline
   $\bb_4$ & ULA &        13 &         26 &        0.5 &    3.9e+03 &     3.3e-3 &         22 &       0.59 &    3.4e+03 &     3.8e-3 \\
   &MALA &       3.1 &         25 &       0.12 &    3.6e+03 &    8.7e-4 &         21 &       0.14 &    3.3e+03 &    9.5e-4 \\
   &RWM &       7.5 &         19 &        0.4 &    2.5e+03 &      3.0e-3 &         18 &       0.43 &    2.4e+03 &     3.0e-31 \\
   \hline
   $\bb_1^2$ & ULA &       4.6 &         10 &       0.46 &    5.5e+02 &     8.4e-3 &        9.3 &       0.49 &    4.8e+02 &     9.5e-3 \\
   &MALA &      0.98 &        9.6 &        0.1 &    4.6e+02 &     2.1e-3 &        8.6 &       0.11 &    4.2e+02 &     2.3e-3 \\
   &RWM &         3 &        8.3 &       0.36 &    4.3e+02 &     6.9e-3 &          8 &       0.37 &    4.3e+02 &     6.9e-3 \\
   \hline
   $\bb_2^2$ & ULA &        29 &         11 &        2.6 &    5.2e+02 &      0.055 &         10 &        2.8 &    4.7e+02 &      0.062 \\
   &MALA &         7 &         11 &       0.64 &    5.2e+02 &      0.013 &         10 &       0.68 &    4.8e+02 &      0.014 \\
   &RWM &        16 &        9.1 &        1.8 &    4.4e+02 &      0.037 &        8.8 &        1.8 &    4.3e+02 &      0.037 \\
   \hline
   $\bb_3^2$ & ULA &        46 &         11 &        4.1 &    6.7e+02 &      0.069 &         10 &        4.5 &    5.9e+02 &      0.079 \\
   &MALA &        11 &         11 &       0.97 &      6e+02 &      0.018 &         10 &          1 &    5.6e+02 &      0.019 \\
   &RWM &        26 &          9 &        2.9 &    4.3e+02 &      0.061 &        8.6 &        3.1 &    4.2e+02 &      0.062 \\
   \hline
   $\bb_4^2$ & ULA &   5.1e+02 &         14 &         37 &    8.2e+02 &       0.62 &         12 &         43 &    6.9e+02 &       0.73 \\
   &MALA &   1.2e+02 &         14 &          9 &    7.9e+02 &       0.15 &         12 &         10 &    7.1e+02 &       0.17 \\
   &RWM &   2.9e+02 &         11 &         27 &    5.8e+02 &       0.51 &         10 &         29 &    5.6e+02 &       0.53 \\
 \hline
\end{tabular}
\caption{Estimates of the asymptotic variances for ULA, MALA and RWM and each parameter $\bb_i$, $\bb_i^2$ for $i\in\{1,\ldots,d\}$, and of the variance reduction factor (VRF) on the example of the logistic regression.}
\label{table:1}
\end{table}

\section{Proofs of \Cref{item-thm-var-3} and \Cref{prop:dev-weak-error}}
\label{sec:proofs}
In the proof
the notation $A(\step,n,x,f) \lesssim B(\step,n,x,f)$ means that there exist $\bgamma > 0$, and $C < \infty$ such that for all  $f \in \setpolyinf(\rset^d,\rset)$, $\step \in \ocint{0,\bgamma}$, $x \in \rset^d$, $n \in \nset$, $A(\step,n,x,f) \leq C B(\step,n,x,f)$.

We preface the proofs by a technical result which follows from \cite[Lemma~2.6, Proposition~2.7]{kopec:2015} and \eqref{eq:def_poisson_int} establishing the regularity of solutions of Poisson's equation.
\begin{proposition}
\label{prop:existence-sol-Poisson}
Assume  \Cref{assumption:U-Sinfty} and \Cref{ass:geo_ergod} and let $k \in \nsets$.
    For all $\tf\in\setpoly{\infty}(\rset^d,\rset)$, there exists $\sPoif\in\setpoly{\infty}(\rset^d,\rset)$ such that
    $\generator \sPoif = -\tzf$, where $\tzf = f- \pi(f)$, $\generator$ is the generator of the Langevin diffusion defined in \eqref{eq:def-generator}. In addition, for all $p\in\nset$, there exist $C\geq 0$, $q\in\nset$ such that  for all $\tf\in\setpoly{\infty}(\rset^d,\rset)$, $\Vnorm[k,q]{\sPoif} \leq C \Vnorm[k,p]{f}$.
\end{proposition}

\subsection{Proof of \Cref{item-thm-var-3}}
\label{subsec:proof:item-thm-var-3}
Let  $p \in \nset$.
Under \Cref{assumption:U-Sinfty} and \Cref{ass:geo_ergod}, by \Cref{prop:existence-sol-Poisson}, there exists $q_1\in\nset$  such that for all $\tf \in\setpolyinf(\rset^d,\rset)$, $\Vnorm[\ke,q_1]{\sPoic} \leq C \Vnorm[\ke,p]{\tf}$, where $\generator \sPoic = -\tzf$, $\tzf = f-\pi(f)$. Under \Cref{ass:dev_generator_discrete}, we have for all $\step\in\ocint{0,\bgamma}$,
\begin{equation}\label{eq:proog-tp2}
   \Rkerg \sPoic = \sPoic + \step \generator \sPoic + \step^{\alpha} \genag \sPoic =
   \sPoic - \step \{ \tf - \invpi(\tf) \} + \step^\alpha \genag \sPoic \eqsp.
\end{equation}
Integrating \eqref{eq:proog-tp2} \wrt~$\invpig$, we obtain that $\invpig(\tf) - \invpi(\tf) = \step^{\alpha-1} \invpig(\genag \sPoic)$.
Under \Cref{ass:dev_generator_discrete}, there exists $q_2 \in \nset$ such that
$\VnormEqs[0, q_2]{\genag \sPoic} \lesssim \VnormEqs[\ke,q_1]{\sPoic}$.
By \Cref{ass:geo_ergod}, we get
$|\invpig(\genag \sPoic)| \leq \invpig(|\genag \sPoic|)  \lesssim \VnormEqs[0, q_2]{\genag \sPoic}$, which concludes the proof.

\subsection{Proof of \Cref{prop:dev-weak-error}}
\label{subsec:proof-weak-error-dev}

The proof is divided into two parts. In the first part which gathers \Cref{lemma:Sn-Sn2-discrete-chain}, \Cref{lemma:tech-step-n-dom} and \Cref{lem:bound_asympto_bias_generator}, we establish preliminary and technical results. In particular, we derive in \Cref{lemma:Sn-Sn2-discrete-chain} an elementary bound on the second order moment of the estimator $\invpihat(f)$ defined in \eqref{eq:def-invpihatn},
 where $(X_k)_{k\in\nset}$ is a Markov chain of kernel $\Rkerg$. The arguments are based solely on the study of $\Rkerg$ and rely on \Cref{ass:geo_ergod}.
In a second part, using our preliminary results, the proof of \Cref{prop:dev-weak-error} is then derived.



\begin{lemma}\label{lemma:Sn-Sn2-discrete-chain}
Assume \Cref{assumption:U-Sinfty} and \Cref{ass:geo_ergod}.
Let $\tf: \rset^d \to \rset$ be such that $\VnormEq[V^{1/2}]{f} < \plusinfty$. For all $n\in\nset^*$,
\begin{equation*}
  \expeMarkov{x, \step}{\parenthese{\sum_{k=0}^{n-1} \defEns{\tf(X_k) - \invpig(\tf)}}^2}  \lesssim \step^{-1} \VnormEq[V^{1/2}]{f}^2 \defEns{n + \step^{-1} \lV(x)} \eqsp.
\end{equation*}
\end{lemma}

\begin{proof}
Note that under \Cref{ass:geo_ergod}-\ref{ass:geo_ergod_iii}, by \cite[Definition D.3.1-(i)]{douc:moulines:priouret:soulier:2018} and Jensen inequality,
  \begin{equation}\label{eq:Vbeta-unif}
     \Vnorm[V^{1/2}]{\updelta_x \RKer_\step^n - \invpig} \lesssim  \rhoVunif^{n\step/2} \lV^{1/2}(x)  \eqsp.
  \end{equation}
We have for all $n\in\nset^*$
\begin{multline}\label{eq:vbeta-unif-1}
  \expeMarkov{x, \step}{\parenthese{\sum_{k=0}^{n-1} \defEns{\tf(X_k) - \invpig(\tf)}}^2}  \\
  \lesssim \sum_{k=0}^{n-1} \sum_{s=0}^{n-1-k}  \expeMarkov{x, \step}{\parenthese{f(X_k) - \invpig(f)}\parenthese{f(X_{k+s}) - \invpig(f)}} \eqsp.
\end{multline}
For $k\in\defEns{0,\ldots,n-1}$ and $s\in\defEns{0,\ldots,n-1-k}$,
\begin{equation*}
   \expeMarkov{x, \step}{\parenthese{f(X_k) - \invpig(f)}\parenthese{f(X_{k+s}) - \invpig(f)}}
   = \expeMarkov{x, \step}{\parenthese{f(X_k) - \invpig(f)}\parenthese{ \Rkerg^s f(X_{k}) - \invpig(f)}} \eqsp.
\end{equation*}
By \eqref{eq:Vbeta-unif}, we obtain
\begin{align*}
  & \absolute{\expeMarkov{x, \step}{\parenthese{f(X_k) - \invpig(f)}\parenthese{f(X_{k+s}) - \invpig(f)}}} \\
  & \phantom{---------}\lesssim \VnormEq[V^{1/2}]{f}\rho^{\step s/2}\expeMarkov{x, \step}{\absolute{f(X_k) - \invpig(f)} V^{1/2}(X_k)}  \\
  & \phantom{---------}\lesssim \VnormEq[V^{1/2}]{f}^2\rho^{\step s/2}\expeMarkov{x, \step}{ V(X_k)} \eqsp,
\end{align*}
using that $V\geq 1$ and $\absolute{f(x) - \invpig(f)} \leq \Vnorm[V^{1/2}]{f}(V^{1/2}(x) + \pibar)$ where $\pibar = \sup_{\step\in\ocint{0,\bgamma}} \invpig(V) \lesssim 1$.
By \eqref{eq:discrete-drift-uniform-bound}, we get
\begin{equation*}
  \absolute{\expeMarkov{x, \step}{\parenthese{f(X_k) - \invpig(f)}\parenthese{f(X_{k+s}) - \invpig(f)}}} \lesssim \VnormEq[V^{1/2}]{f}^2\rho^{\step s/2} \defEns{\rho^{k\step} V(x) + \pibar} \eqsp.
\end{equation*}
Combining it with \eqref{eq:vbeta-unif-1}, we have
\begin{equation*}
  \expeMarkov{x, \step}{\parenthese{\sum_{k=0}^{n-1} \defEns{\tf(X_k) - \invpig(\tf)}}^2}  \lesssim  \frac{\VnormEq[V^{1/2}]{f}^2}{1-\rho^{\step/2}} \defEns{\frac{V(x)}{1-\rho^\step} + n\pibar}  \eqsp.
\end{equation*}
Using that $1-\rhoVunif^{\upbeta \step } \geq \upbeta  \step  \log(1/\rhoVunif) \rhoVunif^{\upbeta \step }$ for all $\upbeta\in\ocint{0,1}$ concludes the proof.
\end{proof}

Define for any $f : \rset^d \to \rset$, $x \in \rset^d$ and
$\gamma \in \ocint{0,\bgamma}$, such that
$R_{\gamma}f^2(x) < \plusinfty$,
\begin{equation*}
  \mmtilf{\step}(x) = \expeMarkov{x,\step}{\{f(X_1) - \Rkerg f(x)\}^2} \eqsp.
\end{equation*}
\begin{lemma}\label{lemma:tech-step-n-dom}
  Assume \Cref{assumption:U-Sinfty}  and \Cref{ass:dev_generator_discrete}.
  For all $\gamma \in \ocint{0,\bgamma}$ and $f\in\setpolyinf(\rset^d,\rset)$, $  \mmtilf{\step} \in \setpolyinfr$ and in addition for all $p\in\nset$ there exists $q \in \nset$ such that for all $\step\in\ocint{0,\bgamma}$, $\Vnorm[0,q]{\mmtilf{\step}} \lesssim \step \Vnorm[\ke,p]{f}^2$.
\end{lemma}

\begin{proof}
Let $p\in\nset$ and $f \in \setpolyinfr$.
By \Cref{ass:dev_generator_discrete}, for all  $\step\in\ocint{0,\bgamma}$ and $x\in\rset^d$,
\begin{align}
\nonumber
  0 \leq \mmtilf{\step}(x) &= \expeMarkov{x,\step}{ \defEns{ f(X_1) - f(x) - \step \generator f(x) - \step^\alpha \genag f(x)}^2} \\
  \nonumber
  &= \expeMarkov{x,\step}{ \defEns{f(X_1) - f(x)}^2} - \step^2 \defEns{\generator f(x) + \step^{\alpha-1} \genag f(x)}^2 \eqsp \\
  \label{eq:tech-step-n-dom-1}
  &\leq \expeMarkov{x,\step}{ \defEns{f(X_1) - f(x)}^2} \eqsp.
\end{align}
Besides, for all $\step\in\ocint{0,\bgamma}$ and $x\in\rset^d$,
\begin{align*}
  &\expeMarkov{x,\step}{ \defEns{f(X_1) - f(x)}^2} = \expeMarkov{x,\step}{ f^2(X_1)} + f^2(x) - 2f(x) \expeMarkov{x,\step}{ f(X_1)} \\
  &\phantom{----}= \step \generator (f^2)(x) + \step^\alpha \genag (f^2)(x) - 2\step f(x) \generator f(x) - 2 \step^\alpha f(x) \genag f(x) \\
  &\phantom{----}= \step \defEns{ 2\norm[2]{\nabla f(x)} + \step^{\alpha -1} \parenthese{\genag (f^2)(x) - 2f(x) \genag f(x)}} \eqsp.
\end{align*}
Then, combining this result and \eqref{eq:tech-step-n-dom-1}, under \Cref{ass:dev_generator_discrete}, $\mmtilf{\step} \in\setpolyinf(\rset^d,\rset)$ and since $\ke \geq 2$, there exists $q\in\nset$ such that $\Vnorm[0,q]{\mmtilf{\step}} \lesssim \step \Vnorm[\ke,p]{f}^2$.
\end{proof}

\begin{lemma}
  \label{lem:bound_asympto_bias_generator}
  Assume \Cref{assumption:U-Sinfty}, \Cref{ass:geo_ergod} and \Cref{ass:dev_generator_discrete}. Then for any $p \in \nset$,
  \begin{align}
    \label{eq:borne-invpig-invpi-1}
        \absolute{\invpig ( \sPoif \generator \sPoif )- \invpi ( \sPoif \generator \sPoif )} &\lesssim \VnormEq[\ke+2,p]{f}^2 \step^{\alpha-1} \eqsp, \\
        \label{eq:borne-sigma-inf}
        \varinf(f) = -2 \invpi(\sPoif \generator \sPoif) &\lesssim \VnormEq[2,p]{f}^2 \eqsp,
      \end{align}
      where for any $f \in \setpolyinfr$, $\sPoif$ is the solution of Poisson's equation \eqref{eq:poisson-eq-langevin} (see \Cref{prop:existence-sol-Poisson}).
\end{lemma}
\begin{proof}
  Let $p \in \nset$.  By \Cref{prop:existence-sol-Poisson} and \Cref{assumption:U-Sinfty}, there
  exists $q \in \nset$ satisfying
  \begin{equation}
 \label{eq:1:lem:bound_asympto_bias_generator}
    \Vnorm[\ke+2,q]{\sPoif} \lesssim \Vnorm[\ke+2,p]{f}  \text{  and  } \Vnorm[\ke+1,q]{U} \lesssim 1 \eqsp.
  \end{equation}
In addition, using \Cref{item-thm-var-3}, we have
  \begin{equation*}
\absLigne{\invpig ( \sPoif \generator \sPoif )- \invpi ( \sPoif \generator \sPoif )} \lesssim  \step^{\alpha-1} \Vnorm[\ke,3 q]{\sPoif \generator \sPoif}  \eqsp.
\end{equation*}
Using that for any $k \in \nset$ and $p_1,p_2 \in \nset$, there exists $C_{k,p_1,p_2} \geq 0$ such that for
any $g_1,g_2 \in \setpolyinfr$,
$\Vnorm[k,p_1+p_2]{fg} \leq C_{k,p_1,p_2}  \Vnorm[k,p_1]{f} \Vnorm[k,p_2]{g}$ by the
general Leibniz rule, we get by definition of $\generator$ \eqref{eq:def-generator},
\begin{equation*}
  \absolute{\invpig ( \sPoif \generator \sPoif )- \invpi ( \sPoif \generator \sPoif )} \lesssim \step^{\alpha-1} \Vnorm[\ke,q]{\sPoif} \Vnorm[\ke,2q]{ \generator \sPoif} \lesssim \step^{\alpha-1} \Vnorm[\ke+2,q]{\sPoif}^2  \Vnorm[\ke+1,q]{  U } \eqsp.
\end{equation*}
The proof of \eqref{eq:borne-invpig-invpi-1} then follows from \eqref{eq:1:lem:bound_asympto_bias_generator}.
Similarly, by \Cref{ass:geo_ergod},
\begin{equation*}
  \varinf(f) = -2 \invpi(\sPoif \generator \sPoif) \lesssim \Vnorm[0, 3q]{\sPoif \generator \sPoif} \lesssim \Vnorm[0,q]{\sPoif} \Vnorm[0,2q]{ \generator \sPoif} \lesssim \Vnorm[2,q]{\sPoif}^2 \Vnorm[1,q]{U} \eqsp,
\end{equation*}
since $\Vnorm[1,q]{U} \leq \Vnorm[\ke+1,q]{U} \lesssim 1$. Using that $\Vnorm[2,q]{\sPoif} \leq \Vnorm[2,p]{f}$ concludes the proof of \eqref{eq:borne-sigma-inf}.
\end{proof}

\begin{proof}[Proof of \Cref{prop:dev-weak-error}]
Let $p \in \nset$.
  For any $f \in \setpolyinfr$, let $\sPoif\in\setpolyinf(\rset^d,\rset)$ be the solution of  Poisson's equation $\generator \sPoif = -\tzf $ (see \Cref{prop:existence-sol-Poisson}).
Using~\Cref{ass:dev_generator_discrete}, we get for all $\step\in\ocint{0,\bgamma}$,
\begin{equation}\label{eq:proof-dev-weak-1}
\Rker_\step \sPoif = \sPoif + \step \generator \sPoif + \step^\alpha \genag \sPoif
= \sPoif - \step \{f-\invpig(f)\} +  \step^\alpha \genag \sPoif - \step \{ \invpig(f)- \invpi(f) \}\eqsp,
\end{equation}
which implies that
\begin{multline}\label{eq:decompo-poisson-12}
\sum_{k=0}^{n-1} \defEns{\tf(X_k) - \invpig(\tf)} = \frac{\sPoif(X_0) - \sPoif(X_n)}{\step} + \frac{1}{\step}\sum_{k=0}^{n-1} \defEns{\sPoif(X_{k+1}) - \Rker_\step \sPoif(X_k)} \\
+ \step^{\alpha-1}\sum_{k=0}^{n-1} \defEns{ \genag \sPoif(X_k) -\step^{1-\alpha}\parenthese{\invpig(\tf)-\invpi(\tf)} } \eqsp.
\end{multline}
Consider the following decomposition based on \eqref{eq:decompo-poisson-12},
\begin{equation*}
n^{-1}\expeMarkov{x,\step}{\parenthese{\sum_{k=0}^{n-1} \defEns{\tf(X_k) - \invpig(\tf)}}^2} = \sum_{i=1}^{4} A^{f}_i(x,n,\step) \eqsp,
\end{equation*}
where,
\begin{align*}
&A^{f}_1(x,n,\step)  \\
& \phantom{--}=  \frac{\step^{2(\alpha-1)}}{n}\expeMarkov{x,\step}{\parenthese{\sum_{k=0}^{n-1} \defEns{ \genag \sPoif(X_k) -\step^{1-\alpha}\parenthese{\invpig(\tf)-\invpi(\tf)}}}^2} \eqsp, \\
&A^{f}_2(x,n,\step) = (n \step^2 )^{-1} \expeMarkov{x,\step}{(\sPoif(X_0) - \sPoif(X_n))^2 } \eqsp , \\
&A^{f}_3(x,n,\step) = (n \step^2)^{-1} \expeMarkov{x,\step}{\parenthese{\sum_{k=0}^{n-1} \sPoif(X_{k+1}) - \Rker_\step \sPoif(X_k)}^2} \eqsp,
\end{align*}
and by Cauchy-Schwarz inequality,
\begin{equation}\label{eq:A4-1}
(1/2) \absolute{A^{f}_4(x,n,\step)} \leq \sum_{1 \leq i < j \leq 3} A^{f}_i(x,n,\step)^{1/2} A^{f}_j(x,n,\step)^{1/2}  \eqsp.
\end{equation}
We bound below $ \absolute{A_i^f(x,n,\step)}$ for any $i\in\defEns{1,\ldots,4}$.
By \Cref{prop:existence-sol-Poisson}, there exists $q_1\in\nset$  such that
\begin{equation}
  \label{eq:control_Poisson_proof_lem_prelim_bootstrap}
  \Vnorm[\ke,q_1]{\sPoif} \lesssim \Vnorm[\ke,p]{f} \eqsp,
\end{equation}
 which combined  with \Cref{ass:geo_ergod}-\ref{ass:geo_ergod_iii} and \eqref{eq:discrete-drift-uniform-bound} yield for all  $n\in\nset^*$,
\begin{equation}\label{eq:A1-1}
  A^{f}_2(x,n,\step)\lesssim \Vnorm[V]{\sPoif^2} \lV(x) / (n\step^2) \lesssim \VnormEq[\ke,p]{f}^2 \lV(x) / (n\step^2) \eqsp.
\end{equation}
For any $\step\in\ocint{0,\bgamma}$, by \eqref{eq:proof-dev-weak-1} and since $\generator \sPoif = -\tzf $, $\invpig(\genag \sPoif) = \step^{1-\alpha}\{\invpig(f)-\invpi(f)\}$.
Under \Cref{ass:dev_generator_discrete}, there exists $q_3\in\nset$ such that for all $\step \in \ocint{0,\bgamma}$,
$\VnormEqs[V^{1/2}]{\genag \sPoif} \lesssim \VnormEqs[0,q_3]{\genag \sPoif}\lesssim \VnormEqs[\ke,q_1]{ \sPoif} \lesssim \VnormEqs[\ke,p]{f}$ by \eqref{eq:control_Poisson_proof_lem_prelim_bootstrap}.
Hence, applying \Cref{lemma:Sn-Sn2-discrete-chain} and using $\alpha \geq 3/2$ yield
\begin{align}
\label{eq:A1-true}
  A^{f}_1(x,n,\step)&\lesssim \frac{\step^{2(\alpha-1)}}{n} \frac{\VnormEqs[\ke,p]{f}^2}{\step}\parenthese{n + \frac{V(x)}{\step}} \\
\nonumber
  &\lesssim
  \VnormEqs[\ke,p]{f}^2\defEns{1 + V(x)/(n\step)} \eqsp.
\end{align}
Since $(\sum_{k=0}^{n-1} \sPoif(X_{k+1}) - \Rker_\step \sPoif(X_k))_{k\in\nset}$ is a $\PP_{x,\step}$-square integrable martingale,  we get that for all $n \in \nset$,
\begin{equation}\label{eq:A2-1}
A^{f}_3(x,n,\step) = \step^{-2} \expeMarkov{x,\step}{n^{-1} \sum_{k=0}^{n-1} \tg(X_k)} \eqsp,
\end{equation}
where
\begin{equation}
  \label{eq:def_tg_proof_lem_preli_bootstrap}
  \tg(x) = \expeMarkov{x,\step}{\{\sPoif(X_1) - \Rker_\step \sPoif(x)\}^2} \eqsp.
\end{equation}
\Cref{lemma:tech-step-n-dom} shows that  $\tg\in\setpolyinf(\rset^d,\rset)$ and that there exists $q_2 \in \nset$ such that $\Vnorm[V]{\tg}\lesssim \Vnorm[0,q_2]{\tg} \lesssim \step \Vnorm[\ke,q_1]{\sPoif}^2 \lesssim \step \Vnorm[\ke,p]{\tf}^2 $. Applying \eqref{eq:def-V-unif}, we get that for all  $n\in\nset^*$,
\begin{multline}\label{eq:A2-1-1}
  \absolute{\expeMarkov{x,\step}{n^{-1} \sum_{k=0}^{n-1} \tg(X_k)} - \invpig(\tg)} \\ \lesssim \VnormEq[V]{\tg} (n\step)^{-1} V(x) \lesssim n^{-1} \VnormEq[\ke,p]{f}^2 V(x) \eqsp.
\end{multline}
We now show that $\invpig(\tg)$ is approximately equal to $\step \varinf(\tf)$.
Observe that by \eqref{eq:def_tg_proof_lem_preli_bootstrap} and since $\pi_{\gamma}$ is invariant for $R_{\gamma}$, for any $\gamma \in \ocint{0,\bgamma}$,
\begin{align}
\nonumber
\invpig(\tg) &= \expeMarkov{\invpig,\step}{\{\sPoif(X_1) - \Rker_\step \sPoif(X_0)\}^2} \\
\label{eq:A2-3}
&= \expeMarkov{\invpig,\step}{\{\sPoif(X_1) - \sPoif(X_0)\}^2} - \expeMarkov{\invpig,\step}{\{\sPoif(X_0) - \Rker_\step \sPoif(X_0)\}^2} \eqsp.
\end{align}
Using that $\invpig$ is the invariant distribution for $\Rker_\step$ again  and  \eqref{eq:proof-dev-weak-1}, we have for any $\gamma \in \ocint{0,\bgamma}$,
\begin{align}
\nonumber
\expeMarkov{\invpig,\step}{\{\sPoif(X_1) - \sPoif(X_0)\}^2}
&= 2 \expeMarkov{\invpig,\step}{\sPoif(X_0)\{\sPoif(X_0) -\Rker_\step \sPoif(X_0)\}} \\
\label{eq:temp-invpig-invpi}
&= -2 \step\invpig ( \sPoif \generator \sPoif ) - 2 \step^{\alpha} \invpig (\sPoif \genag\sPoif) \eqsp.
\end{align}
In the next step, we consider separately the cases  $\invpig = \invpi$ and $\invpig \ne\invpi$.
If $\invpi = \invpig$, then
\begin{equation}
      \label{eq:diff-invpig-invpi_0}
    -\invpig ( \sPoif \generator \sPoif )=(1/2)\varinf(\tf) \eqsp.
\end{equation}
If $\invpig \neq \invpi$, \Cref{lem:bound_asympto_bias_generator} shows that
    \begin{align}
      \label{eq:diff-invpig-invpi}
    \absolute{\invpig ( \sPoif \generator \sPoif )+(1/2)\varinf(\tf)} &=     \absolute{\invpig ( \sPoif \generator \sPoif )- \invpi ( \sPoif \generator \sPoif )}  \\
    \nonumber
    &\lesssim \VnormEq[\ke+2,p]{f}^2 \step^{\alpha-1} \eqsp.
  \end{align}
Using \Cref{ass:dev_generator_discrete}, \eqref{eq:discrete-drift-uniform-bound} and $\absolute{\invpig (\sPoif \genag\sPoif)} \lesssim \Vnorm[\ke,p]{f}^2$ in \eqref{eq:temp-invpig-invpi}, we obtain that
\begin{multline}
  \label{eq:A2-4}
  \Big|\expeMarkov{\invpig,\step}{\{\sPoif(X_1) - \sPoif(X_0)\}^2} +2 \step\invpig ( \sPoif \generator \sPoif )\Big|
  \\ = 2 \step^\alpha \Big| \invpig(\sPoif \genag \sPoif)\Big| \lesssim \VnormEq[\ke,p]{f}^2 \step^{\alpha} \eqsp.
\end{multline}
Similarly, using  \Cref{ass:geo_ergod}-\ref{ass:geo_ergod_ii}, \eqref{eq:discrete-drift-uniform-bound}, \eqref{eq:proof-dev-weak-1}, \eqref{eq:def-generator}, \Cref{ass:dev_generator_discrete} and \eqref{eq:control_Poisson_proof_lem_prelim_bootstrap}, it holds since $\ke \geq 2$ that
\begin{equation*}
  \expeMarkov{\invpig,\step}{\{\sPoif(X_0) - \Rker_\step \sPoif(X_0)\}^2} \lesssim \Vnorm[\ke,q_1]{\sPoif}^2 \step^2 \lesssim \Vnorm[\ke,p]{f}^2 \step^2 \eqsp.
\end{equation*}
Combining this result with  \eqref{eq:diff-invpig-invpi_0} or \eqref{eq:diff-invpig-invpi} and \eqref{eq:A2-4} in  \eqref{eq:A2-3} and using that $ \VnormEq[\ke,p]{f} \leq  \VnormEq[\ke+2,p]{f}$,  we obtain
\begin{equation*}
  \absolute{\invpig(\tg) - \step \varinf(\tf)} \lesssim \VnormEq[\ke+2,p]{f}^2 \step^{\alpha \wedge 2} \eqsp.
\end{equation*}
Plugging this inequality and  \eqref{eq:A2-1-1} in \eqref{eq:A2-1},
we obtain for all $n\in\nset^*$,
\begin{equation}\label{eq:A2-5}
  \absolute{A^{f}_3(x,n,\step) - \step^{-1} \varinf(\tf)} \lesssim \VnormEq[\ke+2,p]{f}^2 \defEns{ \step^{(\alpha-2) \wedge 0} + (n\gamma^2)^{-1} \lV(x)} \eqsp.
\end{equation}
Note that since $\alpha \geq 1$, by \eqref{eq:borne-sigma-inf} and \eqref{eq:A2-5},
\begin{equation*}
  A^{f}_3(x,n,\step) \lesssim \VnormEq[\ke+2,p]{f}^2 \defEns{ \step^{-1} + (n\gamma^2)^{-1} \lV(x)} \eqsp.
\end{equation*}
Combining it with \eqref{eq:A4-1}, \eqref{eq:A1-1} and \eqref{eq:A1-true} conclude the proof.
\end{proof}

\section{Geometric ergodicity for the ULA and MALA algorithms}
\label{sec:geom-ergodicity-mala}

In this Section, we show that  \eqref{eq:def-V-unif} in \Cref{ass:geo_ergod} is satisfied for  the family of Markov kernel $\{\Rula \, : \, \gamma \in \ocint{0,\bgamma}\}$ and $\{\Rmala \, : \, \gamma \in \ocint{0,\bgamma}\}$, with $\bgamma >0$, associated to the ULA and MALA algorithms (see \eqref{eq:def-kernel-ULA} and \eqref{eq:def-kernel-MALA}). Assume that there exist $V \in \rmC^2(\rset^d,\coint{1,\plusinfty})$ and $a>0$ and $b \geq 0$ such that
\begin{equation}
  \label{eq:drift_cont}
  \generator V \leq - a V + b \eqsp.
\end{equation}
Then,
\cite[Theorem 2.2]{roberts:tweedie-Langevin:1996} and \cite[Theorem 4.5]{meyn:tweedie:1993:III} show that $\pi(V) < \plusinfty$ and \eqref{eq:def-V-unif_ii} is satisfied. It is standard to show that \eqref{eq:drift_cont} holds under \Cref{ass:condition_MALA} but this result is given below for completeness.

We begin the proof by two technical lemmas, \Cref{lem:quadratic_behaviour,lem:bounde_pertub_hessian} which are used repeatedly throughout this Section.
In this Section, we assume without loss of generality that $\nabla U(0)=0$. Note that under \Cref{assumption:U-Sinfty} and \Cref{ass:condition_MALA}, $m \leq L$.

\begin{lemma}
  \label{lem:quadratic_behaviour}
  Assume \Cref{assumption:U-Sinfty} and \Cref{ass:condition_MALA}. Then there exists $\raymala_2 \geq 0$ such that for any $x \not \in \ball{0}{\raymala_2}$, $\ps{\nabla U(x)}{x} \geq (m/2) \norm[2]{x}$ and in particular $\norm{\nabla U(x)} \geq (m/2) \norm{x}$.
\end{lemma}
\begin{proof}
  Using \Cref{assumption:U-Sinfty} and \Cref{ass:condition_MALA}, we have for any $x \in \rset^d$, $\norm{x }\geq \raymala_1$,
  \begin{align*}
    \ps{\nabla U(x)}{x}
    &= \int_{0}^{\raymala_1/\norm{x}} \DD^2 U(t x ) [x^{\otimes 2}] \rmd t + \int_{\raymala_1/\norm{x}} ^ 1 \DD^2 U(t x ) [x^{\otimes 2}] \rmd t\\
    & \geq m\norm[2]{x} \{1- \raymala_1 (1 +L/m)   / \norm{x} \}\eqsp,
  \end{align*}
which proves the first statement. The second statement is obvious.
\end{proof}

\begin{lemma}
  \label{lem:bounde_pertub_hessian}
  Assume \Cref{assumption:U-Sinfty} and \Cref{ass:condition_MALA}. Then, for any $t \in \ccint{0,1}$,  $\gamma \in \ocint{0,1/(4L)}$ and $x,z \in \rset^d$, $\norm{z} \leq \norm{x}/(4\sqrt{2\gamma})$, it holds
  \begin{equation*}
    \norm{x+t\{-\gamma \nabla U(x) + \sqrt{2\gamma}z \}} \geq \norm{x}/2 \eqsp.
  \end{equation*}
\end{lemma}

\begin{proof}
  Let $t \in \ccint{0,1}$,  $\gamma \in \ocint{0,1/(4L)}$ and $x,z \in \rset^d$, $\norm{z} \leq \norm{x}/(4\sqrt{2\gamma})$.
  Using the triangle inequality and \Cref{assumption:U-Sinfty}, we have since $t \in \ccint{0,1}$
  \begin{equation*}
    \norm{x+t\{-\gamma \nabla U(x) + \sqrt{2\gamma}z \}} \geq (1-\gamma L ) \norm{x} -\sqrt{2\gamma} \norm{z} \eqsp.
  \end{equation*}
  The conclusion then follows from $\gamma \leq 1/(4L)$ and $\norm{z} \leq \norm{x}/(4\sqrt{2\gamma})$.
\end{proof}

We now show that \eqref{eq:drift_cont} holds.
\begin{proposition}
  \label{propo:drift_cont}
  Assume \Cref{assumption:U-Sinfty} and  \Cref{ass:condition_MALA}. Then, for any $\eta \in \ocint{0, m/8}$, \eqref{eq:drift_cont} holds with $V=V_{\eta}$, $a=2\eta$ and
  \begin{equation*}
    b = 2\eta\exp\parenthese{\eta\defEns{K_2^2 \vee 4(d+1)/m}}\left[d +1 + (2\eta + L)\defEns{K_2^2 \vee 4(d+1)/m} \right] \eqsp,
  \end{equation*}
  where $K_2$ is defined in \Cref{lem:quadratic_behaviour}.
\end{proposition}
\begin{proof}
  Let $\eta \in \ocint{0, m/8}$. By \eqref{eq:def-generator}, for all $x\in\rset^d$,
  \begin{equation*}
    \generator V_{\eta}(x) / (2\eta V_{\eta}(x)) = - \ps{\nabla U(x)}{x} + d + 2\eta\norm[2]{x} \eqsp.
  \end{equation*}
  By \Cref{lem:quadratic_behaviour}, for all $x\in\rset^d$, $x\geq \max(K_2, 2\sqrt{(d+1)/m})$,
  \begin{equation*}
    \generator V_{\eta}(x) / (2\eta V_{\eta}(x)) \leq - \defEns{(m/2) - 2\eta} \norm[2]{x} + d \leq -1 \eqsp,
  \end{equation*}
  which concludes the proof.
\end{proof}

Therefore, to check \Cref{ass:geo_ergod}, it remains to show that for any $\gamma \in \ocint{0,\bgamma}$, for $\bgamma >0$, $\Rula$ (resp. $\Rmala$) has an invariant distribution $\pi_{\gamma}$ (resp. $\pi$) and there exists $\bareta >0$ such that $\pi_{\gamma}(V_{\bareta}) < \plusinfty$ and \eqref{eq:def-V-unif} holds with $V=V_{\bareta}$.

To this end, we establish minorization and drift conditions on $\RKer_\step = \Rula$ and $\RKer_\step= \Rmala$, see \eg~\cite[Chapter~19]{douc:moulines:priouret:soulier:2018} with an explicit dependence with respect to the parameter $\step$.
More precisely, assume that
\begin{enumerate}[label=(\Roman*)]
\item\label{item:condition_ergo_I} there exist $\lambdaFL\in\ooint{0,1}$ and $\cFL<\plusinfty$ such that for all $\step\in\ocint{0,\bgamma}$
  \begin{equation}
    \label{eq:def-discrete-drift}
\RKer_\step \lV_{\bareta} \leq \lambdaFL^\step \lV_{\bareta} + \step \cFL  \eqsp;
\end{equation}
\item\label{item:condition_ergo_II}  there exists $\varepsilon\in\ocint{0,1}$ such that for all $\step\in\ocint{0,\bgamma}$ and $x,x'\in\defEnsLigne{\lV_{\bareta} \leq \widetilde{M}}$,
\begin{equation*}
  \tvnorm{\Rker_\step^{\ceil{1/\step}}(x,\cdot) - \Rker_\step^{\ceil{1/\step}}(x',\cdot)} \leq 2(1 - \varepsilon) \eqsp,
\end{equation*}
where
\begin{equation*}
  \widetilde{M}>\parenthese{\frac{4b \lambda^{-\bgamma}}{\log(1/\lambda)}-1} \vee 1 \eqsp.
\end{equation*}
\end{enumerate}
Then,  \ref{item:condition_ergo_I} implies by \cite[Lemma 1]{durmus:moulines:2015} that for any $\gamma \in \ocint{0,\bgamma}$,
\begin{equation}\label{eq:def-discrete-drift_2}
  \RKer_\step^{\ceil{1/\gamma}} \lV_{\bareta} \leq \lambdaFL \lV_{\bareta} +  \cFL\lambda^{-\bgamma}/\log(1/\lambda)  \eqsp.
\end{equation}
Therefore, applying
\cite[Theorem~19.4.1]{douc:moulines:priouret:soulier:2018} to $\RKer^{\ceil{1/\gamma}}_{\gamma}$ for $\gamma \in \ocint{0,\bgamma}$ using \ref{item:condition_ergo_II} and \eqref{eq:def-discrete-drift_2}, it follows  that \eqref{eq:def-V-unif} holds with $V= \lV_{\bareta}$ and $\pi_{\gamma}(V_{\bareta}) < \plusinfty$. Accordingly, it is enough to show that conditions \ref{item:condition_ergo_I} and \ref{item:condition_ergo_II} hold. This is achieved for ULA in \Cref{propo:super_lyap_ula} and \Cref{propo:small_set_ula} in \Cref{subsec:geom-ergodicity-ula} and relying on these results and the analysis of ULA, the Markov kernel of MALA is shown to fulfill  \ref{item:condition_ergo_I} and \ref{item:condition_ergo_II}  in \Cref{propo:lyap_mala_total} and \Cref{propo:small_set_mala} in \Cref{subsec:geom-ergodicity-mala}.

For ease of notations, we  denote in this Section $\Rmala$ by $\Rkerg$ and $\Rula$ by $\Qgam$ for any $\gamma >0$.


\subsection{Geometric ergodicity for the ULA algorithm}
\label{subsec:geom-ergodicity-ula}

\begin{proposition}
  \label{propo:small_set_ula}
  Assume \Cref{assumption:U-Sinfty}. Then for any $\Rrm \geq 0$,  $x,y \in \rset^d$, $\norm{x}\vee \norm{y} \leq \Rrm$, and $\gamma \in \ocint{0,1/L}$ we have
  \begin{equation*}
    \tvnorm{\updelta_x \Qgam^{\ceil{1/\gamma}} - \updelta_y \Qgam^{\ceil{1/\gamma}}}   \leq 2 (1-\varepsilon) \eqsp.
  \end{equation*}
with  $\varepsilon = 2\Phibf\parenthese{-(1+1/L)^{1/2}(3L)^{1/2}K}$.
\end{proposition}

\begin{proof}
By \Cref{assumption:U-Sinfty} for any $x,y \in \rset^d$,
 \[ \norm[2]{x-y-\gamma\{\nabla U(x) - \nabla U(y)\}} \leq (1+ \gamma \upkappa(\gamma)) \norm[2]{x-y} \]
 where $\upkappa(\gamma) =  (2 L+L^2 \gamma)$. The proof follows from  \cite[Corollary 5]{debortoli2018back}.
\end{proof}

\begin{proposition}
  \label{propo:super_lyap_ula}
  Assume \Cref{assumption:U-Sinfty} and \Cref{ass:condition_MALA} and let $\bgamma \in \ocint{0,m/(4L^2)}$. Then, for any $\gamma \in \ocint{0,\bgamma}$,
  \begin{equation*}
    Q_{\gamma} V_{\bareta}(x) \leq \exp\parenthese{-\bareta m \gamma \norm[2]{x}/4} V_{\bareta}(x) + b_{\bareta} \gamma \1_{\ball{0}{\raymala_3}}(x) \eqsp,
  \end{equation*}
  where $\bareta  = \min(m/16,(8\bgamma)^{-1})$, $\raymala_3 = \max(\raymala_2,4\sqrt{d/m})$, and
  \begin{equation}
  \label{eq:coeffs_super_lyap_mala}
  \begin{aligned}
    b_{\bareta} &= \parentheseDeux{\bareta \defEns{ m/4+   (1+16\bareta\bgamma)(4\bareta + 2 L + \bgamma L^2)} \raymala^2_3 +4 \bareta d  } \\
    &  \qquad \times \exp\parentheseDeux{\bgamma\bareta\defEns{m/4+   (1+16\bareta\bgamma)(4\bareta + 2 L + \bgamma L^2)} \raymala_3^2 + (d/2)\log(2)}\eqsp.
  \end{aligned}
\end{equation}
\end{proposition}

\begin{proof}
  Let $\gamma \in \ocint{0,\bgamma}$.
  First since for any $x \in  \rset^d$, we have
  \begin{multline*}
    \bareta \norm[2]{x-\gamma \nabla U(x) + \sqrt{2\gamma} z} -\norm[2]{z} /2 \\
    = -\frac{1-4\bareta\gamma}{2} \norm[2]{z-\frac{ 2(2\gamma)^{1/2}\bareta}{1-4\bareta\gamma}\{x-\gamma\nabla U(x)\}} + \frac{\bareta}{1-4\bareta\gamma} \norm[2]{x- \gamma \nabla U(x)}  \eqsp,
  \end{multline*}
which implies since $1-4 \bareta\gamma > 0$ that
  \begin{align}
    \nonumber
    \Qgam V_{\bareta}(x) & = (2\uppi)^{-d/2}\int_{\rset^d} \exp\parenthese{    \bareta \norm[2]{x-\gamma \nabla U(x) + \sqrt{2\gamma} z} -\norm[2]{z} /2} \rmd z \\
    \label{eq:1:propo:super_lyap_mala}
    & =  (1-4\bareta\gamma)^{-d/2} \exp\parenthese{ \bareta(1-4\bareta\gamma)^{-1}\norm[2]{x- \gamma \nabla U(x)}} \eqsp.
  \end{align}
  We now distinguish the case when $\norm{x} \geq \raymala_3$ and $\norm{x} < \raymala_3$.

  By \Cref{ass:condition_MALA} and \Cref{lem:quadratic_behaviour}, for any $x \in \rset^d$, $\norm{x} \geq \raymala_3 \geq \raymala_2$, using that $\bareta \leq m/16$ and $\gamma \leq \bgamma \leq m/(4L^2)$, we have
  \begin{multline*}
    (1-4\bareta\gamma)^{-1}  \norm[2]{x- \gamma \nabla U(x)} -\norm[2]{x}\\
    \leq \gamma \norm[2]{x}(1-4\bareta \gamma)^{-1} \parenthese{4\bareta - m + \gamma L^2} \leq -\gamma (m/2) \norm[2]{x} (1-4\bareta\gamma)^{-1}\eqsp.
\end{multline*}
Therefore, \eqref{eq:1:propo:super_lyap_mala} becomes
  \begin{align*}
    \Qgam V_{\bareta}(x)
   & \leq   \exp\parenthese{ -\gamma \bareta (m/2) (1-4\bareta\gamma)^{-1}\norm[2]{x} - (d/2)\log(1-4\bareta\gamma)} V_{\bareta}(x) \\
    & \leq \exp\parenthese{ \gamma \bareta\{- (m/2) \norm[2]{x} + 4 d\}} V_{\bareta}(x)  \eqsp,
  \end{align*}
  where we have used for the last inequality that $-\log(1-t) \leq 2t$ for $t \in \ccint{0,1/2}$ and $4 \bareta \gamma \leq 1/2$. The proof of the statement then follows since $\norm{x} \geq \raymala_3 \geq 4 \sqrt{d/m}$.

  In the case $\norm{x }< \raymala_3$, by \eqref{eq:1:propo:super_lyap_mala}, \Cref{assumption:U-Sinfty} and since $(1-t)^{-1} \leq 1+4t$ for $t\in\ccint{0,1/2}$, we obtain
\begin{align*}
    (1-4\bareta\gamma)^{-1}\norm[2]{x- \gamma \nabla U(x)} - \norm[2]{x}
    &\leq  \gamma    (1-4\bareta\gamma)^{-1}\{4\bareta + 2 L + \gamma L^2\}\norm[2]{x} \\
    &\leq  \gamma    (1+16\bareta\gamma)\{4\bareta + 2 L + \gamma L^2\}\norm[2]{x} \eqsp,
\end{align*}
which implies that
  \begin{multline*}
    \Qgam V_{\bareta}(x)/V_{\bareta}(x) \leq \rme^{-\bareta  m \gamma \norm[2]{x}/4}  \\
    +  \exp\parentheseDeux{ \gamma \bareta \defEns{ m/4+   (1+16\bareta\gamma)(4\bareta + 2 L + \gamma L^2)}\norm[2]{x} -(d/2)\log(1-4\bareta\gamma)} -1 \eqsp.
  \end{multline*}
  The proof is then completed using that for any $t \geq 0$, $\rme^{t} -1 \leq t \rme^{t}$, for any $s \in \ccint{0,1/2}$, $-\log(1-s) \leq 2s$ and $4\bareta\gamma \leq 1/2$.
\end{proof}

\subsection{Geometric ergodicity for the MALA algorithm}
\label{subsec:geom-ergodicity-mala}

We first provide a decomposition in $\gamma$ of $\alphamala$ defined in \eqref{eq:def-alpha-MALA}. For any $x,z \in \rset^d$,  by \cite[Lemma 24]{durmus:moulines:saksman:2017}\footnote{Note that with the notation of \cite{durmus:moulines:saksman:2017}, MALA corresponds to HMC with only one leapfrog step and step size equals to $(2\gamma)^{1/2}$},  we have  that
\begin{equation}
\label{lem:durmus_moulines_saksman}
\alphamala(x,z) = \sum_{k=2}^6 \gamma^{k/2} A_{k,\gamma}(x,z)
\end{equation}
where, setting $x_t = x+t\{-\gamma \nabla U(x) + \sqrt{2\gamma} z \}$,
\begin{align*}
& A_{2,\gamma}(x,z)= 2 \int_0^1 \DD^2 U(x_t) [z^{\otimes 2}] (1/2-t) \rmd t \\
& A_{3,\gamma}(x,z)= 2^{3/2} \int_{0}^1 \DD^2 U(x_t) [z \otimes \nabla U(x)](t-1/4) \rmd t \,, \\
& A_{4,\gamma}(x,z)= -  \int_{0}^1 \DD^2 U(x_t)[ \nabla U(x)^{\otimes 2}] t \rmd t + (1/2) \norm[2]{ \int_{0}^1 \DD^2 U(x_t) [z] \rmd t }  \\
& A_{5,\gamma}(x,z)=   -(1/2)^{1/2}\ps{\int_{0}^1  \DD^2 U(x_t) [\nabla U(x)] \rmd t }{ \int_{0}^1 \DD^2 U(x_t) [z] \rmd t} \\
& A_{6,\gamma}(x,z)= (1/4) \norm[2]{\int_{0}^1 \DD^2 U(x_t) [\nabla U(x)] \rmd t } \eqsp.
\end{align*}
\begin{proof}[Proof of \Cref{lem:bound_alpha_mala_1}]
Since $\int_{0}^1 \DD^2 U(x) [z^{\otimes 2}](1/2-t) \rmd t = 0$, we get setting  $x_t = x + t \{-\gamma \nabla U(x) + \sqrt{2\gamma} z \}$,
\begin{multline}
\label{eq:decomposition-A-2}
A_{2,\gamma}(x,z) \\= \sqrt{\gamma} \iint_0^1 \DD^3 U (s x_t + (1-s) x) \parentheseDeux{z^{\otimes 2} \otimes \{ -\gamma^{1/2} \nabla U(x) + \sqrt{2} z \}} (1/2-t) t \rmd s \rmd t \eqsp.
\end{multline}
The proof follows from $\sup_{x \in \rset^d} \norm{ \DD^2 U(x)} \leq L $ and $\sup_{x \in \rset^d} \norm{ \DD^3 U(x)} \leq M $.
\end{proof}

\begin{lemma}
  \label{lem:bound_alpha_mala_2}
  Assume \Cref{assumption:U-Sinfty} and \Cref{ass:condition_MALA}. Then, for any $\bgamma \in \ocint{0, m^3/(4L^4)}$ there exists $C_{2,\bgamma} < \infty$ such that  for any $\gamma \in \ocint{0,\bgamma}$, $x,z \in \rset^d$ satisfying $\norm{x} \geq \max(2 \raymala_1 , \raymala_2)$ and $\norm{z} \leq \norm{x}/(4 \sqrt{2 \gamma})$, where $\raymala_2$ is defined in \Cref{lem:quadratic_behaviour}, it holds
  \begin{equation*}
    \alphamala(x,z) \leq C_{2,\bgamma} \gamma \norm[2]{z}\{1+\norm[2]{z}\}  \eqsp.
  \end{equation*}
\end{lemma}

\begin{proof}
  Let $\gamma \in \ocint{0,\bgamma}$, $x,z \in \rset^d$ satisfying $\norm{x} \geq \max(2 \raymala_1 , \raymala_2)$ and $\norm{z} \leq \norm{x}/(4 \sqrt{2 \gamma})$.
  Using \eqref{lem:durmus_moulines_saksman}, we get setting \[ A_{4,0,\gamma}(x,z)=  \int_{0}^1 \DD^2 U(x_t)  [\nabla U(x)^{\otimes 2}] t \rmd t \eqsp, \]
  \begin{multline}
    \label{eq:2}
        \alphamala(x,z) \leq 2 \gamma A_{2,\gamma}(x,z) -\gamma^2 A_{4,0,\gamma}(x,z) \\+(2\gamma)^{3/2}L^2 \norm{z}\norm{x} + (\gamma^2/2) L^2 \norm[2]{z} + (\gamma^5/2)^{1/2} L^3 \norm{z}\norm{x} + (\gamma^3/4) L^4 \norm[2]{x} \eqsp,
  \end{multline}
  By \Cref{ass:condition_MALA}, \Cref{lem:quadratic_behaviour} and \Cref{lem:bounde_pertub_hessian}, we get for any $x \in \rset^d$, $\norm{x} \geq \max(2\raymala_1 ,\raymala_2)$,
  \begin{equation}
    \label{eq:3}
    A_{4,0,\gamma}(x,z) \geq (m/2)^3\norm[2]{x}  \eqsp.
  \end{equation}
  Combining this result with \eqref{eq:decomposition-A-2}, \eqref{eq:3} in \eqref{eq:2}, we obtain using $\gamma \leq \bgamma \leq m^3/(4L^4)$
  \begin{align*}
    \alphamala(x,z)
    & \leq 2 \gamma M \defEns{\sqrt{2\gamma}  \norm[3]{z} + \gamma L \norm[2]{z}\norm{x}} -\gamma^2(m^3/2^4) \norm[2]{x} \\
                   & \qquad +(2\gamma)^{3/2}L^2 \norm{z}\norm{x} + (\gamma^2/2) L^2 \norm[2]{z} + (\gamma^5/2)^{1/2} L^3 \norm{z}\norm{x}\eqsp,
  \end{align*}
Since for any $a, b \in \rset^+$ and $\varepsilon > 0$, $ab \leq  (\epsilon/2) a^2 + 1/(2\epsilon) b^2$,  we obtain
  \begin{align*}
    &    \alphamala(x,z) \leq \gamma \norm[2]{z} \Big\{ 2^{1/2} L^2 \varepsilon^{-1} + (\gamma/2)L^2 + 2^{-3/2}  \gamma^{3/2} L^3 \varepsilon^{-1} \\
    &\qquad \qquad \qquad \qquad \qquad + (2^3\gamma)^{1/2} M \norm{z} + \gamma M L \varepsilon^{-1}\norm[2]{z} \Big\} \\
&\qquad \qquad \qquad     + \norm[2]{x} \gamma^2 \parentheseDeux{\varepsilon\defEns{L M  +2^{1/2} L^2 + 2^{-3/2} \bgamma^{1/2} L^3 } -m^3/2^4} \eqsp.
  \end{align*}
  Choosing $\varepsilon = (m^3/2^4) \defEnsLigne{L M  +2^{1/2} L^2 + 2^{-3/2} \bgamma^{1/2} L^3}^{-1}$ concludes the proof.
\end{proof}

\begin{lemma}
  \label{propo:lyap_mala}
  Assume \Cref{assumption:U-Sinfty}, \Cref{ass:condition_MALA} and let $\bgamma\in\ocint{0,m/(4L^2)}$. Then, for any $\gamma \in \ocint{0,\bgamma}$ and $x\in\rset^d$,
  \begin{equation*}
    \int_{\rset^d} \norm[2]{y} Q_{\gamma}(x,\rmd y) \leq \defEns{1-(m\gamma)/2} \norm[2]{x}+ \tildeb \gamma \1_{\ball{0}{\raymala_4}}(x) \eqsp,
  \end{equation*}
  where $Q_{\gamma}$ is the Markov kernel of ULA defined in \eqref{eq:def-kernel-ULA},
  \begin{equation*}
    \raymala_4 = \max\parenthese{\raymala_2, 2\sqrt{(2d)/m}} \eqsp, \quad
    \tildeb = 2d + \raymala_4^2 \parenthese{\bgamma L^2 + 2L + m/2} \eqsp.
  \end{equation*}
\end{lemma}

\begin{proof}
  Let $\gamma\in\ocint{0,\bgamma}$ and $x\in\rset^d$. By \Cref{assumption:U-Sinfty}, we have
  \begin{equation*}
    \int_{\rset^d} \norm[2]{y} Q_{\gamma}(x,\rmd y) \leq
    2\gamma d + \norm[2]{x}(1+\gamma^2 L^2) - 2\gamma\ps{\nabla U(x)}{x} \eqsp.
  \end{equation*}
  We distinguish the case when $\norm{x} \geq \raymala_4$ and $\norm{x} < \raymala_4$.
  If $\norm{x} \geq \raymala_4 \geq \raymala_2$, by \Cref{lem:quadratic_behaviour}, and since $\gamma \leq \bgamma \leq m/(4L^2)$, $\norm{x} \geq \raymala_4 \geq 2\sqrt{(2d)/m}$,
  \begin{align*}
    \int_{\rset^d} \norm[2]{y} Q_{\gamma}(x,\rmd y) &\leq
    \norm[2]{x} \parentheseDeux{1-\gamma \defEns{m - \gamma L^2 - (2d)/\norm[2]{x}}} \\
    &\leq \norm[2]{x} \defEns{1-\gamma m /2} \eqsp.
  \end{align*}
  If $\norm{x} < \raymala_4$, we obtain
  \begin{equation*}
    \int_{\rset^d} \norm[2]{y} Q_{\gamma}(x,\rmd y) \leq
    \norm[2]{x} \defEns{1-\gamma m /2} +
    \gamma \norm[2]{x} \parenthese{\gamma L^2 + 2 L + m/2} + 2\gamma d \eqsp,
  \end{equation*}
  which concludes the proof.
\end{proof}

\begin{lemma}
  \label{lem:diff_tv_MALA_ULA}
  Assume \Cref{assumption:U-Sinfty} and \Cref{ass:condition_MALA} and let $\bgamma\in\ocint{0,m/(4L^2)}$. Then, there exist $C_{3,\bgamma},C_{4,\bgamma}\geq 0$ such that for any $x \in \rset^d$ and $\gamma \in \ocint{0,\bgamma}$, we have
  \begin{align}
    \label{eq:1:lem:diff_tv_MALA_ULA}
      \tvnorm{\updelta_x \Qgam - \updelta_x \Rkerg} & \leq C_{3,\bgamma} \gamma^{3/2} (1+\norm[2]{x}) \eqsp, \\
    \label{eq:2:lem:diff_tv_MALA_ULA}
    \tvnorm{\updelta_x \Qgam^{\ceil{1/\gamma}} - \updelta_x \Rkerg^{\ceil{1/\gamma}}} &\leq C_{4,\bgamma} \gamma^{1/2} (1+\norm[2]{x}) \eqsp.
  \end{align}
\end{lemma}
\begin{proof}
  Let $x \in \rset^d$  and $\gamma \in \ocint{0,\bgamma}$.
We first show that \eqref{eq:1:lem:diff_tv_MALA_ULA} holds and then use this result to prove \eqref{eq:2:lem:diff_tv_MALA_ULA}.
Let $f : \rset^d \to \rset$ be a bounded and measurable function. Then, by \eqref{eq:def-kernel-ULA} and \eqref{eq:def-kernel-MALA}, we have
\begin{align*}
  &\abs{\Qgam f(x) - \Rkerg f(x)} \\
  & \qquad  = \Big| \int_{\rset^d}\{f(x-\gamma \nabla U(x) + \sqrt{2 \gamma} z) - f(x)\} \\
  &\phantom{----------} \times \{1 - \min(1,\rme^{-\alphamala(x,z)}) \} \varphibf(z) \rmd z \Big| \\
  & \qquad  \leq 2 \norm{f}_{\infty} \int_{\rset^d} \abs{1 - \min(1,\rme^{-\alphamala(x,z)}) } \varphibf(z) \rmd z  \\
  &\qquad \leq  2 \norm{f}_{\infty} \int_{\rset^d} \abs{\alphamala(x,z)}  \varphibf(z) \rmd z \eqsp.
\end{align*}
The conclusion of \eqref{eq:1:lem:diff_tv_MALA_ULA} then follows from an application of \Cref{lem:bound_alpha_mala_1}.

We now turn to the proof of \eqref{eq:2:lem:diff_tv_MALA_ULA}. Consider the following decomposition
  \begin{equation*}
    \updelta_x \Qgam^{\ceil{1/\gamma}} - \updelta_x \Rkerg^{\ceil{1/\gamma}} = \sum_{k=0}^{\ceil{1/\gamma}-1} \updelta_x \Qgam^k \{\Qgam - \Rkerg\} \Rkerg^{\ceil{1/\gamma}-k-1} \eqsp.
  \end{equation*}
  Therefore using the triangle inequality, we obtain that
  \begin{equation}
            \label{eq:4}
        \tvnorm{\updelta_x \Qgam^{\ceil{1/\gamma}} - \updelta_x \Rkerg^{\ceil{1/\gamma}}} \leq \sum_{k=0}^{\ceil{1/\gamma}-1} \tvnorm{ \updelta_x \Qgam^k \{\Rkerg - \Qgam\} \Rkerg^{\ceil{1/\gamma}-k-1}} \eqsp.
      \end{equation}
      We now bound each term in the sum. Let $k \in \{0,\ldots,\ceil{1/\gamma}-1\}$ and $f : \rset^d \to \rset$ be a bounded and measurable function. By   \eqref{eq:1:lem:diff_tv_MALA_ULA}, we obtain that
      \[ \abs{ \updelta_x \{\Rkerg - \Qgam\} \Rkerg^{\ceil{1/\gamma}-k-1} f} \leq C_{3,\bgamma} \norm{f}_{\infty} \gamma^{3/2} \{1+\norm[2]{x}\} \]
      and therefore using \Cref{propo:lyap_mala}, we get
      \begin{equation*}
\abs{       \updelta_x \Qgam^k   \{\Rkerg - \Qgam\} \Rkerg^{\ceil{1/\gamma}-k-1} f} \leq C_{3,\bgamma} \norm{f}_{\infty} \gamma^{3/2} \{1+(1-m\gamma/2)^k \norm[2]{x} + 2\tildeb/m \} \eqsp.
     \end{equation*}
     Plugging this result in         \eqref{eq:4}, we obtain
     \begin{align*}
       \tvnorm{\updelta_x \Qgam^{\ceil{1/\gamma}} - \updelta_x \Rkerg^{\ceil{1/\gamma}}} &\leq C_{3,\bgamma} \gamma^{3/2} \sum_{k=0}^{\ceil{1/\gamma}-1}   \{1+(1-m\gamma/2)^k \norm[2]{x} + 2\tildeb/m \} \\
       &\leq C_{3,\bgamma} \gamma^{1/2} \{1+2(\norm[2]{x}+\tildeb)/m\} \eqsp,
     \end{align*}
     which concludes the proof.
\end{proof}

\begin{proposition}
  \label{propo:small_set_mala}
  Assume \Cref{assumption:U-Sinfty} and \Cref{ass:condition_MALA}. Then for any $\Rrm \geq 0$ there exist $\bgamma > 0$ and $\varepsilon >0$, such that for any $x,y \in \rset^d$, $\norm{x}\vee \norm{y} \leq \Rrm$, and $\gamma \in \ocint{0,\bgamma}$ we have
  \begin{equation}
\label{eq:small_set_mala_propo}
    \tvnorm{\updelta_x \Rkerg^{\ceil{1/\gamma}} - \updelta_y \Rkerg^{\ceil{1/\gamma}}}   \leq 2 (1-\varepsilon) \eqsp.
  \end{equation}
\end{proposition}

\begin{proof}
First note that for any $x,y \in \rset^d$, $\gamma >0$, by the triangle inequality, we obtain
  \begin{multline}
    \label{eq:decomposiiton_small_set_mala}
       \tvnorm{\updelta_x \Rkerg^{\ceil{1/\gamma}} - \updelta_y \Rkerg^{\ceil{1/\gamma}}} \leq        \tvnorm{\updelta_x \Rkerg^{\ceil{1/\gamma}} - \updelta_x \Qgam^{\ceil{1/\gamma}}}\\ + \tvnorm{\updelta_x \Qgam^{\ceil{1/\gamma}} - \updelta_y \Qgam^{\ceil{1/\gamma}}} + \tvnorm{\updelta_y \Rkerg^{\ceil{1/\gamma}} - \updelta_y \Qgam^{\ceil{1/\gamma}}} \eqsp.
     \end{multline}
     We now give some bounds for each term on the right hand side for any $x,y \in \rset^d$,  $\norm{x}\vee \norm{y} \leq \Rrm$ for a fixed $\Rrm \geq 0$ and $\gamma \leq  1/L$.
     By \Cref{propo:small_set_ula}, there exists $\varepsilon_1>0$ such that for any $x,y \in \rset^d$,  $\norm{x}\vee \norm{y} \leq \Rrm$ and $\gamma \leq  1/L$,
     \begin{equation}\label{eq:bound_small_ULA_proof_small_MALA}
       \tvnorm{\updelta_x \Qgam^{\ceil{1/\gamma}} - \updelta_y \Qgam^{\ceil{1/\gamma}}} \leq
       2(1-\varepsilon_1) \eqsp.
     \end{equation}
  In addition, by \Cref{lem:diff_tv_MALA_ULA}, there exists $C\geq 0$ such that for any $\gamma \in \ocint{0,m/(4L^2)}$, and $z \in \rset^d$, $\norm{z} \leq \Rrm$,
  \begin{equation*}
    \tvnorm{\updelta_z \Qgam^{\ceil{1/\gamma}} - \updelta_z \Rkerg^{\ceil{1/\gamma}}} \leq C \gamma^{1/2}(1+\Rrm^2) \eqsp.
  \end{equation*}
  Combining this result with \eqref{eq:bound_small_ULA_proof_small_MALA} in \eqref{eq:decomposiiton_small_set_mala}, we obtain that for any $x,y\in \rset^d$, $\norm{x}\vee\norm{y} \leq \Rrm$, $\gamma \in \ocint{0,m/(4L^2)}$,
  \begin{equation*}
    \norm{\updelta_x \Rkerg^{\ceil{1/\gamma}} - \updelta_y \Rkerg^{\ceil{1/\gamma}}} \leq 2(1-\varepsilon_1) + 2 C\gamma^{1/2}(1+\Rrm^2) \eqsp.
  \end{equation*}
  Therefore, we obtain that for any $x,y \in \rset^d$, $\norm{x} \vee \norm{y} \leq \Rrm$, $\gamma \in \ocint{0,\bgamma}$, \eqref{eq:small_set_mala_propo} holds with $\varepsilon \leftarrow \varepsilon_1/2$ taking
  \begin{equation*}
    \bgamma = m/(4L^2) \wedge \parentheseDeux{\varepsilon_1^{2}\parenthese{2C(1+\Rrm^2)}^{-2}} \eqsp.
  \end{equation*}
\end{proof}

\begin{lemma}
  \label{lem_tail_chi2}
Let $\bgamma >0$ and $\gamma \in \ocint{0,\bgamma}$. Then, for any $x \in \rset^d$, $\norm{x} \geq 20\sqrt{2\bgamma d}$,
\begin{equation*}
  \int_{\rset^d \setminus \ball{0}{\norm{x}/(4\sqrt{2\gamma})}} \varphibf(z) \rmd z \leq \exp(-\norm{x}^2/(128\gamma)) \eqsp.
\end{equation*}
\end{lemma}

\begin{proof}
Let $x>0$.
By \cite[Lemma 1]{laurent:massart:2000},
\begin{equation*}
  \PP(\norm[2]{Z} \geq 2\{\sqrt{d} + \sqrt{x}\}^2) \leq
  \PP(\norm[2]{Z} \geq d + 2 \sqrt{dx} + 2x) \leq
  \rme^{-x} \eqsp,
\end{equation*}
where $Z$ is a $d$-dimensional standard Gaussian vector.
Setting $t=2\{\sqrt{d} + \sqrt{x}\}^2$, we obtain
\begin{equation*}
  \PP(\norm[2]{Z} \geq t) \leq \exp\parenthese{-\defEns{d + t/2 - \sqrt{2td}}} \eqsp,
\end{equation*}
and for $\sqrt{t} \geq 5\sqrt{d}$, we get $\PP(\norm{Z} \geq \sqrt{t}) \leq \rme^{-t/4}$ which gives the result.
\end{proof}

\begin{proposition}
  \label{propo:lyap_mala_total}
  Assume \Cref{assumption:U-Sinfty} and \Cref{ass:condition_MALA}. There exist $\bgamma>0$, $\tildem>0$, and $\raymala_5,\barb \geq 0$ such that for any $\gamma \in \ocint{0,\bgamma}$ and $x\in\rset^d$,
  \begin{equation*}
     R_{\gamma}V_{\bareta}(x) \leq (1-\tildem \gamma)V_{\bareta}(x)+ \barb \gamma \1_{\ball{0}{\raymala_5}}(x) \eqsp,
  \end{equation*}
  where $R_{\gamma}$ is the Markov kernel of MALA defined by \eqref{eq:def-kernel-MALA} and $\bareta$ is given by \eqref{eq:coeffs_super_lyap_mala}.
\end{proposition}

\begin{proof}
  Let $\bgamma_1 = m/(4L^2)$. By \eqref{eq:diff-rula-rmala} and \Cref{propo:super_lyap_ula}, for any $\gamma \in \ocint{0,\bgamma_1}$ and $x  \in \rset^d$,
  \begin{align*}
    \Rkerg V_{\bareta}(x)
&\leq \Qgam V_{\bareta}(x) + V_{\bareta}(x)\int_{\rset^d} \{1-\min(1,\rme^{-\alphamala(x,z)}\} \varphibf(z) \rmd z\\
  & \leq \rme^{-\bareta m \gamma \norm[2]{x}/4}  V_{\bareta}(x) + b_{\bareta} \gamma \1_{\ball{0}{\raymala_3}}(x) \\
  &\phantom{----}+ V_{\bareta}(x)\int_{\rset^d} \{1-\min(1,\rme^{-\alphamala(x,z)}\} \varphibf(z) \rmd z \eqsp,
  \end{align*}
where $\raymala_3$ and $b_{\bareta}$ are given in \eqref{eq:coeffs_super_lyap_mala}.
Let
\begin{equation*}
  \bgamma_2 = \min\parenthese{1, \bgamma_1, m^3/(4L^4)} \eqsp, \quad
  \Rrm_1 = \max\parenthese{1, 2 \raymala_1 , \raymala_2, \raymala_3, 20\sqrt{2\bgamma_2 d}} \eqsp.
\end{equation*}
Then,  by \Cref{lem:bound_alpha_mala_2} and \Cref{lem_tail_chi2}, there exist $C_{1}\geq 0$ such that for any $x \in \rset^d$, $\norm{x} \geq \Rrm_1$ and
 $\gamma \in \ocint{0,\bgamma_2}$,
  \begin{align*}
    \Rkerg V_{\bareta}(x) & \leq  \rme^{-\bareta m \gamma \norm[2]{x}/4}  V_{\bareta}(x) + V_{\bareta}(x) \defEns{C_1 \gamma  + \exp(-\norm[2]{x}/(128\gamma))} \\
    & \leq  \rme^{-\bareta m \gamma \norm[2]{x}/4}  V_{\bareta}(x) + V_{\bareta}(x) \defEns{C_1 \gamma  + \exp(-1/(128\gamma))} \eqsp.
  \end{align*}
  Using that there exists $C_2 \geq 0$ such that $\sup_{t \in \ooint{0,1}} \{t^{-1} \exp(-1/(128 t))\} \leq C_2$ we get  for any $x \in \rset^d$, $\norm{x} \geq \Rrm_1$, $\gamma \in \ocint{0,\bgamma_2}$,
  \begin{equation*}
    \Rkerg V_{\bareta}(x) \leq  \rme^{-\bareta m \gamma \norm[2]{x}/4}  V_{\bareta}(x) + V_{\bareta}(x) \gamma \defEns{C_1   + C_2} \eqsp.
  \end{equation*}
  Let
  \begin{equation*}
    \Rrm_2 = \max\parenthese{\Rrm_1, 4(C_1 + C_2)^{1/2} (\bareta m)^{-1/2}} \eqsp, \quad
    \bgamma_3 = \min\parenthese{\bgamma_2, 4\defEns{m \bareta \Rrm_2^2}^{-1}} \eqsp.
  \end{equation*}
  Then, since for any $t \in \ccint{0,1}$, $\rme^{-t} \leq 1-t/2$, we get for any $x \in \rset^d$, $\norm{x} \geq \Rrm_2$, $\gamma \in \ocint{0,\bgamma_3}$,
  \begin{align}
    \nonumber
    \Rkerg V_{\bareta}(x)& \leq  \rme^{-\bareta m \gamma \Rrm_2^2 /4}  V_{\bareta}(x) + V_{\bareta}(x) \gamma \defEns{C_1   + C_2} \\
    \nonumber
    & \leq \parentheseDeux{1-\gamma\defEns{\bareta m \Rrm_{2}^2 /8 -C_1-C_2}} V_{\bareta}(x) \\
    \label{eq:drift_mala_totla_2}
    & \leq \defEns{1-\gamma \bareta m \Rrm_{2}^2 /16} V_{\bareta}(x) \eqsp.
  \end{align}
  In addition, by \Cref{lem:bound_alpha_mala_1}, using that for any $t \in \rset$, $1-\min(1,\rme^{-t}) \leq \abs{t}$, there exists $C_3\geq 0$ such that for any  $x \in \rset^d$, $\norm{x} \leq \Rrm_2$ and $\gamma \in \ocint{0,\bgamma_3}$,
  \begin{align*}
    \Rkerg V_{\bareta}(x) & \leq V_{\bareta}(x) +b_{\bareta} \gamma \1_{\ball{0}{\raymala_3}}(x) + C_3 \gamma^{3/2} \int_{\rset^d} \{1+\norm[2]{x} + \norm[4]{z}\} \varphibf(z) \rmd z \\
    & \leq (1-\gamma \bareta m \Rrm_{2}^2 /16) V_{\bareta}(x) +  \gamma \bareta m \Rrm_2^2 \rme^{\bareta \Rrm_2^2} / 16 + \gamma b_{\bareta} \\
    &\phantom{-------------}+ C_3 \gamma \bgamma_3^{1/2} \defEns{1 + \Rrm^2_2 + C_4} \eqsp,
  \end{align*}
  where $C_4 = \int_{\rset^d} \norm[4]{z} \varphibf(z) \rmd z$.
  Combining this result and \eqref{eq:drift_mala_totla_2} completes the proof.

\end{proof}


\bibliographystyle{alpha}
\bibliography{bibliographie}


\clearpage

\appendix

\section{Random Walk Metropolis (RWM) algorithm}
\label{sec:additional-proofs}
We show \eqref{eq:def-discrete-drift} for the RWM algorithm. 
For that purpose, consider the following additional assumption on $\pU$.
\begin{assumptionS}\label{assumption:U-dom-drift-RWM}
There exist $\chirwm, \widetilde{K}>0$ such that for all $x\in\rset^d$, $\norm{x} \geq \widetilde{K}$,
\begin{equation*}
\norm{\nablaU(x)} \geq \chirwm^{-1}\eqsp, \quad
\norm{\DD^3 \pU(x)} \leq \chirwm  \norm{\DD^2 \pU(x)}\eqsp , \quad
\norm{\DD^2 \pU(x)} \leq \chirwm \norm{\nablaU(x)}
\end{equation*}
and $\lim_{\norm{x}\to\plusinfty} \norm{\DD^2 \pU(x)} / \norm[2]{\nablaU(x)} = 0$.
\end{assumptionS}
\begin{lemma}\label{prop:RWM-drift}
Assume that $U\in\setpoly{3}(\rset^d,\rset)$ and \Cref{assumption:U-dom-drift-RWM}.
There exists $\bgamma>0$ such that for all $\step\in\ocint{0,\bgamma}$, $\Rrwm$ satisfies the drift condition \eqref{eq:def-discrete-drift} with $\lV=\exp(\pU/2)$.
\end{lemma}


The proof requires several intermediate results.
In the sequel, $\crwm$ is a positive constant which can change from line to line but does not depend on $\step$. We first introduce some notation and state two technical lemmas.
For $M\in\rset^{d \times d}$, denote by $\Vnorm[\text{F}]{M}$ the Frobenius norm of $M$.
For a set $A \subset \rset^d$, define by $A^{\complementaire} = \rset^d \setminus A$.
For all $x \in \rset^{\tilde{d}}$ and $K >0$, we denote by $\bouled{x}{K}{\tilde{d}}$ (respectively $\boulefermeed{x}{K}{\tilde{d}}$), the open (respectively close) ball centered at $x$ of radius $K$. When the dimension $d$ of the state space $\rset^d$ is unambiguous, they are respectively denoted by $\boule{x}{K}$ and $\boulefermee{x}{K}$.
For all $x\in\rset^d$ and $\step>0$, define the acceptance region
\begin{equation}
\label{eq:def-accept-region-rwm}
  \acceptrwm_{x,\step} = \defEns{z\in\rset^d : \alpharwm(x,z) \leq 0} \eqsp.
\end{equation}
For all $x\in\rset^d$ and $\step>0$, define $\Grwm:\rset_+\to\ccint{0,1}$ for all $t\geq 0$ by
\begin{equation}\label{eq:def-Grwm}
  \Grwm(t) = 1/2 + 2\rme^{t^2/2} \cdfc(t) - \rme^{2t^2} \cdfc(2t) \eqsp.
\end{equation}

\begin{lemma}\label{lemma:Grwm}
There exists $t_0>0$ such that for all $t\in\ccint{0,t_0}$, $\Grwm(t) \leq 1 - (t^2 /2)$ and the function $\Grwm$ is non-increasing.
\end{lemma}

\begin{proof}
We have for all $t\geq 0$,
\begin{equation}\label{eq:Grwm-derivative}
  \Grwm'(t) = 2t\rme^{t^2/2} \defEns{\cdfc(t) - 2\rme^{(3t^2)/2} \cdfc(2t)}
\end{equation}
and $G'(0)=0$, $G''(0)=-1$ so there exists $t_0>0$ such that for all $t\in\ccint{0,t_0}$, $\Grwm(t) \leq 1 - (t^2 / 2)$, which is the first statement of the lemma. Regarding the second statement, by an integration by parts, we have for all $s>0$
\begin{equation*}
  \cdfc(s) = \frac{\rme^{-s^2/2}}{\sqrt{2\uppi}s} - \frac{1}{\sqrt{2\uppi}} \int_s^{\plusinfty} \frac{\rme^{-u^2/2}}{u^2} \rmd u
\end{equation*}
and using a change of variables $u=v+t$, we get for all $t>0$
\begin{equation*}
  \cdfc(t) - 2\rme^{(3t^2)/2} \cdfc(2t) = \int_t^{\plusinfty} \defEns{\frac{2\rme^{t(t-v)}}{(v+t)^2}-\frac{1}{v^2}} \frac{\rme^{-v^2/2}}{\sqrt{2\uppi}} \rmd v \eqsp.
\end{equation*}
We now show that $\cdfc(t) - 2\rme^{(3t^2)/2} \cdfc(2t) \leq 0$ for all $t\geq 0$ which will finish the proof using \eqref{eq:Grwm-derivative}. We distinguish the case $t\geq 0.4$ and $t\in\ccint{0,0.4}$. For $t\geq 0.4$, define $\th_t:\coint{t,\plusinfty}\to\rset$ given for all $v\geq t$ by
\begin{equation*}
  \th_t(v) = 2\ln(1+t/v) - \ln(2) - t^2 + vt \eqsp.
\end{equation*}
We show in the sequel that $\th_t(v) \geq 0$ for all $v \geq t \geq 0.4$, which implies $\cdfc(t) - 2\rme^{(3t^2)/2} \cdfc(2t) \leq 0$ for all $t\geq 0.4$. We have for all $v\geq t$
\begin{equation*}
   \th_{t}'(v) = t\defEns{-2/\{v(t+v)\} +1}
\end{equation*}
and $\th_t$ is decreasing on $\ccint{t, \vmin \vee t}$ and increasing on $\coint{\vmin \vee t, \plusinfty}$ where $\vmin=(-t+\sqrt{t^2+8})/2$. Note that $\vmin \geq t$ is equivalent to $t \leq 1$ and for all $t\geq 1$, $\th_t(t) = \ln(2)>0$. Define $\ell:\ocint{0,1}\to\rset$ given for all $t\in\ocint{0,1}$ by
\begin{align*}
  \ell(t) = \th_t(\vmin) &= 2\ln\parenthese{\frac{\sqrt{t^2+8}+t}{\sqrt{t^2+8}-t}} -\ln(2) + (t/2)\parenthese{-3t+\sqrt{t^2+8}} \\
  &= 5\ln(2) - 4\ln\parenthese{-t+\sqrt{t^2+8}} + (t/2)\parenthese{-3t+\sqrt{t^2+8}} \eqsp.
\end{align*}
We have for all $t\in\ocint{0,1}$
\begin{equation*}
  \ell'(t) = -3t+\sqrt{t^2+8} \geq 0 \eqsp,
\end{equation*}
$\ell$ is non-decreasing and $\ell(0.4)>0$, which implies that for all $t\in\ccint{0.4,1}$ and $v\geq t$, $\th_t(v) \geq 0$. Therefore, $\Grwm'(t) \leq 0$ for all $t\geq 0.4$.

For $t\in\ccint{0,0.4}$, we use the following lower and upper bounds by \cite[Theorems 1 and 2]{5963622} for all $s\geq 0$
\begin{equation*}
  \frac{\sqrt{\rme}}{3\sqrt{\uppi}} \rme^{-(3/4)s^2} \leq \cdfc(s) \leq (1/2) \rme^{-s^2/2}
\end{equation*}
and we get for all $t\in\ccint{0,0.4}$
\begin{equation*}
  2\rme^{(3t^2)/2} \cdfc(2t) - \cdfc(t) \geq \rme^{-t^2/2} \defEns{\frac{2\sqrt{\rme}}{3\sqrt{\uppi}}\rme^{-t^2} - \frac{1}{2}} \eqsp.
\end{equation*}
The right hand side is decreasing on $\ccint{0,0.4}$ and positive because \[ (2\sqrt{\rme}\rme^{-(0.4)^2})/(3\sqrt{\uppi}) - (1/2) \geq 0.02 \eqsp, \]
which implies that $\Grwm'(t) \leq 0$ for all $t\in\ccint{0,0.4}$.
\end{proof}

\begin{lemma}\label{lemma:drift-RWM-borned2U}
Assume that $U\in\setpoly{3}(\rset^d,\rset)$ and \Cref{assumption:U-dom-drift-RWM}. Let $x\in\rset^d$, \linebreak $\norm{x} \geq \widetilde{M}$ and $\Krwm>0$. For all $\step>0$ and $z\in\boulefermee{0}{\Krwm}$, we have
\begin{multline*}
  \norm{\DD^2 \pU(x+\sqrt{2\step}z)} \leq \norm{\DD^2 \pU(x)}\defEns{1+\crwm(\Krwm)} \\
  \text{ where } \crwm(\Krwm) = (\crwm \chirwm \Krwm)^{1/2} \step^{1/4} \rme^{\crwm \chirwm \sqrt{\step}\Krwm /2} \eqsp.
\end{multline*}
\end{lemma}

\begin{proof}
Let $z\in\boulefermee{0}{\Krwm}$. Define $\tf:\ccint{0,1} \to \rset^{d\times d}$ by $\tf(t)=\DD^2 \pU(x+t\sqrt{2\step}z)-\DD^2 \pU(x)$ for $t\in\ccint{0,1}$. We have
\begin{equation*}
\frac{\rmd}{\rmd t}\VnormEq[\text{F}]{\tf(t)}^2 = \ps{\tf(t)}{\DD^3 \pU(x+t\sqrt{2\step}z) \cdot \sqrt{2\step} z}_{\text{F}}
\end{equation*}
where for $i,j\in\defEns{1,\ldots,d}$
\begin{equation*}
  \parenthese{\DD^3 \pU(x+t\sqrt{2\step}z) \cdot \sqrt{2\step} z}_{ij} = \sum_{k=1}^{d} \partial_{ijk} \pU(x+t\sqrt{2\step}z) \sqrt{2\step} z_k \eqsp.
\end{equation*}
Using the equivalence of norms in finite dimension and \Cref{assumption:U-dom-drift-RWM}, we get
\begin{align*}
\absolute{\frac{\rmd}{\rmd t}\VnormEq[\text{F}]{\tf(t)}^2} &\leq \crwm \VnormEq[\text{F}]{\tf(t)}\norm{\DD^3 \pU(x+t\sqrt{2\step}z)} \sqrt{2\step} \norm{z} \\
&\leq \crwm \chirwm \parenthese{\VnormEq[\text{F}]{\tf(t)}^2 + \norm[2]{\DD^2 \pU(x)}}\sqrt{\step} \norm{z}
\end{align*}	
which gives by Grönwall's inequality,
\begin{equation*}
\norm[2]{\tf(1)} \leq \norm[2]{\DD^2 \pU(x)} \parenthese{\rme^{\crwm \chirwm \sqrt{\step}\norm{z}} -1} \eqsp.
\end{equation*}
Using $(\rme^{s} -1)^{1/2} \leq \sqrt{s} \rme^{s/2}$ for all $s\geq 0$, we get the result.
\end{proof}

We now proceed to the proof of \Cref{prop:RWM-drift}. Note that we have for all $x\in\rset^d$ and $\step>0$
\begin{multline}\label{eq:RWM-drift-1}
\frac{\Rrwm \lV(x)}{\lV(x)} = \int_{\acceptrwm_{x,\step}} \sqrt{\frac{\invpi(x)}{\invpi(x+\sqrt{2\step}z)}} \frac{\rme^{-\norm[2]{z}/2}}{(2\uppi)^{d/2}} \rmd z \\
+ \int_{(\acceptrwm_{x,\step})^{\complementaire}} \defEns{1+\sqrt{\frac{\invpi(x+\sqrt{2\step}z)}{\invpi(x)}}-\frac{\invpi(x+\sqrt{2\step}z)}{\invpi(x)}} \frac{\rme^{-\norm[2]{z}/2}}{(2\uppi)^{d/2}} \rmd z
\end{multline}
where $\acceptrwm_{x,\step}$ is defined in \eqref{eq:def-accept-region-rwm}.

\paragraph*{Intuition behind the proof}
Before giving the proof of the lemma, we sketch here the analysis of a simple case in one dimension where $\pU(x) = \aU \absolute{x}$ (with a proper regularization near $0$), $\aU>0$ and let $x>0$ be large enough. By \eqref{eq:RWM-drift-1}, we get
\begin{align*}
  \frac{\Rrwm \lV(x)}{\lV(x)} &\approx \int_0^{\plusinfty} \rme^{-\aU\sqrt{\step/2} z} \frac{\rme^{-z^2/2}}{\sqrt{2\uppi}} \rmd z
  \\
  &\phantom{---}+ \int_0^{\plusinfty} \defEns{1+\rme^{-\aU\sqrt{\step/2} z}-\rme^{-\aU\sqrt{2\step} z /2}} \frac{\rme^{-z^2/2}}{\sqrt{2\uppi}} \rmd z \\
  &= (1/2) + 2\rme^{\aU^2 \step /4} \cdfc(\sqrt{\step/2}\aU) - \rme^{\aU^2 \step} \cdfc(\sqrt{2\step}\aU) \\
  &= \Grwm(a\sqrt{\step/2}) \approx 1 - (\step\aU^2)/4 + O(\step^{3/2} \aU^3)
\end{align*}
and the expected contraction in $1-\crwm \step$. The proof below is devoted to make this intuition rigorous and the main steps are a localization argument, a comparison to the one dimensional case and an upper bound on the remainder terms.

\begin{figure}
  \centering
  \begin{tikzpicture}[scale=3]
  \draw [->] (-1.25,0) -- (1.25,0);
  \draw (1.25,0) node[right] {$z_1$};
  \draw [->] (0,-1.25) -- (0,1.25);
  \draw (0,1.25) node[above] {$(z_2,\ldots,z_d)$};
  \draw (0,0) circle (1);
  \draw (60:1) -- (240:1);
  \draw (120:1) -- (300:1);
  \draw (70:1.1) node[above] {$\cone{0}{\thetag}$};
  \draw (150:1.1) node[left] {$\boulefermee{0}{\rayrwm}$};
  \draw (60:0.2) arc (60:90:0.2);
  \draw (70:0.25) node[above] {$\thetag$};
  \draw [dashed] (0,0.5) -| (0.866,0) node[below] {$\brwm(z_{-1})$};
  \draw [dashed] (0.288,0.5) -- (0.288,0) node[below] {$\corwm(z_{-1})$};
  \draw (100:1) .. controls (95:0.75) and (110:0.25) .. (0,0);
  \draw (-70:1) .. controls (-90:0.5) and (-60:0.25) .. (0,0);
  \draw (-70:1) node[below] {$\phirwm(z_{-1})$};
  \end{tikzpicture}
  \caption{\label{fig-cone-brwm-phirwm-corwm} Figure illustrating the definitions of $\cone{0}{\thetag}$, $\brwm(z_{-1})$, $\corwm(z_{-1})$ and $\phirwm(z_{-1})$.}
\end{figure}

In the sequel, let $x\in\rset^d$, $\norm{x} \geq \widetilde{M}$ where $\widetilde{M}$ is given by \Cref{assumption:U-dom-drift-RWM}.

\paragraph*{Step 1: restriction to $\boulefermee{0}{\rayrwm}$}
Define for all $\step>0$
\begin{equation}\label{eq:def-rayrwm}
  \rayrwm = \{8\log((1/\step) \vee 1) + 2d\log(2)\}^{1/2} \eqsp.
\end{equation}
Let $\nZ$ be a standard $d$-dimensional Gaussian vector. By Markov's inequality and \eqref{eq:def-rayrwm}, we have
\begin{equation}\label{eq:RWM-drift-2}
\PP\parenthese{\norm{\nZ} \geq \rayrwm} \leq \rme^{-\rayrwm^2/4} \expe{\rme^{\norm[2]{\nZ}/4}} \leq \exp\parenthese{-\frac{\rayrwm^2}{4} + \frac{d}{2}\log(2)} \leq \step^2 \eqsp.
\end{equation}
Using $\invpi(x)/\invpi(x+\sqrt{2\step}z) \leq 1$ for $z\in\acceptrwm_{x,\step}$, \[ 1+\sqrt{\invpi(x+\sqrt{2\step}z)/\invpi(x)}-\invpi(x+\sqrt{2\step}z)/\invpi(x) \leq 5/4 \] for $z\in(\acceptrwm_{x,\step})^{\complementaire}$, \eqref{eq:RWM-drift-1} and \eqref{eq:RWM-drift-2}, we get
\begin{multline}\label{eq:RWM-drift-1-2}
\frac{\Rrwm \lV(x)}{\lV(x)} \leq (5/4)\step^2 + \int_{\acceptrwm_{x,\step}} \1_{\boulefermee{0}{\rayrwm}}(z) \sqrt{\frac{\invpi(x)}{\invpi(x+\sqrt{2\step}z)}} \frac{\rme^{-\norm[2]{z}/2}}{(2\uppi)^{d/2}} \rmd z \\
+ \int_{(\acceptrwm_{x,\step})^{\complementaire}} \1_{\boulefermee{0}{\rayrwm}}(z) \defEns{1+\sqrt{\frac{\invpi(x+\sqrt{2\step}z)}{\invpi(x)}}-\frac{\invpi(x+\sqrt{2\step}z)}{\invpi(x)}} \frac{\rme^{-\norm[2]{z}/2}}{(2\uppi)^{d/2}} \rmd z \eqsp.
\end{multline}

\paragraph*{Step 2: splitting $\boulefermee{0}{\rayrwm}$ into $\boulefermee{0}{\rayrwm} \cap \acceptrwm_{x,\step}$ and $\boulefermee{0}{\rayrwm} \cap (\acceptrwm_{x,\step})^{\complementaire}$}
In this paragraph, we introduce several geometric quantities illustrated with \Cref{fig-cone-brwm-phirwm-corwm}. Define $\bgamma>0$ by
\begin{equation}\label{eq:def-gambar}
  \max\defEns{(\crwm\chirwm\rayrwm[\bgamma])^{1/2}\bgamma^{1/4} \exp(\crwm\chirwm\bgamma^{1/2}\rayrwm[\bgamma]/2), \eqsp (3/2)\sqrt{2\bgamma}\rayrwm[\bgamma]\chirwm} = 1/2 \eqsp,
\end{equation}
where $\crwm$ is the positive constant given in \Cref{lemma:drift-RWM-borned2U}. Denote by
\begin{equation}\label{eq:def-crwm-1}
  \crwm_1 = (\crwm\chirwm\rayrwm[\bgamma])^{1/2}\bgamma^{1/4} \exp(\crwm\chirwm\bgamma^{1/2}\rayrwm[\bgamma]/2) \in \ccint{0,1/2} \eqsp.
\end{equation}
Let $e_1(x) = \nablaU(x)/\norm{\nablaU(x)}$ and consider the decomposition $z=(z_1, \ldots, z_d)$ of $z$ in an orthonormal basis $(e_1(x), e_2(x),\ldots,e_d(x))$ of $\rset^d$. For all $z\in\rset^d$, denote by $z_{-1}=(z_2,\ldots,z_d)\in\rset^{d-1}$.
For all $\step\in\ocint{0,\bgamma}$, define $\thetag\in\ccint{0,\uppi/4}$ by
\begin{equation}\label{eq:def-tan-thetag}
  \tan \thetag  = 2\sqrt{2\step} \rayrwm \frac{\norm{\DD^2 \pU(x)}}{\norm{\nablaU(x)}} \parenthese{1+\crwm_1} \in\ccint{0,1} \eqsp.
\end{equation}
Denote by
\[ \cone{0}{\thetag} = \defEns{z\in\rset^d : \absolute{z_1} \leq (\tan\thetag) \norm{z_{-1}}} \eqsp. \]
Define $\brwm,\corwm:\boulefermeed{0}{\rayrwm}{d-1} \to \rset_+$ for all $z_{-1}\in\boulefermeed{0}{\rayrwm}{d-1}$ by
\begin{equation}\label{eq:def-brwm-corwm}
  \brwm(z_{-1}) = (\rayrwm^2 - \norm[2]{z_{-1}})^{1/2} \quad \text{and} \quad
  \corwm(z_{-1}) = (\tan\thetag)\norm{z_{-1}} \eqsp.
\end{equation}
By \Cref{lemma:drift-RWM-borned2U} with $\Krwm=\rayrwm$, we have for all $z\in\boulefermee{0}{\rayrwm}$
\begin{equation}\label{eq:RWM-drift2-1}
  \norm{\DD^2 \pU(x+\sqrt{2\step}z)} \leq \norm{\DD^2 \pU(x)}\parenthese{1+\crwm_1} \eqsp.
\end{equation}
where $\crwm_1$ is given in \eqref{eq:def-crwm-1}. By Taylor's theorem, we have for all $z\in\boulefermee{0}{\rayrwm}$
\begin{equation}\label{eq:RWM-drift-taylor-U}
\pU(x+\sqrt{2\step}z) - \pU(x) = \sqrt{2\step} \norm{\nablaU(x)} z_1 + 2 \rrwm(z)
\end{equation}
where $\rrwm:\boulefermee{0}{\rayrwm}\to\rset$ is defined for all $z\in\boulefermee{0}{\rayrwm}$ by
\begin{equation}\label{eq:def-rrwm}
  \rrwm(z) = \step \int_0^1 (1-t) \DD^2 \pU(x+t\sqrt{2\step}z)[z^{\otimes 2}] \rmd t \eqsp.
\end{equation}
By \eqref{eq:def-tan-thetag}, \eqref{eq:RWM-drift2-1} and \eqref{eq:def-rrwm}, we have for all $z\in\boulefermee{0}{\rayrwm} \cap \cone{0}{\thetag}^{\complementaire}$
\begin{align}
  \nonumber
  4 \rrwm(z) &\leq 2\step \rayrwm \norm{\DD^2\pU(x)} \parenthese{1+\crwm_1} \parenthese{\absolute{z_1} + \norm{z_{-1}}} \\
  \label{eq:driftRWM-rest-term-dom}
  &\leq \sqrt{2\step} \norm{\nablaU(x)} (1/2) \tan \thetag \defEns{1+(\tan \thetag)^{-1}} \absolute{z_1} \leq \sqrt{2\step} \norm{\nablaU(x)} \absolute{z_1} \eqsp.
\end{align}
By \eqref{eq:RWM-drift-taylor-U} and \eqref{eq:driftRWM-rest-term-dom}, we obtain for all $z\in\boulefermee{0}{\rayrwm} \cap \cone{0}{\thetag}^{\complementaire}$, $z\neq 0$,
\begin{equation}\label{eq:sign-areas-phirwm}
  \defEns{\pU(x+\sqrt{2\step}z) - \pU(x)}z_1 > 0 \eqsp.
\end{equation}
Moreover, by \Cref{assumption:U-dom-drift-RWM} and \eqref{eq:RWM-drift2-1}, we have for all $z\in\boulefermee{0}{\rayrwm}$
\begin{align*}
\ps{e_1(x)}{\nabla \pU(x+\sqrt{2\step}z)} - \norm{\nabla \pU(x)} &= \sqrt{2\step} \int_0^1 \DD^2 \pU(x+t\sqrt{2\step}z) [z, e_1(x)] \rmd t \eqsp,\\
\absolute{\ps{e_1(x)}{\nablaU(x+\sqrt{2\step}z)} - \norm{\nablaU(x)}} &\leq \sqrt{2\step} (1+\crwm_1) \chirwm \rayrwm \norm{\nablaU(x)}
\end{align*}
and $\ps{e_1(x)}{\nablaU(x+\sqrt{2\step}z)} >0$.
By a version of the implicit function theorem given in \Cref{prop:implicit-function-thm}, there exists $\phirwm : \boulefermeed{0}{\rayrwm}{d-1} \to \rset$ continuous such that for all $\step\in\ocint{0,\bgamma}$,
\begin{multline}\label{eq:def-phirwm}
\defEns{z\in\boulefermee{0}{\rayrwm} : \pU(x+\sqrt{2\step}z) = \pU(x)} \\
 = \defEns{\parenthese{\phirwm(z_{-1}), z_{-1}} : z_{-1} \in\boulefermeed{0}{\rayrwm}{d-1}} \eqsp.
\end{multline}
Combining \eqref{eq:sign-areas-phirwm} and \eqref{eq:def-phirwm}, we obtain for all $\step\in\ocint{0,\bgamma}$,
\begin{align}
  \label{eq:phirwm-prop-1}
  \acceptrwm_{x,\step} \cap \boulefermee{0}{\rayrwm} &= \defEns{z\in\boulefermee{0}{\rayrwm} : z_1 \leq \phirwm(z_{-1})} \eqsp,\\
  \label{eq:phirwm-prop-2}
  (\acceptrwm_{x,\step})^{\complementaire} \cap \boulefermee{0}{\rayrwm} &= \defEns{z\in\boulefermee{0}{\rayrwm} : z_1 \geq \phirwm(z_{-1})} \eqsp,
\end{align}
and for all $z_{-1}\in\boulefermeed{0}{\rayrwm}{d-1}$, $\absolute{\phirwm(z_{-1})} \leq \corwm(z_{-1})$.
These properties and definitions are summarized in \Cref{fig-cone-brwm-phirwm-corwm}.

\paragraph*{Step 3: intermediate upper bound on $\Rrwm \lV(x)/\lV(x)$}
Using \eqref{eq:RWM-drift-1-2} and the definitions of $\brwm$ and $\phirwm$, see \eqref{eq:def-brwm-corwm}, \eqref{eq:def-phirwm}, \eqref{eq:phirwm-prop-1} and \eqref{eq:phirwm-prop-2}, we have
\begin{equation}\label{eq:upper-bound-RrwmV-1}
\frac{\Rrwm \lV(x)}{\lV(x)} \leq (5/4) \step^2 + \int_{z_{-1}\in\boulefermeed{0}{\rayrwm}{d-1}} \tg(z_{-1}) \frac{\rme^{-\norm[2]{z_{-1}}/2}}{(2\uppi)^{(d-1)/2}} \rmd z_{-1}
\end{equation}
where $\tg:\boulefermeed{0}{\rayrwm}{d-1}\to\rset_+$ is defined for all $z_{-1}\in\boulefermeed{0}{\rayrwm}{d-1}$ by
\begin{multline*}
\tg(z_{-1}) = \int_{-b(z_{-1})}^{(\phirwm(z_{-1}) \vee -b(z_{-1})) \wedge b(z_{-1})} \sqrt{\frac{\invpi(x)}{\invpi(x+\sqrt{2\step}z)}} \frac{\rme^{-z_{1}^2/2}}{(2\uppi)^{1/2}} \rmd z_1 \\
+ \int_{(\phirwm(z_{-1}) \vee -b(z_{-1})) \wedge b(z_{-1})}^{b(z_{-1})} \defEns{1+\sqrt{\frac{\invpi(x+\sqrt{2\step}z)}{\invpi(x)}}-\frac{\invpi(x+\sqrt{2\step}z)}{\invpi(x)}} \frac{\rme^{-z_1^2/2}}{(2\uppi)^{1/2}} \rmd z_1 \eqsp.
\end{multline*}
For all $z_{-1}\in\boulefermeed{0}{\rayrwm}{d-1}$, we decompose $\tg(z_{-1})$ in $\tg(z_{-1}) = A_{1}(z_{-1}) + A_{2}(z_{-1})$ where $A_1(z_{-1})$ and $A_2(z_{-1})$ are defined by
\begin{align}
\nonumber
A_{1}(z_{-1}) & = \int_{-b(z_{-1})}^{(\phirwm(z_{-1}) \vee -b(z_{-1})) \wedge 0} \sqrt{\frac{\invpi(x)}{\invpi(x+\sqrt{2\step}z)}} \frac{\rme^{-z_{1}^2/2}}{(2\uppi)^{1/2}} \rmd z_1 \\
\label{eq:def-RWMdrift-A1}
& + \int_{(\phirwm(z_{-1}) \vee -b(z_{-1})) \wedge 0}^{0} \defEns{1+\sqrt{\frac{\invpi(x+\sqrt{2\step}z)}{\invpi(x)}}-\frac{\invpi(x+\sqrt{2\step}z)}{\invpi(x)}} \frac{\rme^{-z_1^2/2}}{(2\uppi)^{1/2}} \rmd z_1 \eqsp, \\
\nonumber
A_{2}(z_{-1}) & = \int_{0}^{(\phirwm(z_{-1}) \vee 0) \wedge b(z_{-1})} \sqrt{\frac{\invpi(x)}{\invpi(x+\sqrt{2\step}z)}} \frac{\rme^{-z_{1}^2/2}}{(2\uppi)^{1/2}} \rmd z_1 \\
\label{eq:def-RWMdrift-A2}
& + \int_{(\phirwm(z_{-1}) \vee 0) \wedge b(z_{-1})}^{b(z_{-1})} \defEns{1+\sqrt{\frac{\invpi(x+\sqrt{2\step}z)}{\invpi(x)}}-\frac{\invpi(x+\sqrt{2\step}z)}{\invpi(x)}} \frac{\rme^{-z_1^2/2}}{(2\uppi)^{1/2}} \rmd z_1 \eqsp.
\end{align}
Combining it with \eqref{eq:upper-bound-RrwmV-1}, we obtain
\begin{equation}\label{eq:upper-bound-RrwmV-1-2}
\frac{\Rrwm \lV(x)}{\lV(x)} \leq (5/4) \step^2 + \int_{z_{-1}\in\boulefermeed{0}{\rayrwm}{d-1}} \defEns{A_{1}(z_{-1}) + A_{2}(z_{-1})} \frac{\rme^{-\norm[2]{z_{-1}}/2}}{(2\uppi)^{(d-1)/2}} \rmd z_{-1} \eqsp.
\end{equation}

By \eqref{eq:RWM-drift-taylor-U} and \eqref{eq:def-RWMdrift-A1}, we have for all $z_{-1}\in\boulefermeed{0}{\rayrwm}{d-1}$
\begin{equation}\label{eq:A1z1}
  A_1(z_{-1}) = A_{11}(z_{-1}) + A_{12}(z_{-1}) + A_{13}(z_{-1}) + A_{14}(z_{-1})
\end{equation}
where
\begin{align*}
  A_{11}(z_{-1}) & =  \int_{-b(z_{-1})}^{0} \rme^{\sqrt{\step/2}  \norm{\nablaU(x)} z_1} \frac{\rme^{-z_{1}^2/2}}{(2\uppi)^{1/2}} \rmd z_1 \eqsp, \\
  A_{12}(z_{-1}) &= \int_{-\brwm(z_{-1})}^{-\brwm(z_{-1}) \vee -\corwm(z_{-1})} \rme^{\sqrt{\step/2} \norm{\nablaU(x)} z_1 + \rrwm(z)} \defEns{1 - \rme^{-\rrwm(z)}} \frac{\rme^{-z_1^2/2}}{(2\uppi)^{1/2}} \rmd z_1 \eqsp,\\
  A_{13}(z_{-1}) &= \int_{-\brwm(z_{-1}) \vee -\corwm(z_{-1})}^{(\phirwm(z_{-1}) \vee -\brwm(z_{-1}))\wedge 0} \rme^{\sqrt{\step/2} \norm{\nablaU(x)} z_1 + \rrwm(z)} \defEns{1 - \rme^{-\rrwm(z)}} \frac{\rme^{-z_1^2/2}}{(2\uppi)^{1/2}} \rmd z_1 \eqsp,\\
  A_{14}(z_{-1}) &= \int_{(\phirwm(z_{-1}) \vee -\brwm(z_{-1})) \wedge 0}^0 \Bigg\{1+\sqrt{\frac{\invpi(x+\sqrt{2\step}z)}{\invpi(x)}}-\frac{\invpi(x+\sqrt{2\step}z)}{\invpi(x)}\\
  &\phantom{------------------}-\rme^{\sqrt{\step/2} \norm{\nablaU(x)} z_1}\Bigg\} \frac{\rme^{-z_1^2/2}}{(2\uppi)^{1/2}} \rmd z_1 \eqsp.
\end{align*}
By \eqref{eq:RWM-drift-taylor-U} and \eqref{eq:def-RWMdrift-A2}, we have for all $z_{-1}\in\boulefermeed{0}{\rayrwm}{d-1}$
\begin{align}
  \nonumber
  A_2(z_{-1}) &= A_{21}(z_{-1}) + A_{22}(z_{-1}) + A_{23}(z_{-1}) + A_{24}(z_{-1}) + A_{25}(z_{-1}) \\
  \nonumber
  &+ \int_{0}^{(\phirwm(z_{-1}) \vee 0) \wedge b(z_{-1})} \Bigg\{\sqrt{\frac{\invpi(x)}{\invpi(x+\sqrt{2\step}z)}}-1-\rme^{-\sqrt{\step/2}  \norm{\nablaU(x)} z_1} \\
  \label{eq:A2z1-temp}
  &\phantom{------------------}+\rme^{-\sqrt{2\step} \norm{\nablaU(x)} z_1} \Bigg\} \frac{\rme^{-z_{1}^2/2}}{(2\uppi)^{1/2}} \rmd z_1
\end{align}
where
\begin{align*}
A_{21}(z_{-1}) & = \int_0^{b(z_{-1})} \defEns{1+\rme^{-\sqrt{\step/2}  \norm{\nablaU(x)} z_1}-\rme^{-\sqrt{2\step} \norm{\nablaU(x)} z_1}} \frac{\rme^{-z_{1}^2/2}}{(2\uppi)^{1/2}} \rmd z_1 \eqsp, \\
A_{22}(z_{-1}) &= \int_{(\phirwm(z_{-1}) \vee 0) \wedge \brwm(z_{-1})}^{\corwm(z_{-1}) \wedge \brwm(z_{-1})} \rme^{-\sqrt{\step/2} \norm{\nablaU(x)} z_1 - \rrwm(z)} \defEns{1 - \rme^{\rrwm(z)}} \frac{\rme^{-z_1^2/2}}{(2\uppi)^{1/2}} \rmd z_1 \eqsp,\\
A_{23}(z_{-1}) &= \int_{(\phirwm(z_{-1}) \vee 0) \wedge \brwm(z_{-1})}^{\corwm(z_{-1}) \wedge \brwm(z_{-1})} \rme^{-\sqrt{2\step}\norm{\nablaU(x)} z_1} \defEns{1 - \rme^{-2\rrwm(z)}} \frac{\rme^{-z_1^2/2}}{(2\uppi)^{1/2}} \rmd z_1 \eqsp, \\
A_{24}(z_{-1}) &= \int_{\corwm(z_{-1}) \wedge \brwm(z_{-1})}^{\brwm(z_{-1})} \rme^{-\sqrt{\step/2} \norm{\nablaU(x)} z_1 - \rrwm(z)} \defEns{1 - \rme^{\rrwm(z)}} \frac{\rme^{-z_1^2/2}}{(2\uppi)^{1/2}} \rmd z_1 \eqsp,\\
A_{25}(z_{-1}) &= \int_{\corwm(z_{-1}) \wedge \brwm(z_{-1})}^{\brwm(z_{-1})} \rme^{-\sqrt{2\step}\norm{\nablaU(x)} z_1} \defEns{1 - \rme^{-2\rrwm(z)}} \frac{\rme^{-z_1^2/2}}{(2\uppi)^{1/2}} \rmd z_1 \eqsp.
\end{align*}
By \eqref{eq:phirwm-prop-1}, $\{\invpi(x)/\invpi(x+\sqrt{2\step}z)\}^{1/2} \leq 1$ for all $z_1\in\ccint{0,\phirwm(z_{-1}) \vee 0}$. Hence, the last term in the right hand side of \eqref{eq:A2z1-temp} is nonpositive and we get
\begin{equation}\label{eq:A2z1}
  A_2(z_{-1}) \leq A_{21}(z_{-1}) + A_{22}(z_{-1}) + A_{23}(z_{-1}) + A_{24}(z_{-1}) + A_{25}(z_{-1}) \eqsp.
\end{equation}
Combining \eqref{eq:A1z1} and \eqref{eq:A2z1}, we obtain for all $z_{-1}\in\boulefermeed{0}{\rayrwm}{d-1}$
\begin{multline}\label{eq:RWM-drift-3}
A_1(z_{-1}) +A_2(z_{-1}) \leq A_{11}(z_{-1}) + A_{21}(z_{-1}) + A_{12}(z_{-1}) + A_{13}(z_{-1}) + A_{14}(z_{-1}) \\
+ A_{22}(z_{-1}) + A_{23}(z_{-1}) + A_{24}(z_{-1}) + A_{25}(z_{-1}) \eqsp.
\end{multline}

\paragraph*{Step 4: upper bound on $A_1(z_{-1}) +A_2(z_{-1})$}
We upper bound each term in the right hand side of \eqref{eq:RWM-drift-3} and we first consider the terms $A_{11} + A_{21}$. Define $\arwm:\ocint{0,\bgamma} \times \rset^d \to \rset_+$ for all $\steptilde\in\ocint{0,\bgamma}$ and $\xtilde\in\rset^d$, $\norm{\xtilde}\geq \widetilde{M}$ by
\begin{equation}\label{eq:def-arwm}
  \arwm(\steptilde,\xtilde) = \sqrt{\steptilde/2} \norm{\nablaU(\xtilde)} \eqsp.
\end{equation}
We have for all $z_{-1}\in\boulefermeed{0}{\rayrwm}{d-1}$,
\begin{equation}\label{eq:up-bound-Jz1}
  A_{11}(z_{-1}) + A_{21}(z_{-1}) \leq \Grwm(\arwm(\step,x))
\end{equation}
where $\Grwm$ is defined in \eqref{eq:def-Grwm}.

We now consider the remainder terms $A_{12}(z_{-1}), A_{13}(z_{-1}), A_{14}(z_{-1}), A_{22}(z_{-1})$, $A_{23}(z_{-1}),A_{24}(z_{-1})$ and $A_{25}(z_{-1})$ in \eqref{eq:RWM-drift-3}.
Let $z_{-1}\in\boulefermeed{0}{\rayrwm}{d-1}$. By definition of $\corwm(z_{-1})$, see \eqref{eq:def-brwm-corwm}, we have for all $z_{1}\in\ccint{-\brwm(z_{-1}), -\corwm(z_{-1}) \vee -\brwm(z_{-1})}$, $z\notin\cone{0}{\thetag}$, and by \eqref{eq:driftRWM-rest-term-dom}
\begin{equation*}
  \sqrt{\step/2} \norm{\nablaU(x)} z_1 + \rrwm(z) \leq (1/2)\sqrt{\step/2} \norm{\nablaU(x)} z_1 \eqsp.
\end{equation*}
Combining it with $1-\rme^{s} \leq \absolute{s}$ for all $s\in\rset$, \eqref{eq:RWM-drift2-1} and \eqref{eq:def-rrwm}, we get
\begin{equation*}
  A_{12}(z_{-1}) \leq \crwm \int_{-\brwm(z_{-1})}^{-\corwm(z_{-1}) \vee -\brwm(z_{-1})} \rme^{(1/2)\sqrt{\step/2} \norm{\nablaU(x)} z_1} \step \norm{\DD^2 \pU(x)} \norm[2]{z} \frac{\rme^{-z_1^2/2}}{(2\uppi)^{1/2}} \rmd z_1 \eqsp.
\end{equation*}
Considering the upper bound $\norm[2]{z} \leq \rayrwm^2$ or the decomposition $\norm[2]{z} = z_1^2 + \norm[2]{z_{-1}}$, we obtain
\begin{equation*}
  A_{12}(z_{-1}) \leq \crwm \step \norm{\DD^2 \pU(x)} \min\defEns{\rayrwm^2 \rme^{\arwm(\step,x)^2/8} \cdfc(a(\step,x)/2), (\norm[2]{z_{-1}}+1)}
\end{equation*}
where $\arwm(\step,x)$ is defined in \eqref{eq:def-arwm}, and using for all $t>0$, $\rme^{t^2/8} \cdfc(t/2) \leq \sqrt{2}/(\sqrt{\uppi}t)$, we get
\begin{equation}\label{eq:up-bound-I1z1}
  A_{12}(z_{-1}) \leq \crwm \min\parenthese{\sqrt{\step} \rayrwm^2 \frac{\norm{\DD^2 \pU(x)}}{\norm{\nablaU(x)}}, (\norm[2]{z_{-1}}+1) \frac{\norm{\DD^2 \pU(x)}}{\norm{\nablaU(x)}^2}\arwm(\step,x)^2} \eqsp.
\end{equation}
Similarly, we have the same upper bound \eqref{eq:up-bound-I1z1} for $A_{24}(z_{-1})$ and $A_{25}(z_{-1})$.

Using for all $s\in\rset$, $1-\rme^{s} \leq \min(1,\absolute{s})$, $\invpi(x)/\invpi(x+\sqrt{2\step}z) \leq 1$ for $z\in\acceptrwm_{x,\step}$, \eqref{eq:def-tan-thetag}, \eqref{eq:def-brwm-corwm}, \eqref{eq:RWM-drift2-1}, \eqref{eq:RWM-drift-taylor-U}, \eqref{eq:def-rrwm} and \eqref{eq:phirwm-prop-1}, we have for all $z_{-1}\in\boulefermeed{0}{\rayrwm}{d-1}$,
\begin{align}
  \nonumber
  A_{13}(z_{-1}) & \leq \int_{-\brwm(z_{-1}) \vee -\corwm(z_{-1})}^{(\phirwm(z_{-1}) \vee -\brwm(z_{-1}))\wedge 0} \min(1,\absolute{\rrwm(z)}) \frac{\rme^{-z_1^2/2}}{(2\uppi)^{1/2}} \rmd z_1 \\
  \nonumber
  & \leq \corwm(z_{-1}) \min(1,\crwm \norm{\DD^2 \pU(x)} \step \rayrwm^2) \\
  \nonumber
  & \leq \crwm \sqrt{\step} \rayrwm^2 \frac{\norm{\DD^2 \pU(x)}}{\norm{\nablaU(x)}} \min(1,\crwm \norm{\DD^2 \pU(x)} \step \rayrwm^2) \\
  \label{eq:up-bound-I2z1}
  & \leq \crwm \min\parenthese{\sqrt{\step} \rayrwm^2 \frac{\norm{\DD^2 \pU(x)}}{\norm{\nablaU(x)}}, \sqrt{\step}\rayrwm^4 \frac{\norm[2]{\DD^2 \pU(x)}}{\norm{\nablaU(x)}^3} \arwm(\step,x)^2} \eqsp.
\end{align}
where $\arwm(\step,x)$ is defined in \eqref{eq:def-arwm}. Similarly, we have the same upper bound \eqref{eq:up-bound-I2z1} for $A_{22}(z_{-1})$ and $A_{23}(z_{-1})$.

Concerning $A_{14}(z_{-1})$, note first that by definition of $\phirwm(z_{-1})$, see \eqref{eq:def-phirwm}, \eqref{eq:phirwm-prop-1}, \eqref{eq:phirwm-prop-2}, and \eqref{eq:RWM-drift-taylor-U}, \eqref{eq:def-rrwm} we have for all $z_1\in\ccint{(\phirwm(z_{-1}) \vee -\brwm(z_{-1}))\wedge 0, 0}$
\begin{equation}\label{eq:rrwm-geq-nablaU}
  2\rrwm(z) \geq \absolute{\sqrt{2\step} \norm{\nablaU(x)} z_1} \eqsp.
\end{equation}
Using $1-\rme^{s} \leq \absolute{s}$ for all $s\in\rset$, $\sqrt{\invpi(x+\sqrt{2\step}z)/\invpi(x)} \leq 1$ for all $z\in(\acceptrwm_{x,\step})^{\complementaire}$, \eqref{eq:RWM-drift2-1}, \eqref{eq:def-rrwm} and \eqref{eq:rrwm-geq-nablaU}, we obtain
\begin{align*}
&\defEns{1-\rme^{\sqrt{\step/2} \norm{\nablaU(x)} z_1}} +\sqrt{\frac{\invpi(x+\sqrt{2\step}z)}{\invpi(x)}} \defEns{1-\rme^{-\sqrt{\step/2}\norm{\nablaU(x)} z_1 -\rrwm(z)}} \\
&\phantom{-------} \leq \min\parenthese{1, \sqrt{\step/2} \norm{\nablaU(x)} \absolute{z_1}} \\
&\phantom{------------}+ \min\parenthese{1, \absolute{\sqrt{\step/2} \norm{\nablaU(x)} z_1 + \rrwm(z)}} \\
&\phantom{-------} \leq \crwm \min\parenthese{1, \step \norm{\DD^2 \pU(x)} \rayrwm^2} \eqsp.
\end{align*}
By \eqref{eq:def-tan-thetag}, \eqref{eq:def-brwm-corwm} and using $\absolute{\phirwm(z_{-1})} \leq \corwm(z_{-1})$, we obtain
\begin{equation}\label{eq:up-bound-I3z1}
A_{14}(z_{-1}) \leq \crwm \min\parenthese{\sqrt{\step} \rayrwm^2 \frac{\norm{\DD^2 \pU(x)}}{\norm{\nablaU(x)}}, \sqrt{\step} \rayrwm^4 \frac{\norm[2]{\DD^2 \pU(x)}}{\norm{\nablaU(x)}^3} \arwm(\step,x)^2}
\end{equation}
where $\arwm(\step,x)$ is defined in \eqref{eq:def-arwm}.

\paragraph*{Step 5: conclusion}
Let $\epsilon=(1/4)\min(1, t_0^2)$ where $t_0$ is defined in \Cref{lemma:Grwm}.
Let $\steptilde>0$ be defined by $\crwm \sqrt{\steptilde}\rayrwm[\steptilde]^2 \chirwm \max\parenthese{1, \rayrwm[\steptilde]^2 \chirwm^2} = \epsilon$ where $\crwm$ is the maximum of the positive constants given in \eqref{eq:up-bound-I1z1}, \eqref{eq:up-bound-I2z1} and \eqref{eq:up-bound-I3z1}.
Define then $\bgamma_1 = \bgamma \wedge \steptilde \wedge t_0^2 \wedge \min(1,\chirwm^{2}/2)/10$ where $\bgamma$ is given in \eqref{eq:def-gambar}.
By \Cref{assumption:U-dom-drift-RWM}, there exists $\overline{M}\geq \widetilde{M}$ such that for all $x\in\rset^d$, $\norm{x}\geq \overline{M}$, $\crwm d \norm{\DD^2 \pU(x)} / \norm[2]{\nablaU(x)} \leq \epsilon$, where $\crwm$ is given in \eqref{eq:up-bound-I1z1}.

By \eqref{eq:up-bound-I1z1}, \eqref{eq:up-bound-I2z1} and \eqref{eq:up-bound-I3z1}, we have for all $x\in\rset^d$, $\norm{x}\geq \overline{M}$ and $\step\in\ocint{0,\bgamma_1}$
\begin{align}
\nonumber
&\int_{z_{-1}\in\boulefermeed{x}{\rayrwm}{d-1}} \{A_{12}(z_{-1}) + A_{13}(z_{-1}) + A_{14}(z_{-1}) + A_{22}(z_{-1}) \\
\nonumber
&\phantom{------}+ A_{23}(z_{-1}) + A_{24}(z_{-1}) + A_{25}(z_{-1})\} \frac{\rme^{-\norm[2]{z_{-1}}/2}}{(2\uppi)^{(d-1)/2}} \rmd z_{-1} \\
\label{eq:drift-remainder-terms}
&\phantom{-----------------}\leq \min(\epsilon, \epsilon \arwm(\step,x)^2)
\end{align}
where $\arwm(\step,x)$ is defined in \eqref{eq:def-arwm}. We consider now two cases:
\begin{itemize}
  \item
  if $\arwm(\step,x)>t_0$, by \eqref{eq:RWM-drift-3}, \eqref{eq:up-bound-Jz1}, \eqref{eq:drift-remainder-terms} and \Cref{lemma:Grwm}, for all $x\in\rset^d$, $\norm{x} \geq \overline{M}$, $\step\in\ocint{0,\bgamma_1}$
  \begin{multline*}
  \int_{z_{-1}\in\boulefermeed{0}{\rayrwm}{d-1}} \defEns{A_1(z_{-1}) + A_2(z_{-1})} \frac{\rme^{-\norm[2]{z_{-1}}/2}}{(2\uppi)^{(d-1)/2}} \rmd z_{-1} \\
  \leq 1 -(t_0^2/2) + \epsilon \leq 1 - (t_0^2/4) \leq 1 - (1/4) \step \eqsp.
  \end{multline*}
  \item
  if $\arwm(\step,x)\in\ocint{0,t_0}$, by \eqref{eq:RWM-drift-3}, \eqref{eq:up-bound-Jz1}, \eqref{eq:drift-remainder-terms}, \Cref{lemma:Grwm} and \Cref{assumption:U-dom-drift-RWM}, for all $x\in\rset^d$, $\norm{x} \geq \overline{M}$, $\step\in\ocint{0,\bgamma_1}$,
  \begin{multline*}
  \int_{z_{-1}\in\boulefermeed{0}{\rayrwm}{d-1}} \defEns{A_1(z_{-1}) + A_2(z_{-1})} \frac{\rme^{-\norm[2]{z_{-1}}/2}}{(2\uppi)^{(d-1)/2}} \rmd z_{-1} \\
  \leq 1 -(1/2-\epsilon)\arwm(\step,x)^2 \leq 1 - \frac{\step\norm[2]{\nablaU(x)}}{8} \leq 1 - \frac{\chirwm^{-2} \step}{8} \eqsp.
  \end{multline*}
\end{itemize}
Combining it with \eqref{eq:upper-bound-RrwmV-1-2}, we obtain for all $x\in\rset^d$, $\norm{x}\geq \overline{M}$, $\step\in\ocint{0,\bgamma_1}$,
\[ \Rrwm \lV(x) / \lV(x) \leq 1 - \min(1,\chirwm^{-2}/2)\step/8  \eqsp. \]
Besides, denote by
\begin{equation*}
  A = \sup_{y,\norm{y}\leq \overline{M}} \defEns{\frac{\generator \lV(y)}{\lV(y)} + \bgamma_1^{1/2} \frac{\genrrwm \lV(y)}{\lV(y)}} \eqsp.
\end{equation*}
By \Cref{prop:RWM-dev-ergo}, we have for all $x\in\rset^d$, $\norm{x} \leq \overline{M}$, $\step\in\ocint{0,\bgamma_1}$, $\Rrwm \lV(x)/\lV(x) \leq 1 + \step A$. We get then for all $x\in\rset^d$, $\step\in\ocint{0,\bgamma_1}$,
\begin{multline*}
  \Rrwm \lV(x) \leq \parenthese{1 - \frac{\min(1,\chirwm^{-2}/2)\step}{8}} \lV(x) \\
  + \step\parenthese{A+\frac{\min(1,\chirwm^{-2}/2)}{8}} \lV(x) \1\defEns{\norm{x} \leq \overline{M}}
\end{multline*}
which concludes the proof.

\subsection*{A version of the implicit function theorem}

The following proposition is taken from \cite[Theorem 7.21]{apostol1969calculus} and \cite[Theorem 6]{Border2013NotesOT}.

\begin{proposition}\label{prop:implicit-function-thm}
Let $\compact$ be a compact metric space and $f:\rset \times \compact \to \rset$ be a continuous function. Assume that there exist $M\geq m >0$ such that for all $z\in\compact$, $x,y\in\rset$, $x\neq y$,
\begin{equation}\label{eq:assum-implicit-func-thm}
  m \leq \frac{f(x,z) - f(y,z)}{x - y} \leq M \eqsp.
\end{equation}
Then, there exists a unique continuous function $\xi:\compact\to\rset$ satisfying for all $z\in\compact$, $f(\xi(z), z) = 0$.
\end{proposition}

\begin{proof}
  Denote by $\Csetfunction(\compact)$ the set of real continuous functions on $\compact$. By standard arguments, $\Csetfunction(\compact)$ is complete under the uniform norm defined for all $g_1,g_2\in\Csetfunction(\compact)$ by $\Vnorm[\infty]{g_1-g_2} = \sup_{z\in\compact} \norm{g_1(z) - g_2(z)}$. Define $\psi:\Csetfunction(\compact)\to\Csetfunction(\compact)$ for all $g\in\Csetfunction(\compact)$ and $z\in\compact$ by
  \begin{equation*}
    \psi(g)(z) = g(z) - (1/M)f(g(z),z) \eqsp.
  \end{equation*}
  By \eqref{eq:assum-implicit-func-thm}, we have for all $g,h\in\Csetfunction(\compact)$ and $z\in\compact$,
  \begin{equation*}
    \absolute{\psi(g)(z) - \psi(h)(z)} \leq \defEns{1-(m/M)}\absolute{g(z) - h(z)}
  \end{equation*}
  and $\Vnorm[\infty]{\psi(g) - \psi(h)} \leq \{1-(m/M)\}\Vnorm[\infty]{g - h}$. $\psi$ is a contraction on $\Csetfunction(\compact)$ and has a unique fixed point $\xi$ in $\Csetfunction(\compact)$ which satisfies $f(\xi(z),z) = 0$ for all $z\in\compact$.
\end{proof}


\section{Additional results for the numerical experiments}
\label{sec:additional-results-numeric}

\subsection{One dimensional example: from theory to practice}
\label{subsec:1-2d-numerics-practice}

We consider the setup of \Cref{subsec:numerical-comparison-toy-examples}. In order to be able to numerically integrate, we truncate the integrals to a finite interval $\ccint{-\bound, \bound}$ for $\bound>0$, \ie~we approximate $\invpi(f)$, $\sPoic'$, $\invpi(\base_i' \base_j')$, $\invpi(\tzf \base_i)$ for $1\leq i,j \leq \pb$ by
\begin{align*}
  &\invpi(f) \approx \int_{-\bound}^{\bound} f(t) \invpi(t) \rmd t \eqsp, \\
  &\sPoic'(x) \approx -(1/\invpi(x))\int_{-\bound}^{x} \invpi(t) \defEns{ f(t) - \int_{-\bound}^{\bound} f(u) \invpi(u) \rmd u} \rmd t \eqsp, \\
  &\invpi(\base_i' \base_j') \approx \int_{-\bound}^{\bound} \invpi(t) \base_i'(t) \base_j'(t) \rmd t \eqsp, \\
  &\invpi(\tzf \base_i) \approx \int_{-\bound}^{\bound} \defEns{ f(t) - \int_{-\bound}^{\bound} f(u) \invpi(u) \rmd u} \base_i(t) \rmd t \eqsp.
\end{align*}
We consider several values for $\bound \in\defEns{3, 4, 5, 6}$ and we expect that when $\int_{-\bound}^{\bound} \invpi(t) \rmd t$ is close to $1$, the truncation is a good approximation of the true quantity.
We are particularly interested in the value of the asymptotic variance of the Langevin diffusion $\varinf(f) = 2\invpi(\sPoic\tzf)$ and the optimal parameters $\paramstar$, $\paramzv$ defined in \eqref{eq:min-asymp-var-diffusion} and \eqref{eq:paramzv}. Approximations of these quantities are reported in \Cref{table:1d-truncation-var-param-zv-cv} for different truncation boundaries $\bound\in\defEns{3,4,5,6}$; concerning $\paramstar$ and $\paramzv$ which are $\pb$-dimensional vectors, we only list their first coordinate, $[\paramstar]_1$ and $[\paramzv]_1$. We observe that truncating the integrals to $\bound = 5$ is sufficient to obtain valid and stable results. It is coherent with the fact that most of the mass of $\invpi$ is contained in this interval, see \Cref{figure:pi_1d}.

\begin{table}
  \centering
  \begin{tabular}{|c|c|c|c|c|}
    \hline
    $\bound$ & $3$ & $4$ & $5$ & $6$ \\
    \hline
    approx. of $\varinf(t)$ & $89.28$  & $92.41$ & $92.45$ & $92.45$ \\
    \hline
    approx. of $[\paramstar]_1$ & $-30.19$ & $-34.37$ & $-34.42$ & $-34.42$ \\
    \hline
    approx. of $[\paramzv]_1$ & $-27.70$ & $-28.57$ & $-28.56$ & $-28.56$ \\
    \hline
  \end{tabular}
  \caption{Approximations of $\varinf(t)$, $[\paramstar]_1$ and $[\paramzv]_1$, function of the truncation boundary $\bound$.}\label{table:1d-truncation-var-param-zv-cv}
\end{table}

\begin{figure}
\begin{center}
\includegraphics[scale=0.4]{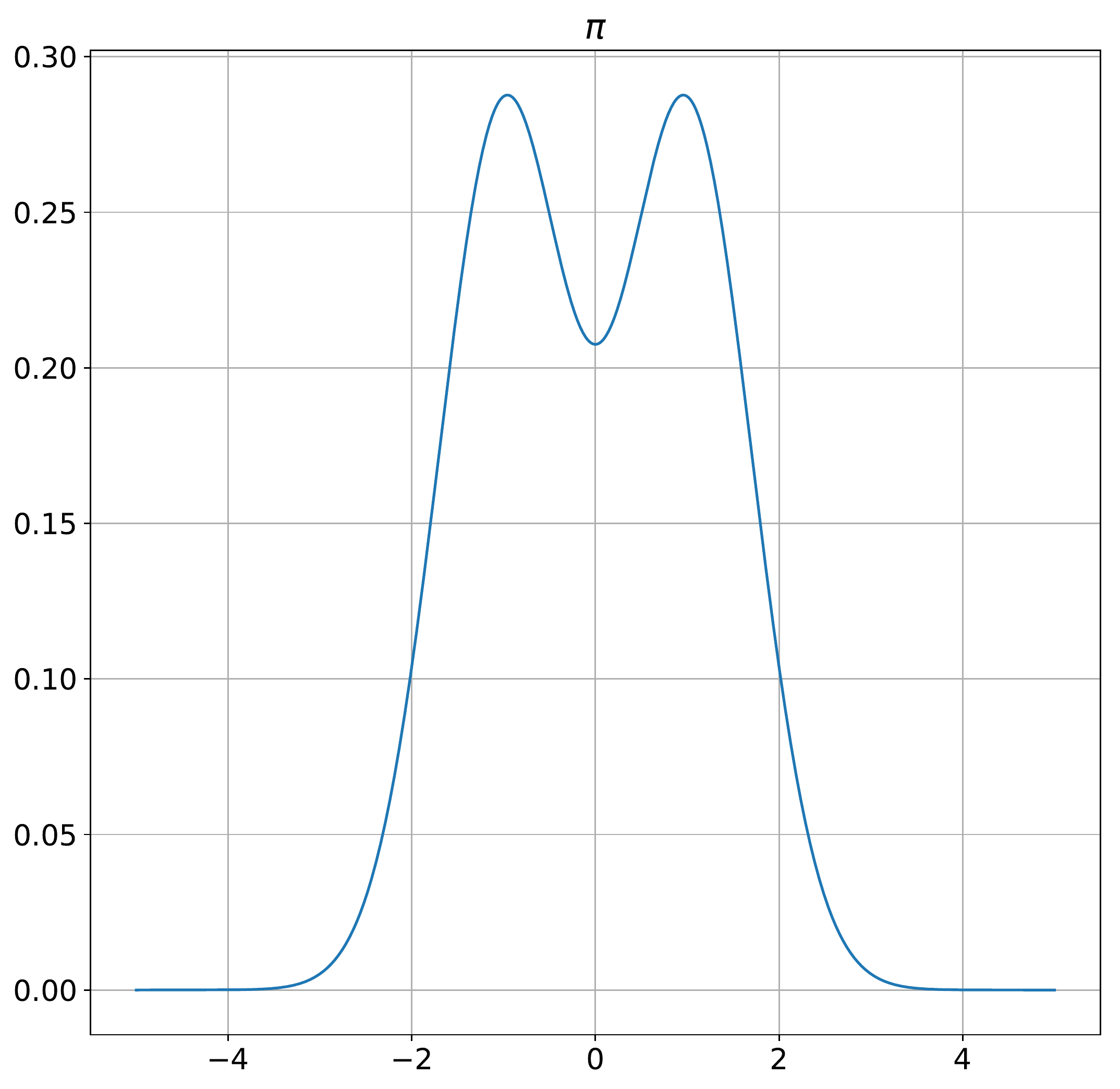}
\end{center}
\caption{\label{figure:pi_1d} Plot of $\invpi$.}
\end{figure}

It is worth to point out that, although the quantities of interest to construct a control variate, \ie~$\varinf(f), \paramstar, \paramzv$, can be accurately estimated by truncating the integrals, others, like $\sPoic'$, highly depend on the truncation boundary $\bound$. We plot in \Cref{figure:dpois_1d} several approximations of $\sPoic'$, by truncating the integrals to $\bound\in\defEns{3,4,5,6}$. Note that by an integration by parts, $\lim_{x\to\pm\infty} \sPoic'(x) / x^2 = C$, with $C>0$. These plots highlight that truncating the integrals has a significant impact on the approximation of $\sPoic'$.

\begin{figure}
\begin{center}
\includegraphics[scale=0.4]{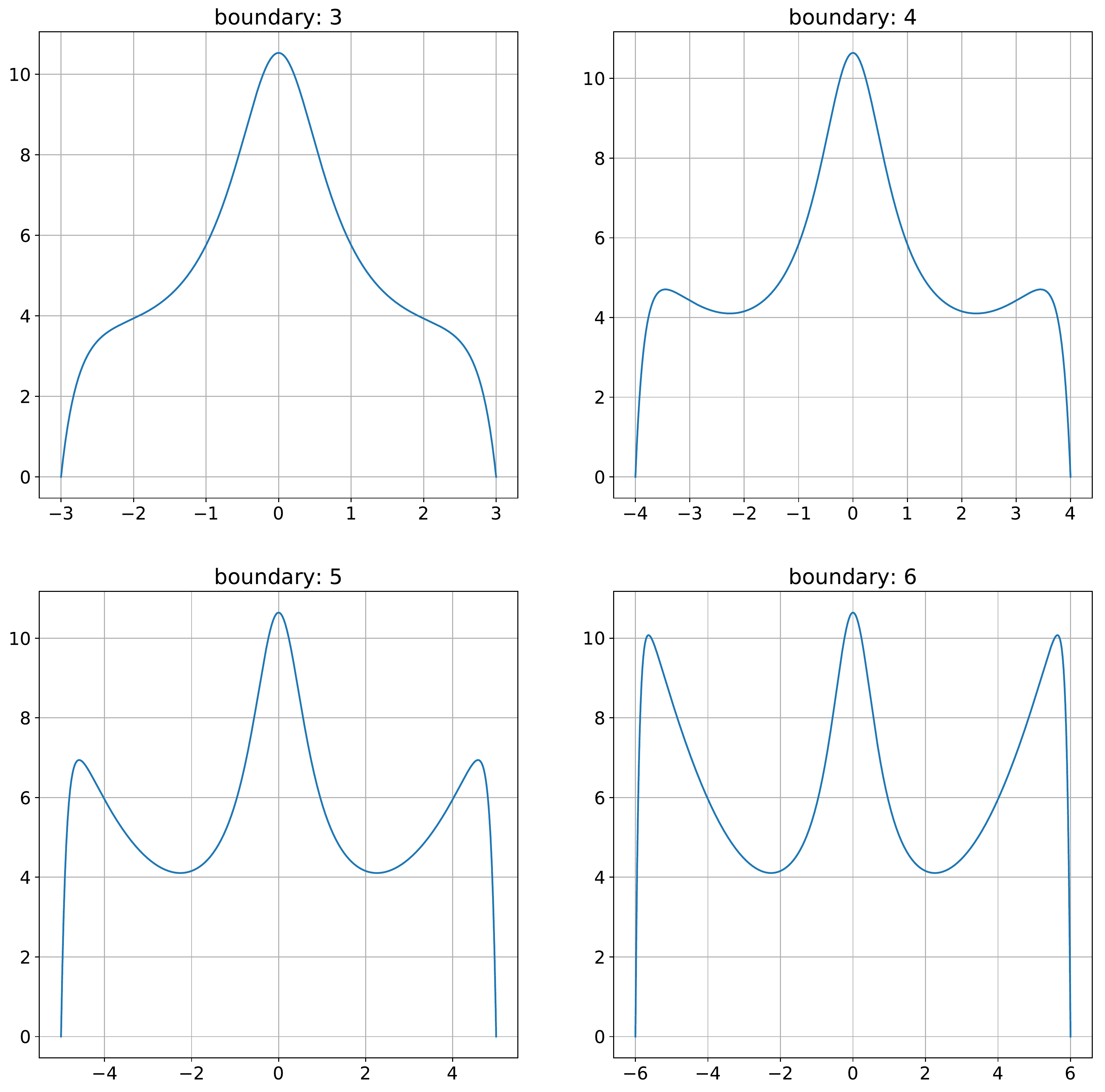}
\end{center}
\caption{\label{figure:dpois_1d} Plots of $\sPoic'$ for $\bound\in\defEns{3,4,5,6}$.}
\end{figure}

In \Cref{figure:approx_dpois_Lpois_1d_1,figure:approx_dpois_Lpois_1d_2}, we plot $\ControlFunc_\param '$ and $\generator \ControlFunc_\param$ for $\param\in\defEns{\paramstar, \paramzv}$ where $\ControlFunc_\param = \ps{\param}{\base}$, $\base = \defEns{\base_i}_1^{\pb}$ are defined in \eqref{eq:def-basis-gaussian-kernels} and $\pb\in\defEns{4,\ldots,10}$. It illustrates that $\sPoic'$ and $\tzf$ are better approximated for even $\pb$; for $\pb\geq 8$, $\ControlFunc_{\paramstar} '$, $\ControlFunc_{\paramzv} '$ and $\generator \ControlFunc_{\paramstar}$, $\generator \ControlFunc_{\paramzv}$ are very close and the two methods obtain similar variance reductions.

\begin{figure}
\begin{center}
\includegraphics[scale=0.4]{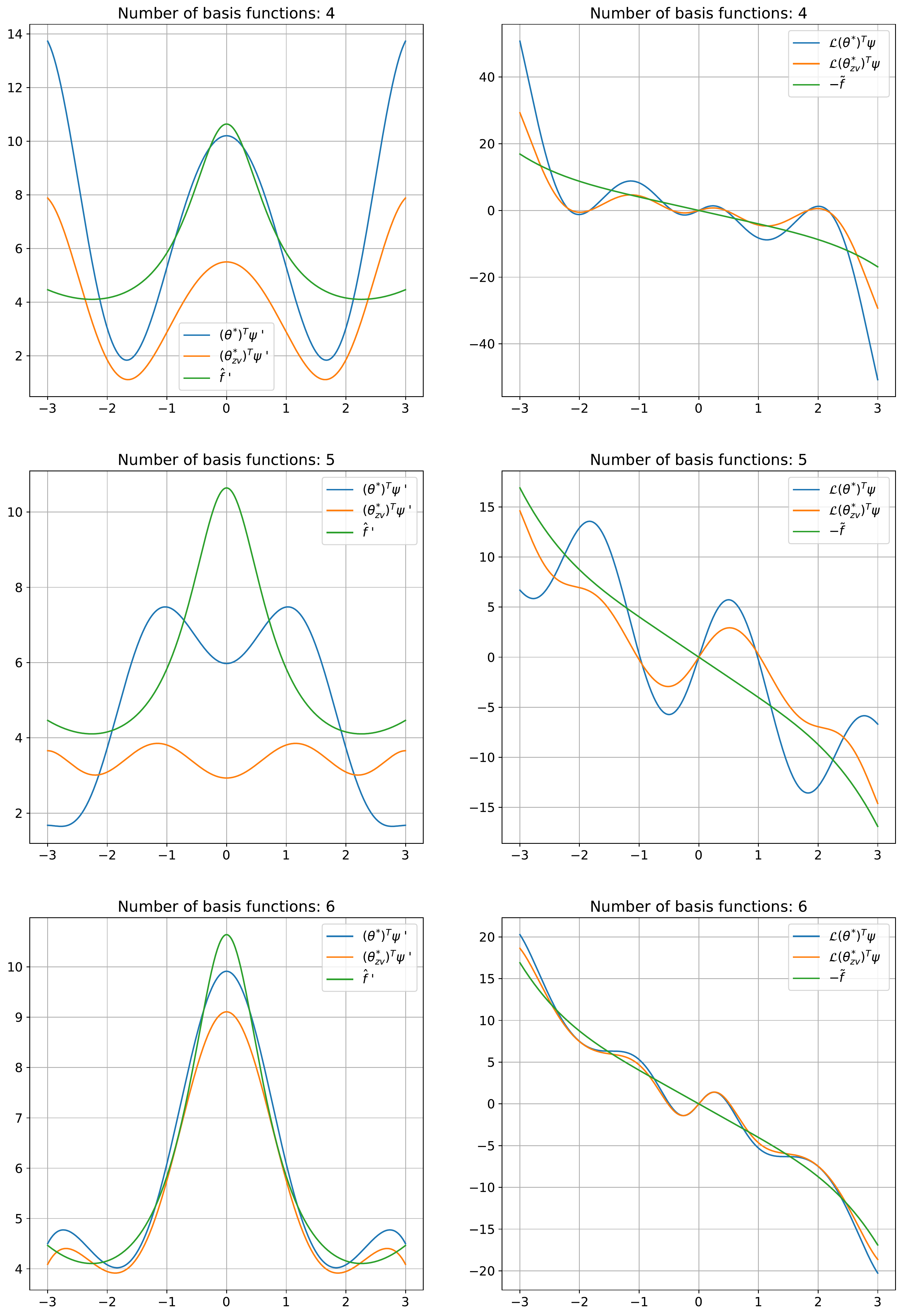}
\end{center}
\caption{\label{figure:approx_dpois_Lpois_1d_1} Plots of $\ControlFunc_\param '$ and $\generator \ControlFunc_\param$ for $\param\in\defEns{\paramstar, \paramzv}$ and $\pb\in\defEns{4,5,6}$.
}
\end{figure}

\begin{figure}
\begin{center}
\includegraphics[scale=0.3]{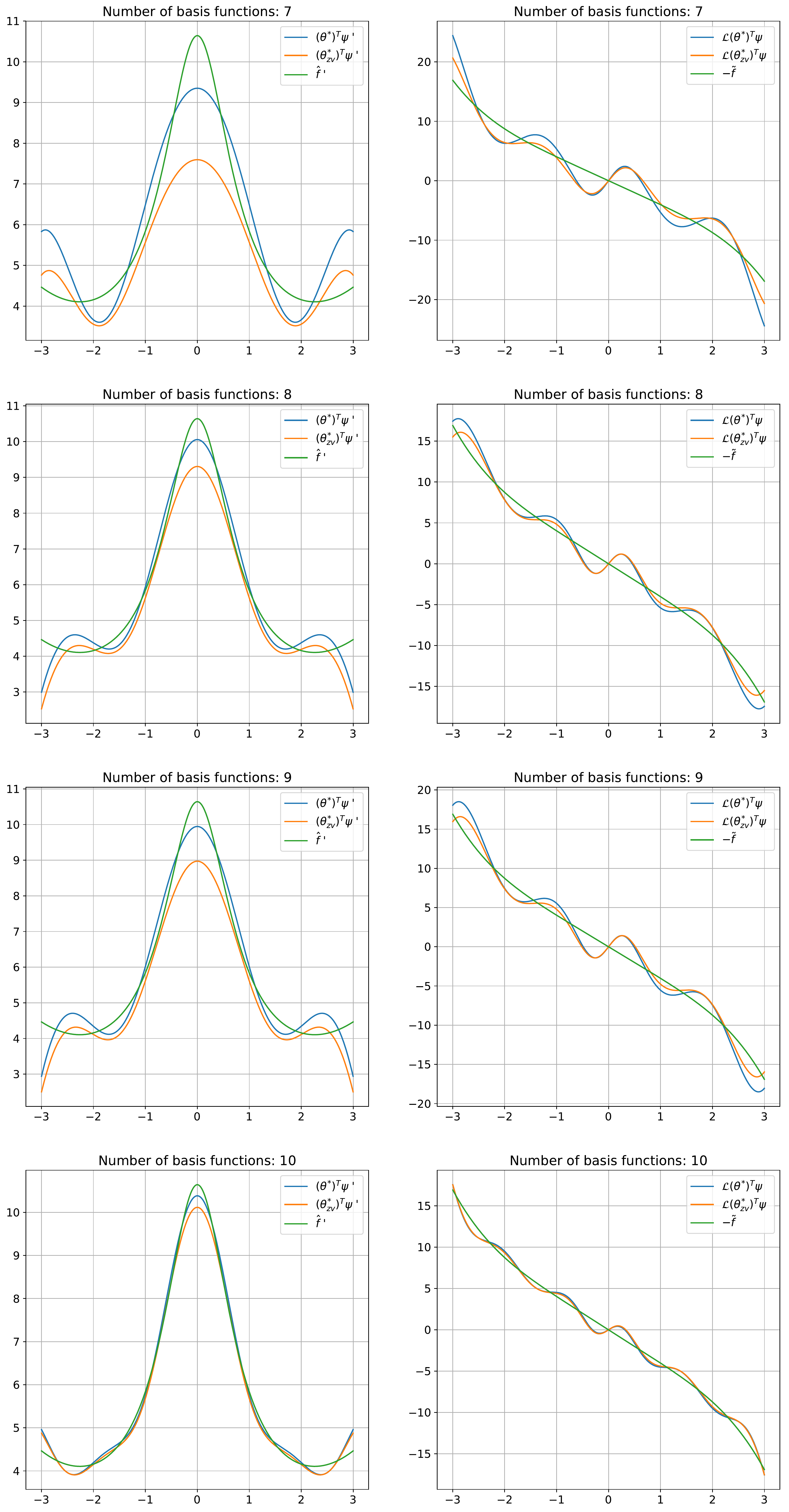}
\end{center}
\caption{\label{figure:approx_dpois_Lpois_1d_2} Plots of $\ControlFunc_\param '$ and $\generator \ControlFunc_\param$ for $\param\in\defEns{\paramstar, \paramzv}$ and $\pb\in\defEns{7,\ldots,10}$.
}
\end{figure}


\subsection{Proof of \Cref{lemma:log-probit-assumptions}}
\label{subsec:proof-log-probit-assumptions}

We have for all $\bb\in\rset^d$
\begin{align*}
  \nabla \Ub{l}(\bb) & = - \xb^{\Tr} \yb + \sum_{i=1}^{\nb} \xb_i / (1+\rme^{-\xb_i^{\Tr} \bb}) + \bb / \varbb \eqsp, \\
  \DD^2 \Ub{l}(\bb) &= \sum_{i=1}^{\nb} \frac{\rme^{-\xb_i^{\Tr} \bb}}{\parenthese{1+\rme^{-\xb_i^{\Tr} \bb}}^2} \xb_i \xb_i^{\Tr} + \Id / \varbb \eqsp, \\
  \DD^3 \Ub{l}(\bb) &= \sum_{i=1}^{\nb} \frac{\rme^{-\xb_i^{\Tr} \bb}}{\parenthese{1+\rme^{-\xb_i^{\Tr} \bb}}^2} \defEns{2\frac{\rme^{-\xb_i^{\Tr} \bb}}{1+\rme^{-\xb_i^{\Tr} \bb}}-1} \xb_i^{\otimes 3} \eqsp.
\end{align*}
Using for all $i\in\defEns{1,\ldots,\nb}$ and $\bb\in\rset^d$ that $0<\rme^{-\xb_i^{\Tr} \bb}/(1+\rme^{-\xb_i^{\Tr} \bb})^2 \leq 1/4$, $\Ub{l}$ is strongly convex, gradient Lipschitz and satisfies \Cref{assumption:U-Sinfty}, \eqref{eq:cond-vgeom-ula}, \Cref{ass:condition_MALA} and \Cref{assumption:U-dom-drift-RWM}.

For $\Ub{p}$, define $\th:\rset\to\rset_{-}$ for all $t\in\rset$ by $\th(t) = \ln(\Phi(t))$. We have for all $t\in\rset$,
\begin{align*}
  & \th'(t) = \frac{\Phi'(t)}{\Phi(t)} \quad,\quad \th''(t) = -\frac{\Phi'(t)}{\Phi(t)}\defEns{t + \frac{\Phi'(t)}{\Phi(t)}} \eqsp,\\
  & \th^{(3)}(t) = \frac{\Phi'(t)}{\Phi(t)}\defEns{2\parenthese{\frac{\Phi'(t)}{\Phi(t)}}^2 + 3t\frac{\Phi'(t)}{\Phi(t)} + t^2 -1}
\end{align*}
and for all $\bb\in\rset^d$
\begin{align*}
  \nabla \Ub{p}(\bb) & = \sum_{i=1}^{\nb} \defEns{(1-\yb_i) \th'(-\xb_i^{\Tr}\bb)-\yb_i \th'(\xb_i^{\Tr}\bb)} \xb_i + \bb/\varbb \eqsp, \\
  \DD^2 \Ub{p}(\bb) &= \sum_{i=1}^{\nb} \defEns{-(1-\yb_i) \th''(-\xb_i^{\Tr}\bb)-\yb_i \th''(\xb_i^{\Tr}\bb)} \xb_i \xb_i^{\Tr} + \Id/\varbb \eqsp, \\
  \DD^3 \Ub{p}(\bb) &= \sum_{i=1}^{\nb} \defEns{(1-\yb_i) \th^{(3)}(-\xb_i^{\Tr}\bb)-\yb_i \th^{(3)}(\xb_i^{\Tr}\bb)} \xb_i^{\otimes 3} \eqsp.
\end{align*}
By an integration by parts, we have for all $t<0$
\begin{equation*}
  t + \frac{\Phi'(t)}{\Phi(t)} = -\frac{t}{\Phi(t)} \int_{-\infty}^{t}\frac{\rme^{-s^2/2}}{\sqrt{2\uppi}s^2} \rmd s
\end{equation*}
and $t+\Phi'(t)/\Phi(t) \geq 0$ for all $t\in\rset$. Let $t<0$ and $s=-t>0$. We have $\Phi(t) = \cdfc(s) = \erfc(s/\sqrt{2})/2$ where $\erfc:\rset\to\rset_+$ is the complementary error function defined for all $u\in\rset$ by $\erfc(u) = (2/\sqrt{\uppi})\int_{u}^{\plusinfty} \rme^{-v^2} \rmd v$. By \cite[Section 8.25, formula 8.254]{gradshteyn2014table}, we have the following asymptotic expansion for $s\to\plusinfty$
\begin{equation*}
  \cdfc(s) = \frac{\rme^{-s^2/2}}{\sqrt{2\uppi}s}\parenthese{1-s^{-2} + 3s^{-4} + O(s^{-6})} \eqsp.
\end{equation*}
Using that $\Phi'(t) = (2\uppi)^{-1/2} \rme^{-t^2/2}$ for all $t\in\rset$, we get asymptotically for $t\to-\infty$ and $s=-t\to\plusinfty$,
\begin{equation}\label{eq:asymptotic-expansion-Phi}
  \Phi'(t) / \Phi(t) = s \parenthese{1+s^{-2} - 2s^{-4} + O(s^{-6})}
\end{equation}
and $\lim_{t\to -\infty} \th''(t) = -1$. There exists then $C>0$ such that for all $t\in\rset$, $-C \leq h''(t) \leq 0$. $\Ub{p}$ is then strongly convex, gradient Lipschitz and satisfies \Cref{assumption:U-Sinfty} and \eqref{eq:cond-vgeom-ula}.
By \eqref{eq:asymptotic-expansion-Phi}, we have for $t\to-\infty$ and $s=-t\to\plusinfty$, $\th^{(3)}(t) = O(s^{-1})$. $\Ub{p}$ satisfies then \Cref{ass:condition_MALA} and \Cref{assumption:U-dom-drift-RWM}.

\subsection{Additional results for the Bayesian logistic and probit regressions}
\label{sec:suppl-probit-reg}

We first define the basis of functions $\basea$ and $\baseb$ based on first and second order polynomials respectively. Let $\basea = (\basea_1,\ldots, \basea_d)$ be given for $i\in\defEns{1,\ldots,d}$ and $\bb=(\bb_1,\ldots,\bb_d)\in\rset^d$ by $\basea_i(\bb) = \bb_i$ and $\baseb = (\baseb_1,\ldots,\baseb_{d(d+3)/2})$ be given for $\bb=(\bb_1,\ldots,\bb_d)\in\rset^d$ by
\begin{align*}
  \baseb_k(\bb) & = \bb_k \quad \text{for} \eqsp k\in\defEns{1,\ldots,d}  \eqsp, \quad
  \baseb_{k+d}(\bb) = \bb_k^2 \quad \text{for} \eqsp k\in\defEns{1,\ldots,d}  \eqsp,\\
  \baseb_k(\bb) & = \bb_i \bb_j \quad \text{for} \eqsp k=2d+(j-1)(d-j/2)+(i-j) \eqsp \text{and all} \eqsp 1 \leq j < i \leq d \eqsp.
\end{align*}
$\basea$ and $\baseb$ are in $\setpoly{\infty}(\rset^d,\rset)$ and are linearly independent in $\mrc(\rset^d,\rset)$.

We provide additional plots for the logistic regression, see \Cref{figure:log-1-add} and \Cref{figure:log-2}, and the results for the Bayesian probit regression presented in \Cref{sec:application_cv}, see Table~\ref{table:probit}, \Cref{figure:pro-1} and \Cref{figure:pro-2}. They are similar to the results obtained for the Bayesian logistic regression.

\begin{figure}
\hspace*{-1.5cm}\includegraphics[scale=0.5]{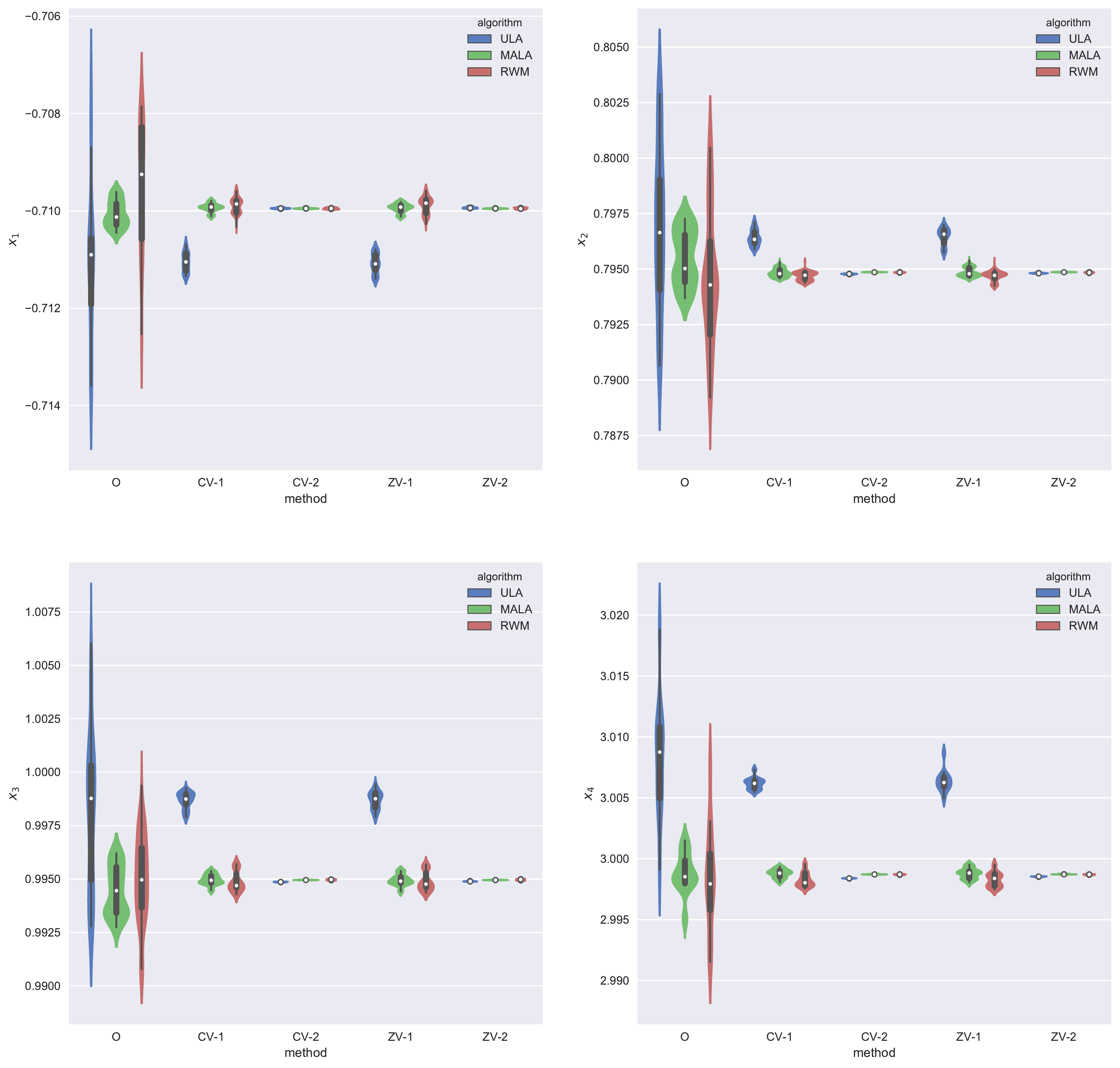}
\caption{\label{figure:log-1-add} Boxplots of $\bb_1,\bb_2,\bb_3,\bb_4$ using the ULA, MALA and RWM algorithms for the logistic regression. The compared estimators are the ordinary empirical average (O), our estimator with a control variate \eqref{eq:def-invpi-cv} using first (CV-1) or second (CV-2) order polynomials for $\base$, and the zero-variance estimator of \cite{papamarkou2014} using a first (ZV-1) or second (ZV-2) order polynomial basis. }
\end{figure}

\begin{figure}
\hspace*{-1.5cm}\includegraphics[scale=0.5]{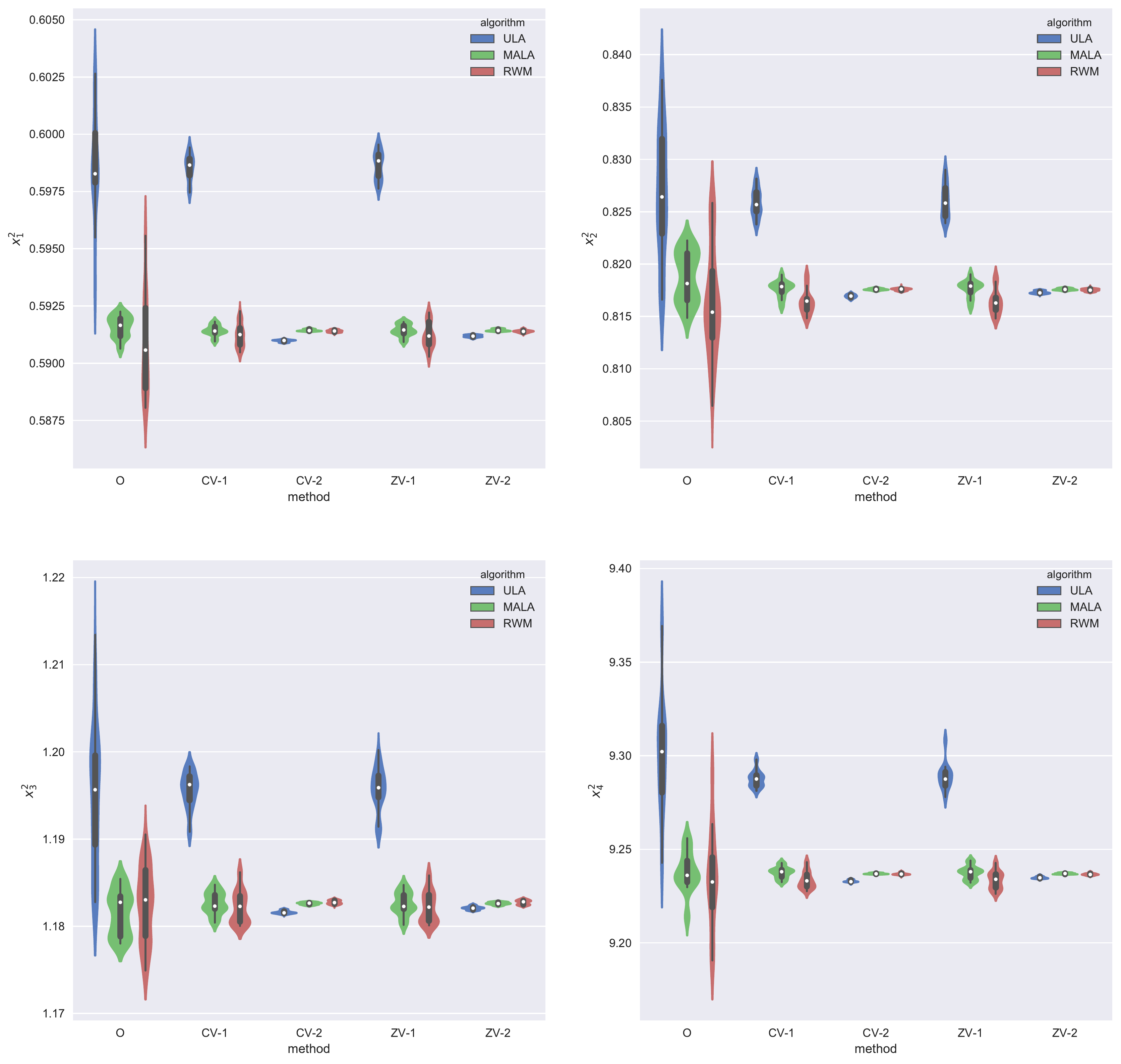}
\caption{\label{figure:log-2} Boxplots of $\bb_1^2,\bb_2^2,\bb_3^2,\bb_4^2$ using the ULA, MALA and RWM algorithms for the logistic regression. The compared estimators are the ordinary empirical average (O), our estimator with a control variate \eqref{eq:def-invpi-cv} using first (CV-1) or second (CV-2) order polynomials for $\base$, and the zero-variance estimator of \cite{papamarkou2014} using a first (ZV-1) or second (ZV-2) order polynomial basis.}
\end{figure}

\begin{figure}
\hspace*{-1.5cm}\includegraphics[scale=0.5]{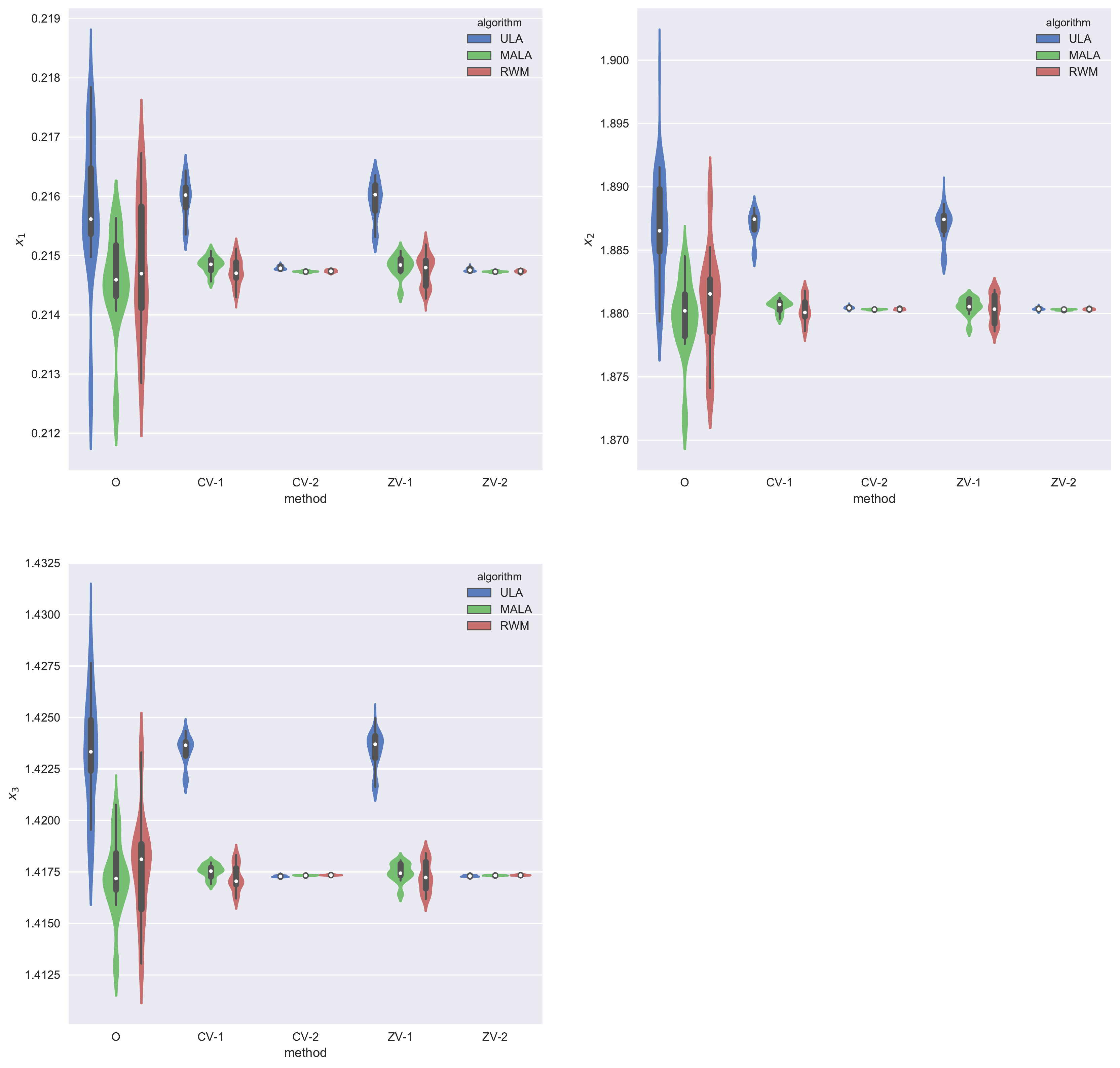}
\caption{\label{figure:pro-1} Boxplots of $\bb_1,\bb_2,\bb_3$ using the ULA, MALA and RWM algorithms for the probit regression. The compared estimators are the ordinary empirical average (O), our estimator with a control variate \eqref{eq:def-invpi-cv} using first (CV-1) or second (CV-2) order polynomials for $\base$, and the zero-variance estimator of \cite{papamarkou2014} using a first (ZV-1) or second (ZV-2) order polynomial basis. }
\end{figure}

\begin{figure}
\hspace*{-1.5cm}\includegraphics[scale=0.5]{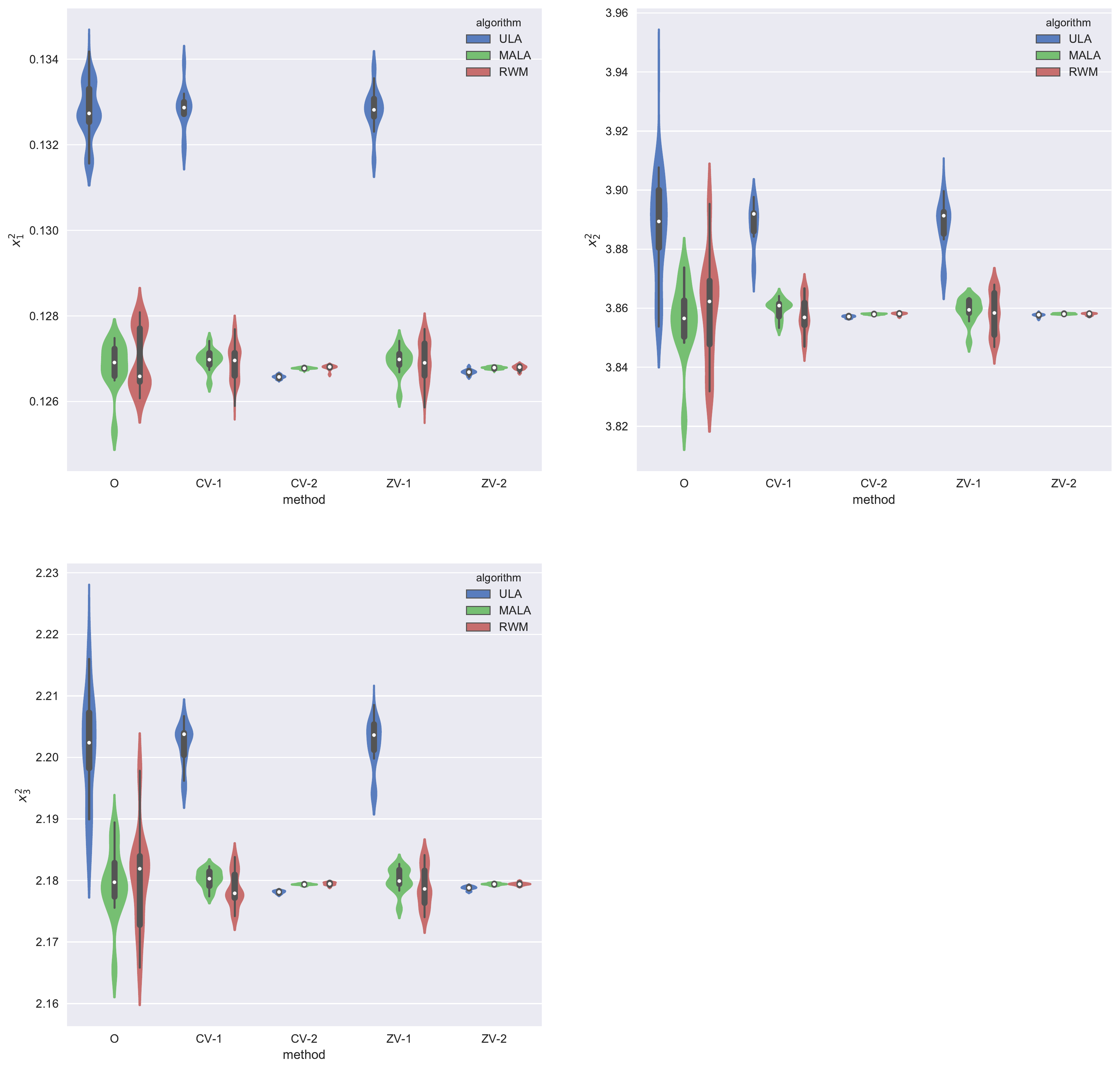}
\caption{\label{figure:pro-2} Boxplots of $\bb_1^2,\bb_2^2,\bb_3^2$ using the ULA, MALA and RWM algorithms for the probit regression. The compared estimators are the ordinary empirical average (O), our estimator with a control variate \eqref{eq:def-invpi-cv} using first (CV-1) or second (CV-2) order polynomials for $\base$, and the zero-variance estimator of \cite{papamarkou2014} using a first (ZV-1) or second (ZV-2) order polynomial basis.}
\end{figure}

\begin{table}
{\small
\begin{tabular}{c|c|c|c c|c c|c c|c c|}
   \multicolumn{11}{c}{} \\
   & & MCMC & \multicolumn{2}{c|}{CV-1-MCMC} & \multicolumn{2}{c|}{CV-2-MCMC}
   & \multicolumn{2}{c|}{ZV-1-MCMC} & \multicolumn{2}{c|}{ZV-2-MCMC} \\
   & & Variance & VRF & Variance & VRF & Variance & VRF & Variance & VRF & Variance \\
   \hline
   $\bb_1$ & ULA &       2.1 &         24 &      0.089 &    2.9e+03 &    0.00073 &         20 &       0.11 &    2.7e+03 &    0.00078 \\
   &MALA &      0.41 &         22 &      0.019 &    2.7e+03 &    0.00015 &         18 &      0.023 &    2.6e+03 &    0.00016 \\
   &RWM &       1.2 &         23 &       0.05 &    2.2e+03 &    0.00054 &         21 &      0.056 &    2.2e+03 &    0.00053 \\
   \hline
   $\bb_2$ & ULA &        27 &         24 &        1.1 &    2.8e+03 &     0.0099 &         18 &        1.5 &    2.4e+03 &      0.011 \\
   &MALA &       6.4 &         24 &       0.27 &    2.9e+03 &     0.0022 &         19 &       0.34 &    2.6e+03 &     0.0025 \\
   &RWM &        13 &         18 &       0.72 &    1.8e+03 &     0.0073 &         16 &       0.81 &    1.8e+03 &     0.0075 \\
   \hline
   $\bb_3$ & ULA &        11 &         24 &       0.47 &    6.7e+03 &     0.0017 &         18 &       0.62 &    6.3e+03 &     0.0018 \\
   &MALA &       2.6 &         23 &       0.11 &      7e+03 &    0.00037 &         18 &       0.14 &    6.8e+03 &    0.00038 \\
   &RWM &       5.5 &         18 &        0.3 &    4.3e+03 &     0.0013 &         16 &       0.34 &    4.3e+03 &     0.0013 \\
   \hline
   $\bb_1^2$ & ULA &      0.75 &        3.5 &       0.22 &    1.6e+02 &     0.0048 &        2.8 &       0.26 &    1.3e+02 &     0.0057 \\
   &MALA &      0.15 &        3.5 &      0.043 &    1.5e+02 &      0.001 &        2.8 &      0.053 &    1.3e+02 &     0.0011 \\
   &RWM &      0.43 &        2.6 &       0.16 &    1.2e+02 &     0.0035 &        2.4 &       0.18 &    1.2e+02 &     0.0037 \\
   \hline
   $\bb_2^2$ &ULA &   4.7e+02 &        9.3 &         51 &    1.4e+03 &       0.33 &        7.5 &         63 &    1.2e+03 &        0.4 \\
   &MALA &   1.1e+02 &        9.1 &         12 &    1.5e+03 &      0.073 &        7.6 &         14 &    1.3e+03 &      0.085 \\
   &RWM &   2.2e+02 &        7.7 &         29 &      1e+03 &       0.22 &        6.9 &         33 &    9.8e+02 &       0.23 \\
   \hline
   $\bb_3^2$ & ULA &   1.1e+02 &        9.8 &         11 &    9.7e+02 &       0.11 &        7.9 &         14 &    7.9e+02 &       0.14 \\
   &MALA &        24 &        9.7 &        2.5 &    9.8e+02 &      0.025 &        8.1 &          3 &    8.5e+02 &      0.029 \\
   &RWM &        52 &        7.9 &        6.7 &    6.1e+02 &      0.086 &        7.1 &        7.4 &    5.9e+02 &      0.088 \\
 \hline
\end{tabular}
}
\caption{\label{table:probit}Estimates of the asymptotic variances for ULA, MALA and RWM and each parameter $\bb_i$, $\bb_i^2$ for $i\in\{1,\ldots,d\}$, and of the variance reduction factor (VRF) on the example of the probit regression.}
\end{table}

\end{document}